\documentclass{amsart}
\usepackage[margin=1in]{geometry}
\usepackage{pifont}
\usepackage{eucal}
\usepackage{setspace}
\usepackage[english]{babel}
\usepackage[utf8]{inputenc}
\usepackage[T1]{fontenc}

\usepackage[dvipsnames]{xcolor}
\usepackage{paralist}
\usepackage{graphicx}
\usepackage{subcaption}
\usepackage{longtable} 
\usepackage{multirow}
\usepackage{listings}
\usepackage{makecell}
\usepackage{array}
\usepackage{float}
\usepackage{dsfont}
\usepackage{rotating}
\usepackage{booktabs}
\usepackage{enumerate}
\usepackage{tikz}
\usepackage{pgf}
\usepackage{enumitem}
\usepackage{lipsum} % for generating filler text
\usepackage{wrapfig}

\usepackage{amsmath,amssymb,amsfonts,amsthm,bbm}
\usepackage{mathtools}
\usepackage{mathrsfs}
\usepackage{mathabx}

\mathtoolsset{showonlyrefs}
\usepackage{natbib}
\usepackage{todonotes}
\usepackage{xr-hyper}
\usepackage[
  colorlinks,
  citecolor=blue,
  linkcolor=red,
  anchorcolor=red,
  urlcolor=blue
]{hyperref}

\usepackage{kantlipsum}

\theoremstyle{plain}
\newtheorem{theorem}{Theorem}[section]
\newtheorem{proposition}[theorem]{Proposition}
\newtheorem{lemma}[theorem]{Lemma}
\newtheorem{example}[theorem]{Example}
\newtheorem{corollary}[theorem]{Corollary}
\theoremstyle{definition}
\newtheorem{definition}[theorem]{Definition}
\newtheorem{assumption}[theorem]{Assumption}
\theoremstyle{remark}
\newtheorem{remark}[theorem]{Remark}

\usepackage[dvipsnames]{xcolor}
\usepackage{paralist}
\usepackage{graphicx}
\usepackage{subcaption}
\usepackage{longtable} 
\usepackage{multirow}
\usepackage{listings}
\usepackage{makecell}
\usepackage{array}
\usepackage{float}
\usepackage{dsfont}
\usepackage{rotating}
\usepackage{booktabs}
\usepackage{enumerate}
\usepackage{tikz}
\usepackage{pgf}
\usepackage{enumitem}
\usepackage{lipsum} % for generating filler text
\usepackage{mathtools}
\usepackage{mathrsfs}
\mathtoolsset{showonlyrefs}
\usepackage{natbib}
\usepackage{todonotes}
\usepackage{xr-hyper}

\newcommand{\R}{\mathbb R}
\newcommand{\EE}{\mathbb{E}}
\newcommand{\PP}{\mathbb{P}}

\DeclareMathOperator*{\argmin}{argmin}
\DeclareMathOperator*{\argmax}{argmax}

\newcommand{\bs}[1]{\ensuremath{\boldsymbol{#1}}}
\newcommand{\mc}{\mathcal}
\newcommand{\opt}{^\star}
\newcommand{\Belief}{\mathcal{B}}
\newcommand{\BR}{\Phi^{\rm BR}}
\newcommand{\veps}{\varepsilon}

\newcommand{\diff}{\textnormal{d}}

\definecolor{myblue}{rgb}{.8, .8, 1}
\definecolor{mathblue}{rgb}{0.2472, 0.24, 0.6} % mathematica's Color[1, 1--3]
\definecolor{mathred}{rgb}{0.6, 0.24, 0.442893}
\definecolor{mathyellow}{rgb}{0.6, 0.547014, 0.24}

\usepackage{pifont}
\definecolor{bblye}{RGB}{60, 88 148}

\usepackage{tikz}
\newcount\Comments
\newcommand{\kibitz}[2]{\ifnum\Comments=1{\textcolor{#1}{\textsf{\footnotesize #2}}}\fi}

\definecolor{portlandorange}{rgb}{1.0, 0.35, 0.21}
\definecolor{aurometalsaurus}{rgb}{0.43, 0.5, 0.5}
\definecolor{blue(munsell)}{rgb}{0.0, 0.5, 0.69}

\usepackage{hyperref}
\usepackage{algorithm}
\usepackage{algpseudocode}

\hypersetup{colorlinks, breaklinks,
linkcolor={[RGB]{174,83,139}},
citecolor={[RGB]{60, 88 148}},
urlcolor={[rgb]{0.6, 0.13, 0.12}}}

\setcitestyle{authoryear}

\title{Statistical Equilibrium of Optimistic Beliefs} 

\author[Gui]{Yu Gui}\thanks{University of Pennsylvania, The Wharton School Department of Statistics and Data Science, \url{yugui@wharton.upenn.edu}}

\author[Taşkesen]{Bahar Taşkesen}\thanks{University of Chicago, Booth School of Business, \url{bahar.taskesen@chicagobooth.edu}}

\begin{document}

\addtocontents{toc}{\protect\setcounter{tocdepth}{-1}}

\begin{abstract}
We study finite normal-form games in which payoffs are subject to random perturbations and players face uncertainty about how these shocks co-move across actions, an ambiguity that naturally arises when only realized (not counterfactual) payoffs are observed. We introduce the Statistical Equilibrium of Optimistic Beliefs (SE-OB), inspired by discrete choice theory.
We model players as \textit{optimistic better responders}: they face ambiguity about the dependence structure (copula) of payoff perturbations across actions and resolve this ambiguity by selecting, from a belief set, the joint distribution that maximizes the expected value of the best perturbed payoff. Given this optimistic belief, players choose actions according to the induced random-utility choice rule. We define SE-OB as a fixed point of this two-step response mapping.
SE-OB generalizes the Nash equilibrium and the structural quantal response equilibrium. We establish existence under standard regularity conditions on belief sets. For the economically important class of marginal belief sets, that is, the set of all joint distributions with fixed action-wise marginals, optimistic belief selection reduces to an optimal coupling problem, and SE-OB admits a characterization via Nash equilibrium of a smooth regularized game, yielding tractability and enabling computation.
We characterize the relationship between SE-OB and existing equilibrium notions and illustrate its empirical relevance in simulations, where it captures systematic violations of independence of irrelevant alternatives that standard logit-based models fail to explain.
\end{abstract}
\maketitle

\section{Introduction}

We study equilibrium in finite normal-form games when players face
\textit{dependence ambiguity} about payoff perturbations and resolve that ambiguity in an
\textit{optimistic} way. In many strategic environments, agents can form reasonably accurate
\textit{baseline} payoff predictions, yet the mapping from actions to realized payoffs remains noisy
and only partially understood. For example, while the marginal distribution of payoff noise for
each action may be estimable from past data, the \textit{joint} dependence of payoff shocks across
actions is typically not identified because counterfactual payoffs are not observed. We model this
as ambiguity over the joint distribution (copula) of payoff perturbations, holding the marginals
fixed.

Nash equilibrium provides a canonical benchmark for strategic interaction, but two considerations
motivate departures from the Nash paradigm. First, computing Nash equilibria in general finite
games is PPAD-complete \citep{daskalakis2009complexity, chen2009settling}. Second, in laboratory
and field settings, observed behavior often departs systematically from Nash predictions, consistent
with bounded rationality and/or misspecified beliefs about the payoff environment
\citep{lieberman1960human,brayer1964experimental,o1987nonmetric,goeree2001ten}.

A leading behavioral extension of Nash equilibrium is Quantal Response Equilibrium (QRE) \citep{mckelvey1995quantal},
which embeds a random-utility (payoff-perturbation) model into strategic interaction and defines
equilibrium as a fixed point of the resulting stochastic choice rules. QRE induces smooth
best-response mappings and is often empirically successful, but common \textit{logit} specifications
(e.g.,  i.i.d.\ extreme value shocks) impose strong restrictions on the joint distribution of payoff shocks, including Luce's independence of irrelevant alternatives (IIA), and can be sensitive to
misspecification.

We introduce a new equilibrium concept, the \textit{Statistical Equilibrium of Optimistic Beliefs
(SE-OB)}. In SE-OB, each player is endowed with a \textit{belief set} of plausible joint
distributions over payoff perturbations. Given the opponents' mixed strategies, the player selects
from this set the distribution that maximizes the expected \textit{maximum} (\textit{i.e.}, ``best'')
perturbed payoff across actions (an optimistic resolution of dependence ambiguity). The player
then plays the induced random-utility choice rule: each action is chosen with the probability that
it is optimal under the selected distribution. An SE-OB is a strategy profile in which each player's
mixed strategy is consistent with this two-step procedure, given the opponents' strategies. SE-OB
nests Nash equilibrium and standard structural QRE as special cases when belief sets collapse to
singletons.

A key feature of SE-OB is tractability for economically natural ambiguity sets. For the important
class of \textit{marginal belief sets} that fix the action-wise marginals
but leave dependence unrestricted, SE-OB induces a tractable regularized best-response mapping on
the simplex, yielding a smooth fixed-point problem amenable to computation and comparative
statics.

% \delbt{\subsection{A motivating example: pricing under forecasting ambiguity}
% Consider two competing retailers who repeatedly choose prices from a finite menu. Each firm can
% estimate the marginal distribution of profit forecast errors at each candidate price using historical
% residuals. However, it is difficult to identify how forecast errors \textit{co-move} across prices
% because the firm observes profit only for the price it posts; profits at unchosen prices are
% counterfactual in the same period. A natural belief set therefore fixes the marginals of forecast
% errors while leaving the dependence structure (the copula) unrestricted. In SE-OB, each firm
% resolves this dependence ambiguity optimistically, that is, it acts as if the unknown dependence is the
% one that makes the \textit{best} price option look most attractive in expectation, and then
% randomizes according to the probability each price is optimal under that selected dependence.
% Section~\ref{sec:SE-OB} formalizes this mechanism and shows that, for marginal belief sets, it
% induces a tractable regularized best-response mapping on the simplex.}

Consider the example of two competing retailers who repeatedly choose a price from a finite menu.
In each period, a retailer forms baseline profit forecasts for all candidate prices using historical data,
but it observes realized profit only for the price it actually posts. As a result, the retailer can learn the uncertainty associated with each individual price in isolation, that is, the marginal distribution of forecasting errors at each price. However, it cannot learn how these errors co-move across prices within the same market state, because profits at unchosen prices are counterfactual.
This distinction matters because dependence across prices encodes economically relevant substitution patterns. Some market states may cause one price point to perform exceptionally well while others underperform, and prices may be correlated through shared demand or cost shocks. Classical Nash equilibrium has no explicit role for such payoff perturbations, while standard logit-based quantal response models typically impose a restrictive dependence structure, such as independent shocks and independence of irrelevant alternatives. These assumptions can obscure precisely the cross-price correlations that are central to many pricing environments.

SE-OB, when specified through marginal belief sets, is designed precisely for this environment. It allows the retailer to retain what can be credibly inferred from data, namely the action-wise marginal distributions of payoff uncertainty, while remaining agnostic about what cannot be learned, namely the dependence structure across prices. The defining behavioral ingredient of SE-OB is optimism in resolving this dependence ambiguity.

Intuitively, when joint co-movement is unknown, the retailer acts as if uncertainty aligns in a way that
makes it likely that some price in the menu will look especially attractive. This form of optimism is not about inflating average payoffs, but about anticipating the possibility of a favorable realization among available alternatives. As Keynes wrote in \citep[Chapter 12]{Keynes:1936},
\begin{samepage}
\begin{quote}\it
...a large proportion of our positive activities depend on spontaneous optimism rather than on a
mathematical expectation, whether moral or hedonistic or economic.
\end{quote}
\end{samepage}
Under SE-OB, this optimism translates into a smooth randomized pricing rule: prices that are more
likely to emerge as the best option under favorable co-movement receive higher choice probabilities.

Our main contributions are as follows:
\begin{enumerate}[label=(\roman*), leftmargin=*]
\item We introduce the {Statistical Equilibrium of Optimistic Beliefs (SE-OB)} for finite normal-form games
with payoff perturbations whose {dependence structure} across actions is ambiguous.
Each player is endowed with a belief set over joint distributions of payoff shocks and
selects, given opponents' mixed strategies, the joint distribution that maximizes the expected value of the
{best} perturbed payoff across actions (optimistic belief selection). The player then plays the induced
random-utility (quantal) choice rule. SE-OB is defined as a fixed point of this two-step response mapping
and nests Nash equilibrium and structural QRE as special cases when belief sets collapse to singletons.
\item We establish the existence of SE-OB under convexity, compactness and moment conditions on belief
sets. Our formulation makes the induced stochastic choice well-defined for arbitrary joint laws of payoff
perturbations, thereby
providing a unified foundation encompassing the usual absolutely-continuous QRE specifications.
\item For the economically relevant class of {marginal belief sets} that fix the
action-wise marginals while leaving dependence unrestricted, we show that optimistic belief selection
reduces to an optimal coupling problem and that SE-OB coincides with Nash equilibrium of an associated
\textit{smooth regularized game} whose regularizer is pinned down by marginal quantile functions.
This characterization yields a tractable smooth fixed-point problem, enables computation and comparative
statics, and recovers familiar response maps (\textit{e.g.}, softmax/logit) as well as sparse choice rules (\textit{e.g.},
sparsemax) as special cases.
\item 
Beyond tractability, marginal belief sets arise as identified sets under realized-payoff feedback.
In a bandit environment, counterfactual payoffs are unobserved, so the copula of payoff perturbations is not
identified from the induced statistical experiment, even with arbitrary exploration policies, while the
one-dimensional marginals are point identified. Hence the sharp nonparametric identified set of joint laws is
the Fr\'echet class of all couplings consistent with the identified marginals, i.e., the marginal belief set.
Since the inclusive value $\EE[\max_i(u_i+\xi_i)]$ is sensitive to dependence, it is only partially identified over
this class; SE-OB corresponds to evaluating it at the upper envelope (the optimistic endpoint).

\item We analyze the empirical content of SE-OB. With free belief set specification sets, SE-OB inherits the classical
non-falsifiability of unrestricted QRE; in contrast, restricting attention to marginal belief sets restores falsifiability and yields testable restrictions on behavior.
We further illustrate empirically that SE-OB can capture systematic violations of Luce's IIA, including
``clone'' effects and sparse support patterns, which standard logit-based models fail to explain. 
\end{enumerate}

\noindent\textbf{Paper organization.}
The rest of the paper unfolds as follows.
Section~\ref{sec:lit-review} discusses related literature.
Section~\ref{sec:problem-setup} introduces the game-theoretic setup and recalls Nash equilibrium and QRE.
Section~\ref{sec:SE-OB} defines SE-OB, establishes existence, and derives the smooth-game characterization
for marginal belief sets.
Section~\ref{sec:mbs-foundations} provides informational foundations for marginal belief
sets via an identification analysis under realized-payoff feedback. 
Section~\ref{sec:falsifiability} studies empirical content and falsifiability of SE-OB.
Section~\ref{sec:empirical} presents empirical illustrations in simulations.
Section~\ref{sec:discuss} concludes and discusses limitations and extensions.
Technical proofs and additional computational details including a convergent numerical routine for SE-OB with marginal-belief-set case are deferred to the appendix.

\section{Literature Review}
\label{sec:lit-review}

Extending the classical Nash equilibrium framework, a large literature incorporates economic and psychological insights to account for systematic deviations from perfect rationality observed in laboratory and field behavior. In the Nash paradigm, players best respond to correct beliefs about others’ actions; empirically, however, play often departs from these predictions in stable and interpretable ways.

In two-player zero-sum games, for example, a series of studies \citep{lieberman1960human,brayer1964experimental,o1987nonmetric,brown1990testing,rapoport1992mixed} tests minimax theory. As noted by \citep{o1987nonmetric}, while aggregate frequencies can track minimax predictions, substantial heterogeneity across subjects suggests that many individuals deviate from minimax play-plausibly reflecting limits in attention and information processing. Relatedly, \citep{nagel1995unraveling} shows that players’ reasoning can be shaped by observed histories, and \citep{goeree2001ten} documents that changes in payoff magnitudes and structures can generate systematic departures from standard equilibrium predictions. These patterns motivated the development of behavioral game theory models that enrich the payoff–choice mapping using, among other ingredients, social preferences, bounded strategic thinking, and subjective beliefs \citep{camerer1997progress,gintis2005behavioral}.

A central stochastic-choice approach is the quantal response equilibrium (QRE) of \citep{mckelvey1995quantal}, which defines equilibrium as a fixed point of noisy (quantal) best responses. QRE has been applied successfully across many environments, but its empirical performance depends on correct statistical specification of how players deviate from perfect optimization, and can be sensitive to misspecification. Moreover, standard implementations often impose homogeneity in noise or responsiveness, motivating extensions that incorporate heterogeneity in choice behavior and beliefs \citep{stahl1994experimental,camerer2004cognitive,friedman2005random}. Recognizing the restrictiveness of perfect payoff expectations and optimization, \citep{goeree2004model} introduced noisy introspection, and \citep{rogers2009heterogeneous} extended QRE by allowing individual differences in payoff responsiveness.

A complementary line of work seeks foundations that place weaker structure on the payoff–choice relationship. \citep{goeree2021m} proposes minimal behavioral assumptions, including only a monotonic relationship between expected payoffs and choice frequencies, and an unbiasedness condition linking observed choices to underlying choice distributions. Building on these assumptions, $M$-equilibrium \citep{goeree2021m} extends Nash equilibrium by allowing set-valued predictions of choices and beliefs. This parameter-free, detail-free approach unifies equilibria derived from strategy perturbation \citep{selten1988reexamination} and payoff perturbation \citep{harsanyi1973games} and encompasses all quantal response equilibria. At the same time, the resulting profile sets may be overly permissive in some applications \citep{goeree2023s}, motivating equilibrium concepts that retain robustness while delivering more informative behavioral predictions.

Another research avenue extends Nash equilibrium through a Bayesian lens to settings with incomplete information about payoffs. Harsanyi’s framework for Bayesian games \citep{harsanyi1967games} assumes a well-specified payoff distribution. Subsequent work introduced ex post equilibrium \citep{cremer1985optimal,holmstrom1983efficient}, a distribution-free concept whose existence is not always guaranteed, and robust Nash equilibrium \citep{aghassi2006robust}, which evaluates payoffs in a worst-case manner over an uncertainty set. More recently, Bayesian games have been studied under distributional ambiguity via distributionally robust equilibria \citep{liu2024bayesian,liu2018distributionally,loizou2016distributionally,loizou2015distributionally}.

SE-OB connects these strands by endogenizing a player’s belief over payoff perturbations within an ambiguity set and allowing that ambiguity to be resolved optimistically. In particular, the principle that players select the perturbation distribution that maximizes their best-case payoff places SE-OB toward the optimistic end of the Hurwicz criterion \citep{hurwicz1951generalized}. Unlike a pure ``maximax'' rule that simply selects the best conceivable outcome, SE-OB selects a belief over a distributional ambiguity set, so the degree of optimism is governed by the chosen distribution. This perspective aligns with models of optimally biased beliefs \citep{brunnermeier2005optimal}, affective decision making \citep{bracha2012affective}, and the broader literature on motivated beliefs and overconfidence \citep{camerer1997progress}, which emphasize that agents may systematically overweight favorable states when forming beliefs.

\paragraph{Notation}We write $ \mathbb N_+ = \{1,2,3,\ldots\} $ for the set of positive integers. For any $ N \in \mathbb N_+ $, let $ [N] = \{1,\ldots,N\} $. For any $ M \in \mathbb N_+ $, let $ \boldsymbol y_i \in \mathbb R^N $ for all $ i \in [M] $, and denote their tuple by $
\boldsymbol Y = (\boldsymbol y_1,\ldots,\boldsymbol y_M).$
For $ i \in [M] $, let $ \boldsymbol Y_{-i} $ denote the tuple $
(\boldsymbol y_1,\ldots,\boldsymbol y_{i-1},\boldsymbol y_{i+1},\ldots,\boldsymbol y_M).$
Given $ \boldsymbol y_i' \in \mathbb R^N $, the tuple $ (\boldsymbol y_i',\boldsymbol Y_{-i}) $ denotes $
(\boldsymbol y_1,\ldots,\boldsymbol y_{i-1},\boldsymbol y_i',\boldsymbol y_{i+1},\ldots,\boldsymbol y_M).$
For any $ N \in \mathbb N_+ $, define the probability simplex $
\Delta_N = \{ x \in \mathbb R_+^N : \sum_{j=1}^N x_j = 1 \}.$
For each $ i \in [N] $, let $ \boldsymbol e_i \in \Delta_N $ denote the $ i $th standard basis vector. 
For vectors $\bs x \in \R^N$, the norm $\|\bs x\|$ denotes the Euclidean norm, and $\|\bs x \|_\infty = \max_{i \in [N]} x_i$ denotes the infinity norm.
For matrices $\bs A \in \R^{N\times N}$, $\|\bs A\|_{\mathrm{op}}$ denotes the operator norm.

\section{Problem Setup}
\label{sec:problem-setup}
Throughout this paper, we study finite normal form games with $M \in \mathbb N_+$ players and $N \in \mathbb N_+$ actions. Each player $j \in [M]$ chooses an action $a_j \in [N]$. Given an action profile $\bs a = (a_1,\ldots,a_M) \in [N]^M$, player $j$ receives payoff $u_j^a(\bs a)$, where $
u_j^a : [N]^M \to [0,1]$ 
is the payoff function of player $j$. The game is specified by the tuple of payoff functions $(u_1^a,\ldots,u_M^a)$~\citep{nash1950equilibrium}.
A mixed strategy $\bs p_j \in \Delta_N$ of player $j$ is a probability distribution over $[N]$. A mixed strategy profile is $
\bs P = (\bs p_1,\ldots,\bs p_M) \in (\Delta_N)^M.$
Under $\bs P$, actions are drawn independently as $a_j \sim \bs p_j$ for $j \in [M]$, and the expected payoff of player $j$ is $
u_j(\bs P)= \EE_{a_1 \sim \bs p_1,\ldots,a_M \sim \bs p_M}[u_j^a(a_1,\ldots,a_M)].$
For a given profile $\bs P$ and an alternative mixed strategy $\bs p_j' \in \Delta_N$, we write $
u_j(\bs p_j'; \bs P_{-j}) = \EE_{a_1 \sim \bs p_1,\ldots,a_j \sim \bs p_j',\ldots,a_M \sim \bs p_M}
[u_j^a(a_1,\ldots,a_M)].$
We denote the mixed extension of the finite game by $
\mathcal{G}([M], \Delta_N, (u_j)_{j \in [M]}).$

\subsection{Equilibria of Fixed Beliefs}\label{sec:qre}
\label{sec:equilibrium-of-fixed-beliefs}
The central concept of game theory is the Nash equilibrium~(NE), and we present its formal definition below for an $M$-player, $N$-action normal-form game with expected payoffs given by~$(u_j)_{j\in[M]}$.
\begin{definition}[Nash equilibrium]
A strategy profile $\bs P^* = (\bs p_1^*, \ldots, \bs p_M^*) \in (\Delta_N)^M$ is a Nash equilibrium of the game~$\mc G([M], \Delta_N, (u_j)_{j\in[M]})$ if, for each $j \in [M]$,
\begin{align}
\label{eq:NE}
    \bs p_j^* \in \argmax\limits_{\bs p \in \Delta_N } u_{j}(\bs p;\bs P^*_{-j})
    % \leq u_j(\bs p_j^*; \bs P^*_{-j}).
    \tag{\text{NE}}
\end{align}
\label{def:ne}
\end{definition}
A Nash equilibrium in non-cooperative games is grounded in three fundamental principles that collectively eliminate errors in the players' decisions and beliefs: 
\begin{itemize}[leftmargin=5mm]
\item[\ding{101}] \textit{Beliefs}: Each player's behavior is shaped by their expectations about how other players will act. 
\item[\ding{101}] \textit{Optimal Responses}: Given these expectations, every player's choice is the best possible response.
\item[\ding{101}] \textit{Perfect Foresight}: Players' expectations of others' behaviors are accurate in a probabilistic sense.
\end{itemize}

When these three conditions are met, they create an internally consistent model of behavior. This consistency leads to a stable equilibrium distribution of action profiles, where no player has an incentive to deviate unilaterally. However, human decision-making is far from flawless. Unlike the precise predictions of traditional game theory, real-world behavior is often messy, unpredictable, and prone to error. Indeed, the foundational assumptions of standard game theory frequently fail to align with data from laboratory experiments, revealing systematic and significant deviations that challenge its core principles \citep{lieberman1960human,brayer1964experimental,o1987nonmetric,brown1990testing,rapoport1992mixed}.
In response to these empirical discrepancies, the Quantal Response Equilibrium (QRE) framework was developed~\citep{mckelvey1995quantal} using standard statistical models of quantal choice in a game-theoretic setting. 

QRE relaxes the core assumption of traditional NE that individual choice behavior is always optimal by allowing players to make mistakes in their decisions. In QRE, the expected payoffs derived from players' expectations may be influenced by ``noise'', which can result in suboptimal decisions. In this framework, even if the first and third principles of the game still hold, players must account for the fact that others do not always optimize perfectly. This introduction of noise fundamentally alters the dynamics of behavior within the game, and thus its equilibrium. 
In (structural) QRE inspired by the additive random payoff models~\citep{ref:mcfadden1976quantal}, each player's expected payoff for each action is perturbed with some random noise $\bs \xi = (\xi_i)_{i \in [N]} \sim \mu$. Specifically, the disturbed expected payoff of player~$j$, according to some distribution $\mu_j \in \mc P(\R^N)$, takes the form $u_j(\bs e_i; \bs P_{-j}) + \xi_i$.
The assumed choice behavior is that each player chooses action~$i$, when $u_j(\bs e_i; \bs P_{-j}) + \xi_i \geq u_j(\bs e_{i'}; \bs P_{-j}) + \xi_{i'}$ for all $i' \in [N]$.{
This behavior of choice induces a distribution over each player's possible actions. 
}
{When the maximizer is not unique, the player may randomize among the maximizing actions. To formalize this, for $\bs u \in \R^N$ and $\bs \xi \in \R^N$, define the set  
\[A(\bs u, \bs \xi) = \argmax\limits_{i \in [N]} \;\left\{u_i + \xi_i \right\}\subseteq [N].\]
Note that $A(\bs u, \bs \xi)$ is non-empty, since the maximum of a finite set is attained.
Next, we define the associated set of optimal mixed strategies as $\Delta(A(\bs u, \bs \xi)) = \left\{\bs p \in \Delta_N~\vert~\mathrm{supp}(\bs p) \subseteq A(\bs u , \bs \xi)\right\}.$
Then, for $(\mu, \bs u)\in \mc P(\R^N) \times \R^N$, define the response correspondence 
\begin{equation}\label{eq:Tau-def}
T(\mu;\bs u)
=\left\{\int_{\mathbb R^N}\alpha(\bs \xi) \mu(\diff \bs \xi)~\Big\vert~ 
\alpha:\mathbb R^N\to\Delta_N\ \text{Borel measurable},
\alpha(\bs \xi)\in\Delta(A(\bs u,\bs \xi))\ \mu\text{-a.s.}\right\}.
\end{equation}
When ties occur with $\mu$-probability zero (\textit{e.g.,} under absolutely continuous perturbation models), the set $A(\bs u, \bs \xi)$ is a singleton $\mu$-almost surely, hence $T(\mu; \bs u)$ is a singleton as well. Elements of $T(\mu;\bs u)$ are the induced choice-probability vectors: a mixed strategy
$\bs p\in T(\mu;\bs u)$ assigns probability $p_i$ to action $i$, where $\bs p$ is obtained by averaging tie-breaking rules over realizations of $\bs\xi$.
}

QRE is a stable state where each player optimizes their decisions as best as possible, given their behavioral limitations, fully acknowledging that other players are doing the same within their own constraints. 
Consequently, QRE models players as \textit{better responders} rather than just best responders. 
Indeed, while quantal responders are more likely to choose strategies with higher expected payoffs, they do not always select the strategy with the absolute highest expected payoff.
\begin{definition}[Quantal Response Equilibrium]
\label{def:aqre}
For a given $\mu_j \in \mc P(\R^N)$, $j\in [M]$, a strategy profile $\bs P^* = (\bs p^*_1, \ldots, \bs p^*_M)$ is a QRE of the game $\mathcal G([M], \Delta_N, (u_j)_{j\in [M]})$ if $\bs p_j^* \in T(\mu_j; (u_j(\bs e_i; \bs P_{-j}^*))_{i \in [N]})$ for all $j \in [M]$.
\end{definition}
{QRE is thus a fixed point of the quantal-response correspondences, analogous to how NE is a fixed-point of the best-response correspondences. }
Therefore, QRE serves as a statistical generalization of NE, with NE representing the limiting case where no decision errors occur.
Among the various formulations of QRE, logit-QRE is the most widely used.

\begin{example}[Logit-QRE~\citep{mckelvey1995quantal}]
    If~$\mu_j = \nu$ for all $j \in [M]$, $\nu=\otimes_{i=1}^N \nu_i$ and if~$\nu_i\in\mathcal{P}(\R)$ is a Gumbel distribution with zero mean and variance $\pi^2\eta^2 / 6$ for some $\eta>0$, then the quantal response functions are available in closed form and are equivalent to the choice probabilities in the celebrated multinomial logit model in discrete choice theory (see \cite[Theorem~5.2]{mcfadden1981econometric}), that is, $T(\nu; \bs u) = \{({\exp({u_i}/{\eta})} /\sum_{k=1}^N \exp({u_k}/{\eta}))_{i \in [N]}\}$.
    \label{ex:logit-QRE}
\end{example}
\begin{remark}Many formulations of QRE in the literature implicitly assume that the payoff perturbations admit a joint density with respect to Lebesgue measure, so that ties occur with probability zero. Under this assumption, the induced choice rule is almost surely single-valued and can be identified with a quantal response function. While convenient, this restriction is substantive: even when all marginal distributions admit densities, a joint distribution may be singular, and the event of payoff ties may occur with positive probability. In such cases, the usual pointwise definition of choice probabilities is no longer uniquely defined without a tie-breaking convention.
Our formulation makes the underlying measure-theoretic structure explicit. By defining the response correspondence~$T(\mu;\bs u)$ directly in terms of measurable selections from the set of maximizers, we allow for arbitrary dependence structures and regularity properties of the perturbation law $\mu$, including singular distributions. This ensures that the induced mixed strategy is always well defined as a Bochner integral, without requiring absolute continuity or genericity assumptions. Classical QRE models, including logit-QRE \citep{mckelvey1995quantal}, are recovered as special cases in which the correspondence $T(\mu;\bs u)$ collapses to a singleton almost surely.
\end{remark}

Recent studies have demonstrated linear-time convergence to logit-QRE using entropy-regularized best-response mappings, notably in two-player zero-sum and multi-player potential games \citep{cen2021fast,cen2023faster}, and general normal-form games via heavily regularized natural policy gradient \citep{sun2024linear}. However, these results specifically exploit entropy regularization, leaving convergence under alternative QRE models open (see Appendix~\ref{sec:computation}).
From the modeling perspective, logit-QRE relies on Luce's independence from irrelevant alternative (IIA) axiom, according to which the odds of choosing an action $i$ over another $i'$ depend on their own payoffs and remain unaffected when additional actions are introduced or removed. When IIA fails, this constraint distorts the predicted choice probabilities.
As demonstrated by \citep{ref:debreu1960review}, adding a duplicate alternative can drastically alter market shares: an effect that logit-based models fail to capture. Consequently, the assumption of identically and independently distributed (i.i.d.) additive payoff perturbations underlying standard QRE models may produce equilibrium predictions that deviate significantly from observed behavior; see Section~\ref{sec:experiments-bus-train} for empirical evidence.
Motivated by these theoretical and empirical limitations of both NE and logit-QRE, we introduce a generalized equilibrium concept that relaxes the restrictive assumption of i.i.d. additive perturbations. 
Crucially, our formulation demonstrates that the connection between equilibrium solutions and regularized best-response mappings defined over the probability simplex, previously attributed exclusively to logit-QRE, extends naturally to a significantly broader class of perturbation models.
Thus, we retain the beneficial convergence and learning dynamics typically associated with logit-QRE, while simultaneously accommodating richer and more empirically realistic behavioral phenomena (see Section~\ref{sec:empirical}).
% Appendix~\ref{sec:computation} provides a convergent algorithm for general and establishes its convergence under the
% conditions stated there.

\section{Statistical Equilibrium of Optimistic Beliefs}
\label{sec:SE-OB}
We introduce SE-OB as a generalization of both NE and QRE. Unlike conventional QRE, where each player’s response function is imposed a priori and typically derived under the assumption of i.i.d. payoff perturbations, SE‑OB permits each player to choose a joint belief about their payoff disturbances from a prescribed belief set. Consequently, SE‑OB relaxes the specification of noise distributions, allowing for arbitrary dependence structures and heterogeneity among the perturbations, and thereby recovers NE and all structural QREs as special cases when the belief sets collapse to a singleton.
More precisely, given an expected payoff vector~$\bs{u} \in \R^N$ and a prescribed belief set $\mc{B} \subseteq \mc{P}(\R^N)$, each player forms beliefs that they will receive the following expected payoff:
\begin{align}
 \label{eq:opt-exp-payoff}
 \tilde u(\bs u ; \mc B) = \sup\limits_{\mu \in \mc B} \EE_{\bs \xi \sim \mu}\left[\max\limits_{i \in [N]} u_i+ \xi_i\right],    
\end{align}
where $\bs \xi \in \R^N$ represents a random vector of perturbations governed by a Borel probability measure $\mu $ from within the belief set $\mc B$. Subsequently, each player responds stochastically according to their optimally chosen belief. When all players adopt this strategy, their collective actions culminate in a statistical equilibrium state, which we refer to as SE-OB, formally defined as follows.

\begin{definition}[Statistical Equilibrium of Optimistic Beliefs]
    For given belief sets $\mc B_j \subseteq \mc P(\R^N)$ for each player $j \in [M]$, a strategy profile $\bs P^* = (\bs p_1^*, \ldots, \bs p_M^*)$ forms an SE-OB of the game $\mc G([M], \Delta_N, (u_j)_{j \in [M]})$ if for every $j\in[M]$ there exists a belief $\mu_j^*\in \mathcal B_j$ such that
    \begin{equation}
      \EE_{\bs \xi \sim \mu}\left[ \max\limits_{i \in [N]} u_j(\bs e_i; \bs P_{-j}^*) + \xi_i \right] \leq  \EE_{\bs \xi \sim \mu^*_j}\left[ \max\limits_{i \in [N]} u_j(\bs e_i; \bs P_{-j}^*) + \xi_i \right] \quad \forall \mu \in \mc B_j,
        \label{eq:b-done}
        \tag{SE-OB}
    \end{equation}
    and $\bs p_j^* \in T(\mu_j\opt; (u_j(\bs e_i, \bs P^*_{-j}))_{i \in [N]}).$
    \label{def:se-ob}
\end{definition}

In QRE, the quantal response functions are specified exogenously, yielding a parametric formulation of statistical equilibrium. Drawing inspiration from the semi-parametric extension of random payoff models~\citep{natarajan2009persistency}, SE-OB emerges as a semi-parametric generalization of QRE. 
By appropriately selecting the players’ belief sets, SE-OB seamlessly reduces to either NE or QRE, thereby unifying and extending these established concepts. In essence, the structure of the belief sets not only drives the dynamics of the game but also determines the eventual equilibrium outcome.

A useful way to interpret \eqref{eq:opt-exp-payoff}-(SE-OB) is as a model of optimism about an unobserved dependence structure. In many environments, a player can form reasonably credible
assessments of each action in isolation (e.g., from historical data or experience), yet cannot identify
how payoff shocks co-move across actions in the same underlying state of the world, because
counterfactual payoffs are not observed. In SE-OB, this lack of identification is captured by a
\textit{belief set} $\mc B_j$ of plausible \textit{joint} distributions for the payoff shocks. The player then resolves
this dependence ambiguity optimistically by selecting the joint model $\mu \in \mc B_j$ that maximizes
$\mathbb{E}_{\xi \sim \mu}[\max_{i \in [N]} u_i+\xi_i]$. 
A particularly important class of belief sets, studied in detail below, fixes the action-wise marginal distributions while leaving their dependence unrestricted. In that case, optimism operates entirely through the choice of dependence structure. The following example illustrates the economic intuition behind this class of belief sets; the formal construction is deferred to Section~\ref{sec:mbs}.

\begin{example}[A one-player illustration: career choice under dependence ambiguity]
To make the above interpretation concrete, consider a degenerate one-player ``game'' in which a single
decision maker chooses between two actions $i\in\{1,2\}$, interpreted as two career paths. Let $u_i$
denote the baseline expected value of path $i$ (salary, growth, amenities, etc.). In addition, there is an
idiosyncratic fit component $\xi_i$ capturing latent enjoyment/aptitude that will only be learned after committing
to the path. Thus, the realized payoff from choosing $i$ is $u_i+\xi_i$.

It is natural to assume that the agent can form informative assessments of each fit component in isolation,
but cannot identify their joint dependence. Concretely, suppose the agent knows the {marginal}
distributions $\xi_i \sim F_i$, $i\in\{1,2\},$
for example, from self-assessment, internships, or observed outcomes of similar individuals.
However, the agent does not know how $\xi_1$ and $\xi_2$ co-vary, because she cannot simultaneously experience
both careers long enough to observe counterfactual fit under the same underlying state. This leads to
a {marginal belief set} of all joint distributions consistent with the given marginals:
\[
\mc B
=
\Big\{
\mu\in\mc P(\R^2):
\mu(\xi_1\le s)=F_1(s)\ \forall s\in\R,
\mu(\xi_2\le s)=F_2(s)\ \forall s\in\R
\Big\}.
\]
In this interpretation, the ambiguity lies entirely in the dependence (copula) linking $\xi_1$ and $\xi_2$.

\noindent \underline{Optimistic belief selection:} An optimistic decision maker resolves this dependence ambiguity in the direction that maximizes the expected
value of the best realized (perturbed) option. Specializing \eqref{eq:opt-exp-payoff} to $N=2$, the decision‑maker selects
\[
\mu^\star \in \argmax_{\mu\in\mc B}
\EE_{\bs \xi\sim\mu}\Big[\max\{u_1+\xi_1,u_2+\xi_2\}\Big].
\]
Importantly, this is {not} optimism about the marginal quality of either path-the marginals $F_1$ and $F_2$
are taken as given. Rather, it is optimism about the {joint structure}: the agent behaves as if the
unknown dependence is the one that is most favorable to ``finding a winner'' (informally, ``at least one of
these will click''). This captures a familiar heuristic in high-uncertainty choices (e.g.,  entrepreneurship):
the agent may not believe that either option has a high {average} payoff, but she may act as if the
uncertainty aligns in a way that makes one option realize substantial upside.

\noindent \underline{Induced stochastic choice:} Given the selected belief $\mu^\star$, the induced random-utility rule chooses the action with the higher
realized payoff $u_i+\xi_i$. Thus, the agent chooses action $i$ with probability\footnote{Under absolutely continuous beliefs, ties occur with probability zero and the above reduces to the probability
that $u_i+\xi_i$ strictly exceeds the alternative. A fully rigorous treatment is provided in Section~\ref{sec:existence}.}
\[
p_i
=
\Pr_{\bs \xi\sim\mu^\star}\left(
i\in\argmax_{b\in\{1,2\}}\{u_b+\xi_b\}
\right),
\qquad i\in\{1,2\}.
\]

\noindent \underline{Why is the resulting mixing exploration-like?}
Let $\Delta=u_1-u_2$. When ties occur with probability zero, we can write
\[
p_1
=
\Pr\big(u_1+\xi_1 \ge u_2+\xi_2\big)
=
\Pr\big(\Delta + (\xi_1-\xi_2)\ge 0\big),
\qquad p_2=1-p_1.
\]
Hence, whenever $\xi_1-\xi_2$ has nontrivial dispersion, the map $\Delta\mapsto p_1$ is smooth and interior:
if $u_1\approx u_2$ (small $\Delta$), then both actions are chosen with positive probability.
Moreover, the {dependence} between $\xi_1$ and $\xi_2$ directly controls this dispersion. For instance,
if $\xi_1=\xi_2$ almost surely (perfect positive dependence), then $\xi_1-\xi_2\equiv 0$ and choice is essentially
deterministic; by contrast, dependence patterns that make $\xi_1-\xi_2$ more variable lead to more mixing
when baseline payoffs are close. In that sense, stochastic choice here has a natural ``exploration'' interpretation:
either (i) across repeated short trials (projects, internships) the agent allocates trials across options according
to $(p_1,p_2)$, or (ii) in cross section, $(p_1,p_2)$ represents the fraction of otherwise similar agents who select
each path due to heterogeneous realized fit shocks.
\end{example}

\subsection{Existence of SE-OB}
\label{sec:existence}
We now establish the existence of an SE-OB under suitable regularity conditions on the belief sets $\mathcal B_j$.
To simplify the expressions, we first introduce some notation: for each player $j \in [M]$ and $i \in [N]$, we define $U_{j,i} = u_j(\bs e_i; \bs P_{-j})$, $U_j(\bs P_{-j}) = (U_{j,i}(\bs P_{-j}))_{i=1}^N \in \R^N.$
For each $j\in[M]$, we define the optimistic value functional as
\begin{align}\label{def:phij}
\Phi_j(\mu, \bs u) =\int_{\mathbb R^N}\left(\max_{i\in[N]} u_i+\xi_i\right) \mu(\diff \bs \xi),~\quad (\mu,\bs  u)\in\mathcal P(\mathbb R^N)\times\mathbb R^N,
\end{align}
and the optimistic-belief correspondence as $
\Belief_j^{\mathrm{opt}}(\bs u) =\argmax_{\mu\in\mathcal B_j}\Phi_j(\mu,\bs u)$.
To establish the existence of an equilibrium, we require the belief sets to satisfy the following regularity conditions.
\begin{assumption} 
\label{ass:topo-belief-sets}
For each $j \in [M]$, the belief set $\mc B_j \subseteq \mc P(\R^N)$ is \textit{nonempty}, \textit{convex} and \textit{weakly compact}. Moreover, there exists $\veps_j > 0$
 such that 
 \begin{equation}
     \sup\limits_{\mu \in \mc B_j} \int_{\R^N} \|\bs \xi \|_\infty^{1+\veps_j} \mu(\diff \bs \xi) < \infty\quad \forall j \in [M].
     \label{eq:moment-bound-belief-sets}
 \end{equation}
\end{assumption}
Next, for each $j\in[M]$, we define the best-response correspondence as
\[
\BR_j(\bs P_{-j})
=\bigcup_{\mu\in\Belief_j^{\mathrm{opt}}(U_j(\bs P_{-j}))} T(\mu;U_j(\bs P_{-j}))
\ \ \subseteq \Delta_N,
\]
and the joint correspondence $
\BR(\bs P) =  \BR_1(\bs P_{-1})\times\cdots\times \BR_M(\bs P_{-M})~ \subseteq (\Delta_N)^M.$
We are now ready to state the existence result.
\begin{theorem}[Existence of SE-OB]\label{thm:exist-endog}
Under Assumption~\ref{ass:topo-belief-sets},
there exists $\bs P^\star\in (\Delta_N)^M$ such that $\bs P^\star\in\BR(\bs P^\star)$.
Equivalently, the game admits an SE-OB.
\end{theorem}

\subsection{Equilibrium of Fixed Beliefs as an Example}
When the belief set is a singleton, then SE-OB corresponds to the equilibrium of fixed beliefs discussed in Section~\ref{sec:equilibrium-of-fixed-beliefs}.
In particular, if each player’s belief set contains a
degenerate distribution, optimistic belief selection becomes vacuous and the equilibrium notion
collapses to Nash equilibrium.
\begin{proposition}
    For some $b\in \R$, if $\mc B_j = \{\otimes_{i=1}^N \delta_{b}\}$ for all $j\in[M]$, then $\bs P \in (\Delta_N)^M$, forms an SE-OB of the game~$\mc G([M], \Delta_N, (u_j)_{j \in [M]})$ if and only if it satisfies~\eqref{eq:NE} of the same game.
    \label{prop:se-ob-ne}
\end{proposition}

Similarly, when each player’s belief set contains a single probability measure, belief selection is trivial and SE-OB reduces to the QRE induced by that distribution.
\begin{proposition}
   For some $\mu_j \in \mc P(\R^N)$ if $\mc B_j = \{\mu_j\}$ for all $j\in[M]$, then $\bs P\in (\Delta_N)^M$ forms an SE-OB of the game~$\mc G([M], \Delta_N, (u_j)_{j \in [M]})$ if and only if it forms a QRE induced by $(\mu_j)_{j \in [M]}$ of the same game. 
   \label{prop:qre-ne}
\end{proposition}
% {\color{red}

% }

\subsection{SE-OB with Marginal Belief Sets}
\label{sec:mbs}
We now investigate the class of belief sets with fixed marginal distributions of the form
\begin{equation}
\mathcal B=\Big\{\mu \in \mathcal P( \mathbb{R}^N )  \mid  \mu( z_{i} \leq s )=F_{i} ( s ) \quad \forall s \in\mathbb{R},~ \forall i \in[N] \Big\},
\label{eq:marginal-belief-set}
\tag{MBS}
\end{equation}
where $F_{i}$ stands for the cumulative distribution function of the uncertain disturbance $\xi_{i},  i \in[N]$. 
Marginal belief sets~\eqref{eq:marginal-belief-set} fully determine the marginal distributions of the components of the random vector $\bs \xi $ while leaving their interdependence (i.e., the copula)
entirely unconstrained. Often referred to as Fréchet ambiguity sets, they also coincide with the feasible set of a multi-marginal optimal transport problem~\citep{ref:abraham2017tomographic}. 
For each marginal distribution $F_i$, we define its left quantile function $F_{i}^{-1} : [ 0, 1 ] \to\mathbb{R}$ by $
F_{i}^{-1} ( t )=\inf \{s \in \R: F_{i} ( s ) \geq t \}  $ for all $ t \in [0,1]$. Using these quantile functions, we next introduce the \textit{smooth} expected payoffs, which serve as regularized counterparts to the original expected payoff functions.

\begin{definition}[Smooth expected payoffs]
    Fix payoff functions $u_j:(\Delta_N)^M\to\R$, $j\in[M]$.
For each $j\in[M]$ and $i\in[N]$, let $F_{j,i}$ be a continuous cumulative distribution function on $\R$
with (left continuous) quantile function $F_{j, i}^{-1}$ satisfying $\int_0^1 |F_{j,i}^{-1}(t)| \diff t < \infty$ for every $i \in [N],~j \in [M]$. The smooth expected payoff of player $j$ is defined by 
    \begin{equation}
        \bar u_j\left(\bs p, \bs P_{-j} ; (F_{j, i})_{i \in [N]}\right) = u_j(\bs p ; \bs P_{-j}) + \sum\limits_{i=1}^N \int_{1-p_i}^1 F_{j,i}^{-1}(t)\diff t.
        \label{eq:smooth-exp-payoffs}
    \end{equation}
    \label{def:smooth-exp-ut}
\end{definition}
Whenever $(F_{j,i})_{i\in[N]}$ are clear from context, we suppress this dependence and write
$\bar u_j(\boldsymbol p,\boldsymbol P_{-j})$ in place of
$\bar u_j(\boldsymbol p,\boldsymbol P_{-j};(F_{j,i})_{i\in[N]})$.
We refer to the game with smooth expected payoffs $\mathcal G([M],$ $\Delta_N, (\bar u_j)_{j \in [M]})$ as the \textit{smooth game} and if the cumulative distribution functions $(F_{j,i})_{i \in [N]}$ for each $j \in [M]$ are continuous, then by Lemma \ref{lem:concave-game} the smooth game is concave.
\begin{assumption}
    For each $j\in[M]$ and $i\in[N]$, the marginal cumulative distribution function $F_{j,i}:\mathbb R\to[0,1]$ is continuous, and it is strictly increasing on the set
$\{s\in\mathbb R: F_{j,i}(s)\in(0,1)\}$.
    \label{ass:marginal-belief-sets-f-strict}
\end{assumption}
The following theorem formally relates the smooth game's NE to an SE-OB of the original game.

\begin{theorem}
    Under Assumption~\ref{ass:marginal-belief-sets-f-strict}, a profile $\bs P\opt \in (\Delta_N)^M$ forms an SE-OB of the game $\mc G([M],$ $\Delta_N$,$ (u_j)_{j \in [M]})$ if and only if it is a NE of the game $\mc G([M], \Delta_N, (\bar u_j)_{j \in [M]})$.
    \label{thm:equiv-se-ob-smooth-game}
\end{theorem}

By \citep{ref:rosen1965existence}, an immediate corollary of the preceding theorem is that an SE-OB exists whenever the belief sets satisfy Assumption~\ref{ass:marginal-belief-sets-f-strict}.
\begin{corollary}
    \label{coro:existence}
    Under Assumption \ref{ass:marginal-belief-sets-f-strict}, SE-OB of the game $\mathcal G([M], \Delta_N, (u_j)_{j\in [M]})$ exists. 
\end{corollary}

Belief sets defined by~\eqref{eq:marginal-belief-set} may initially seem restrictive, as they only encompass distributions with predetermined marginal cumulative distribution functions. Nevertheless, as illustrated by the examples below, carefully choosing these marginal distributions provides substantial flexibility in constructing both established and novel statistical equilibria. Detailed derivations for these examples are presented in Section~\ref{sec:examples-proofs}.
\begin{example}[Exponential distribution model]
\label{ex:exp-marginal}
    Suppose that $\mc B_j$, $j\in[M]$, is a marginal belief set with (shifted) exponential marginals of the form \eqref{eq:marginal-belief-set} induced by the following cumulative distribution functions $F_{j,i}(s) = \max\{0,  1 - \eta_{j, i}\exp(-s/\gamma_j - 1) \} $
    with $\eta_{j,i}  > 0 $ and $\gamma_j > 0$, for all $j \in [M]$ and $i \in [N]$.
    Then, quantal response functions evaluated at the optimistic distributions $\mu_j\opt$'s that solve~\eqref{eq:opt-exp-payoff}
    are available in closed form and equivalent to $({\eta_{j,i} \exp(u_i/ \gamma_j)}/{ \sum_{k=1}^N \eta_{j,k} \exp(u_k/\gamma_j)})_{i \in [N]} \in T(\mu_j\opt; \bs u).$
\end{example}
This example shows that logit-QRE is not only induced by singleton belief sets containing a generalized extreme value distribution (see Example~\ref{ex:logit-QRE}) but also by marginal belief sets with exponential marginals.
\noindent The quantal response functions of Example~\ref{ex:exp-marginal} are known as softmax functions, which project a vector of real numbers to the probability simplex. By construction, the projected vector always has full support. However, in applications with large output spaces, sparse probability distributions are often preferable because they filter out less relevant alternatives, thereby reducing computational complexity and enhancing interpretability.
Responding to this need, \citet{ref:martins2016softmax} proposed the \textit{sparsemax} operator, which projects onto the simplex while setting sufficiently small coordinates exactly to zero. The next example shows that sparsemax choice probabilities arise endogenously in SE-OB when belief sets are marginal belief sets induced by uniform distributions. 
\begin{example}[Uniform distribution model]\label{ex:uniform}
Suppose that each $\mc B_j$, $j\in[M]$, is a marginal belief set of the form
\eqref{eq:marginal-belief-set}, induced by {uniform} marginals with common
width $\gamma_j>0$ and location shifts
$\theta_{j,i}\in\R$, $i\in[N]$, that is, for each $i\in[N]$,
\[
F_{j,i}(s)=
\begin{cases}
0, & s \le \theta_{j,i}-\gamma_j/2,\\[4pt]
({s-(\theta_{j,i}-\gamma_j/2)})/{\gamma_j},
& \theta_{j,i}-\gamma_j/2 < s < \theta_{j,i}+\gamma_j/2,\\[8pt]
1, & s \ge \theta_{j,i}+\gamma_j/2,
\end{cases}
\]
equivalently $\xi_{j,i}\sim \mathrm{Unif}([\theta_{j,i}-\gamma_j/2,\ \theta_{j,i}+\gamma_j/2])$.
Then, for every $\bs u\in\R^N$, any optimistic distribution $\mu_j\opt\in\mc B_j$
solving \eqref{eq:opt-exp-payoff} induces a quantal response with the closed form
\[
\operatorname{sparsemax}\Big(\frac{\bs u+\boldsymbol\theta_j}{\gamma_j}\Big)\ \in\
T(\mu_j\opt;\bs u),
\qquad\text{where }\boldsymbol\theta_j:=(\theta_{j,1},\dots,\theta_{j,N}).
\]
Here $\operatorname{sparsemax}:\R^N\to\Delta_N$ denotes the \emph{sparsemax} operator
\citep{ref:martins2016softmax}, defined as the Euclidean projection onto the simplex: $
\operatorname{sparsemax}(\bs z)=
\argmin_{\bs p\in\Delta_N}\ \|\bs p - \bs z\|_2^2.$
\end{example}

By \cite[Proposition 1]{ref:martins2016softmax}, the sparsemax operator admits a closed-form expression:
$\operatorname{sparsemax}(\bs z)$ is obtained by applying a common threshold to the coordinates of
$\bs z$ and truncating those below the threshold to zero, yielding a probability vector with
possibly sparse support.
Since the equilibrium strategy profile $\bs P^*=(\bs p_1^*,\ldots,\bs p_M^*)$ satisfying
\eqref{eq:b-done} is induced by quantal response functions evaluated at optimistic distributions,
the resulting equilibrium strategies may exhibit sparsity, in the sense that some actions receive
exactly zero probability. This sparsity plays a crucial role for computational tractability
(see \cite[Corollary 3.3]{daskalakis2023smooth}).
In particular, under the uniform marginal model of Example~\ref{ex:uniform} with $\theta_{j,i}=0$,
the induced quantal response takes the form $
p^\star(\bs u)=\operatorname{sparsemax}(\bs u/\gamma_j).$
As shown in Lemma~G.2, if $\gamma_j \ge \sigma N/(1-\sigma)$ for some $\sigma\in(0,1)$, then
$p^\star(\bs u)\in[0,(N\sigma)^{-1}]^N$.
Moreover, if ties occur with $\mu_j^{\opt}$-probability zero, the response correspondence
$T(\mu_j^{\opt};\bs u)$ is single-valued and satisfies $
T(\mu_j^{\opt};\bs u)\subseteq [0,(N\sigma)^{-1}]^N$.
Consequently, the equilibrium profile $\bs P^*$ belongs to the class of admissible policies of
$\sigma$-smooth Nash equilibrium introduced by \cite[Definition 3.2]{daskalakis2023smooth}, and
inherits the associated sparsity and thus the polynomial-time approximability guarantees.
Finally, the sparsity of choice probabilities generated by the uniform distribution model
qualitatively parallels Gabaix’s sparsity-based model of bounded rationality
\citep{ref:gabaix2014sparsity}, in which limited attention leads agents to ignore sufficiently
low-value alternatives.

\begin{example}[Pareto distribution model]\label{ex:pareto}
Suppose that $\mc B_j$, $j \in [M]$, is a marginal belief set with (shifted) Pareto distributed marginals of the form \eqref{eq:marginal-belief-set} induced by the cumulative distribution functions $F_{j,i}(s) = \max\{0, 1- \eta_{j,i} (-{s(q-1) }/{(\gamma_{j} q)}+ q^{-1})^{1 /(q-1)}\}$
with $\eta_{j,i} \geq0$ and $\gamma_j, q>1$ for all $j\in[M]$ and $i \in [N]$. 
Then, the quantal response functions evaluated under the optimistic distributions $\mu\opt_j$ solve \eqref{eq:opt-exp-payoff}, i.e., 
\begin{equation*} \argmax\limits_{\bs p\in \Delta_N} \bs p^\top \bs u-{\gamma_j}{(q-1)^{-1}} \sum_{i=1}^N\eta_{j,i}\left(\left({p_i}/{\eta_{j,i}}\right)^{q} - {p_i}/{\eta_{j,i}}\right) \subseteq T(\mu_j\opt; \bs u).\end{equation*}
The Pareto distribution model encapsulates the exponential model (in the limit $q \rightarrow 1$) and the uniform distribution model (for $q=2$) as special cases. The quantal response functions admit no simple closed-form representation under this model.
\end{example}

\begin{remark}[Computing SE-OB with marginal belief sets]
For marginal belief sets, Theorem~\ref{thm:equiv-se-ob-smooth-game} reduces SE-OB computation
to finding a Nash equilibrium of the associated smooth (regularized) game. We therefore provide a
numerical routine that computes SE-OB given oracle access to expected payoff evaluations and the
marginal quantile functions. While primarily intended as a computational tool for generating model
predictions and for empirical implementation, the routine can also be interpreted as a stylized
adaptive process under payoff-evaluation access. The full algorithm and a convergence guarantee
under standard stability and regularity conditions are reported in Appendix~\ref{sec:computation}.
\end{remark}

{\color{black}
\section{Informational Foundations of Marginal Belief Sets with Bandit Feedback}
\label{sec:mbs-foundations}

Section~4 defines SE-OB as a {static} equilibrium notion for the mixed extension of a finite
normal-form game. Nevertheless, the motivating environments behind SE-OB are inherently {statistical}: players form beliefs about payoff perturbations from data generated by repeated decision problems (\textit{e.g.}, repeated market instances, repeated laboratory rounds, or repeated interactions against a slowly varying population). In such environments, a central informational feature is {missing counterfactual payoffs}: in each period, a player observes the realized payoff of the action she chooses, but does not observe the payoffs of actions she did not choose in the same underlying ``state of the world.''

In this section, we present the broader applicability of optimism with a different feedback structure: the bandit feedback, in which we will also justify the marginal belief sets (MBS) as introduced in Section~\ref{sec:SE-OB} via the notion of observational equivalence.
More concretely, fix a player $j$ and condition on observable covariates, including the opponent's play, so that player $j$ faces a fixed baseline expected-payoff vector
$
\bs u^j = (u^j(\bs e_i;\bs P_{-j}))_{i\in[N]}\in\R^N.
$
The repeated data available to player $j$ consists of realized payoffs from chosen actions. As we show below, under bandit feedback, the joint dependence structure of the payoff
perturbations across actions, that is, the copula, is {not identified}, even under arbitrary adaptive
experimentation policies, whereas the action-wise marginals {are} identified. Consequently, once
a player can learn or estimate the marginals action-by-action, the {sharp nonparametric identified
set} of feasible joint perturbation laws is exactly the Fr\'echet class of all couplings with those marginals~ \eqref{eq:marginal-belief-set}.

\subsection{Bandit Feedback and Non-Identification of Dependence}
\label{sec:bandit-nonidentification}
Fix $N\in\mathbb N_+$ and suppress player indices. Let $\bs u=(u_1,\dots,u_N)\in\R^N$ be a baseline payoff vector.
Let $\{\bs \xi_t\}_{t\ge 1}$ be an i.i.d. sequence in $\R^N$ with common law $\mu\in \mc P(\R^N)$. 
At each time $t$, the decision maker chooses an action $a_t\in[N]$ and observes only the realized payoff $
y_t = u_{a_t}+\xi_{t,a_t}$,
while the counterfactual components $\{u_i+\xi_{t,i}\}_{i\neq a_t}$ are not observed. We interpret $\mu$ as the unknown joint distribution per-period payoff-perturbation vector $\bs \xi_t$, capturing cross-action dependence within a given period.

For $\mu\in \mc P(\R^N)$ and $i\in[N]$, let $\mu_i\in \mc P(\R)$ denote the $i$th marginal of $\mu$. 
We write $\mathscr B(\cdot)$ for Borel $\sigma$-algebras. Let $([N],2^{[N]})$ be the discrete measurable
space and let $(\R,\mathscr B(\R))$ be the real line with its Borel $\sigma$-algebra. 
Define the induced {observable payoff law} $\lambda_i^\mu\in \mc P(\R)$ by the pushforward $
\lambda_i^\mu = (x\mapsto u_i+x)_\# \mu_i$, 
equivalently $\lambda_i^\mu(B)=\mu_i(B-u_i)$ for all $B\in\mathscr B(\R)$.

For $t\ge 1$ define
the history space and $\sigma$-algebra
\[
H_t = ([N]\times\R)^t,
\qquad
\mathcal H_t = (2^{[N]}\otimes \mathscr B(\R))^{\otimes t}.
\]
Let $H_0=\{\varnothing\}$ with $\mathcal H_0=\{\emptyset,H_0\}$, and define the infinite product
\[
H_\infty = ([N]\times\R)^{\mathbb N},
\qquad
\mathcal H_\infty = (2^{[N]}\otimes \mathscr B(\R))^{\otimes \mathbb N}.
\]
We write $h_t=(a_1,y_1,\dots,a_t,y_t)\in H_t$ for the observed history up to time $t$.

A history-dependent and randomized policy is a sequence $\pi=(\pi_t)_{t\ge 1}$ where each $\pi_t$
is a stochastic kernel from $(H_{t-1},\mathcal H_{t-1})$ to $([N],2^{[N]})$; equivalently, for each
$i\in[N]$, the map $h_{t-1}\mapsto \pi_t(\{i\}\mid h_{t-1})$ is $\mathcal H_{t-1}$-measurable.
For brevity, we write $\pi_t(i\mid h_{t-1})$ for $\pi_t(\{i\}\mid h_{t-1})$.
Given $(\pi,\mu)$, define for each $t\ge 1$ a stochastic kernel
$Q_{t,\mu}^\pi(\cdot\mid h_{t-1})$ from $(H_{t-1},\mathcal H_{t-1})$ to
$([N]\times\R,\ 2^{[N]}\otimes\mathscr B(\R))$ by
\[
Q_{t,\mu}^\pi(A\times B \mid h_{t-1})
=
\sum_{i\in A}\pi_t(i\mid h_{t-1})\,\lambda_i^\mu(B),
\qquad
A\subseteq[N],\ B\in\mathscr B(\R).
\]
Because $[N]$ is finite and $\lambda_i^\mu(B)$ is constant in $h_{t-1}$, the mapping
$h_{t-1}\mapsto Q_{t,\mu}^\pi(A\times B\mid h_{t-1})$ is $\mathcal H_{t-1}$-measurable for each
measurable rectangle $A\times B$, hence $Q_{t,\mu}^\pi$ is a well-defined stochastic kernel.
By the Ionescu-Tulcea extension theorem, there exists a unique probability measure
$\mathbb P_\mu^\pi$ on $(H_\infty,\mathcal H_\infty)$ such that the coordinate process
$(a_t,y_t)_{t\ge 1}$ has conditional one-step laws given by $\{Q_{t,\mu}^\pi\}_{t\ge 1}$.
We refer to $\mathbb P_\mu^\pi$ as the law of the observable bandit data $(a_t, y_t)_{t\geq 1}$ when the agent follows $\pi$ and shocks are drawn accoridng to $\mu$.
\begin{definition}[Bandit observational equivalence]
\label{def:bandit-obs-eq}
Two laws $\mu,\mu'\in \mc P(\R^N)$ are {bandit-observationally equivalent} if for every policy $\pi$
and every horizon $T\in\mathbb N_+$,
\[
\mathbb{P}^{\pi}_{\mu}\big((a_1,y_1,\dots,a_T,y_T)\in\cdot\big)
=
\mathbb{P}^{\pi}_{\mu'}\big((a_1,y_1,\dots,a_T,y_T)\in\cdot\big).
\]
\end{definition}

\begin{proposition}[Copula non-identification under bandit feedback]
\label{prop:copula-nonidentification}
Let $\mu,\mu'\in \mc P(\R^N)$ satisfy $\mu_i=\mu'_i$ for all $i\in[N]$.
Then $\mu$ and $\mu'$ are bandit-observationally equivalent.
\end{proposition}
\noindent Let $\Pi$ denote the class of admissible policies and define the {statistical experiment} generated by
$\mu$ under bandit feedback as the family $\{\mathbb P_\mu^\pi:\pi\in\Pi\}$. A feature of $\mu$ is identified from bandit data only if different values of that feature induce different
experiments. Proposition~\ref{prop:copula-nonidentification} shows that the mapping
$
\mu \mapsto \{\mathbb P_\mu^\pi:\pi\in\Pi\}
$
depends on $\mu$ only through its one-dimensional marginals $(\mu_i)_{i\in[N]}$.
Thus the experiment is {not injective} in the dependence component of $\mu$, i.e., the copula is not
identified.
This non-identification result presented in Proposition~\ref{prop:copula-nonidentification} is driven by the fact that the
decision maker observes at most one coordinate of $\bs\xi_t$ in each period. Consequently, regardless of how
$\pi$ adapts to past data, the observable history never contains two components $(\xi_{t,i},\xi_{t,i'})$ from the
same shock vector, so cross-action dependence is not statistically testable. Identifying the copula would
require richer feedback, such as simultaneous multi-action outcomes (\textit{e.g.} A/B tests across subpopulations
exposed to the same underlying shock) or a persistent latent state that can be probed with multiple actions
before it changes; we abstract from such settings here.

\begin{lemma}[Marginals are identified from bandit data]
\label{lem:marginals-identified}
If $\mu,\mu'\in \mc P(\R^N)$ are bandit-observationally equivalent, then $\mu_i=\mu'_i$ for every $i\in[N]$.
\end{lemma}

\begin{corollary}[Marginal belief sets are the sharp identified set]
\label{cor:mbs-sharp-identified}
Fix a data-generating law $\mu_0\in \mc P(\R^N)$. Then, the set of bandit-observationally equivalent distributions of $\mu_0$ coincides with the belief set of the form \eqref{eq:marginal-belief-set} induced by $(\mu_{0, i})_{i \in [N]}$.
\end{corollary}

% \begin{remark}
% If instead the decision maker observes the entire payoff vector
% $(u_1+\xi_{t,1},\dots,u_N+\xi_{t,N})$ each period, then the sample is i.i.d.\ from the joint law of
% $u+\xi$. Under standard regularity, the joint law $\mu$ is nonparametrically identified; in particular, the
% dependence structure is testable and dependence ambiguity vanishes.
% \end{remark}
Corollary~\ref{cor:mbs-sharp-identified} shows that marginal belief sets arise endogenously from the
informational environment. With realized-payoff feedback, the joint dependence of payoff perturbations
cannot be identified from observable data, even under arbitrary exploration and infinite samples, while
the marginal distributions are point identified. Hence, marginal belief sets are not a permissive
relaxation of standard models but represent a nonparametric set of payoff
perturbations consistent with observed behavior.
Thus, (i) restricting belief sets to MBS matches the full set of joint models
that are observationally indistinguishable under realized-payoff feedback, and (ii) optimistic belief
selection in SE-OB can be interpreted as a disciplined {selection rule within an identified set} rather
than an ad hoc preference for correlation patterns.

% In what follows, we translate the identification results of Section~\ref{sec:bandit-nonidentification} into partial-identification statements for the expected maximum of perturbed payoffs under dependence ambiguity. 
% Formally, for a payoff vector $\bs u$, the inclusive value is defined as
% $\mathbb E_{\mu}[\max_{i\in[N]} \{u_i+\xi_i\}]$ where $\bs \xi \sim \mu$, and captures the ex ante value of having
% access to the full action menu before payoff perturbations are realized. We refer to this object as the \emph{inclusive value}.

\subsection{Bounds for the inclusive value}
\label{sec:partial-id-inclusive-value}

We next translate the identification results of Section~\ref{sec:bandit-nonidentification}
into partial-identification statements for the expected maximum of perturbed payoffs
under dependence ambiguity.
Formally, for a payoff vector $\bs u$ and a belief $\mu$, the inclusive value is
\[
V(\bs u ,\mu) = \mathbb E_{\bs \xi \sim \mu}\left[\max_{i\in[N]} \{u_i+\xi_i\}\right],
\]
which captures the ex ante value of having access to the full action menu before payoff
perturbations are realized.

We fix action-wise marginals $(F_i)_{i\in[N]}$ and let $\mc B \subseteq \mc P(\R^N)$ denote the corresponding
marginal belief set of the form \eqref{eq:marginal-belief-set}.
Under bandit feedback, $\mu$ is identified only up to $\mc B$,
so $V(\bs u;\mu)$ is only partially identified. We therefore define the lower and upper envelopes
\[
\underline{V}(\bs u)
=
\inf_{\mu\in \mc B}V(\bs u;\mu),
\qquad
\overline{V}(\bs u)
=
\sup_{\mu\in \mc B} V(\bs u;\mu).
\]
Throughout, assume $\int_0^1 |F_i^{-1}(t)|\diff t<\infty$ for each $i \in [N]$, so that $\underline V(\bs u ),\overline V(\bs u)\in\mathbb{R}$.
For real-valued random variables $X$ and $Y$, we write $X \le_{st} Y$ for first-order stochastic dominance, i.e.,
$\mathbb P(X \le t) \ge \mathbb P(Y \le t)$ for all $t\in\mathbb R$.

Note that, similar to the result established in Theorem~\ref{thm:equiv-se-ob-smooth-game}, the upper envelope can be equivalently written as
\[
\overline{V}(\bs u) = \max_{\bs p\in\Delta_N}
\bs p^\top \bs u
+
\sum_{i=1}^N\int_{1-p_i}^1 F_{i}^{-1}(t)\diff t,
\]
indicating that $\overline{V}(\bs u)$ is only determined by the one-dimensional marginals regardless of the dependence structure.
We can also characterize the lower envelope $\underline V(\bs u)$ as follows. Intuitively, to make the ex-post
best option as small as possible, the dependence structure should minimize the likelihood that at least
one action receives an unusually favorable perturbation. This is achieved by aligning the shocks across
actions so that they move together in the same rank order. The next proposition shows that the resulting
comonotone coupling makes the maximum payoff stochastically smallest among all joint laws with the
given marginals, and it yields a one-dimensional representation of $\underline V(\bs u)$.

\begin{proposition}
\label{prop:sharp-lower-bound}
Let $U\sim \mathrm{Unif}[0,1]$ and define the comonotone coupling $
\xi_i^{\mathrm{co}} = F_i^{-1}(U)$, $i\in[N],$ then the  joint law $\mu^{\mathrm{co}}\in \mc B$ and for every $\mu\in \mc B$,
\[
\max_{i\in[N]}u_i+\xi_i^{\mathrm{co}}\le_{\mathrm{st}} \max_{i\in[N]}u_i+\xi_i.
\]
Additionally, $\underline V(\bs u)$ is attained at $\mu^{\mathrm{co}}$ and is given by
\[
\underline{V}(\bs u)
=
\int_0^1\left( \max_{i\in[N]}u_i+F_i^{-1}(t)\right)\diff t.
\]
\end{proposition}

Proposition~\ref{prop:sharp-lower-bound} implies that $\underline V(\bs u)$ depends only on the
marginal distributions and can be computed from a one-dimensional integral. Economically, the
comonotone coupling corresponds to a common-shock benchmark in which actions are simultaneously
good or bad, minimizing the scope for a favorable outlier \citep{ref:dhaene2002comonotonicity_theory}.
More generally, because bandit feedback does not identify the dependence structure of payoff
perturbations by Proposition~\ref{prop:copula-nonidentification}, the inclusive value $V(\bs u;\mu)$
cannot be pinned down uniquely from the data and may vary across joint laws that share the same
action-wise marginals. The bounds $\underline V(\bs u)$ and $\overline V(\bs u)$ therefore summarize
the range of payoff evaluations consistent with the observed feedback. SE-OB resolves this residual
indeterminacy by evaluating payoffs at the upper end of this range.

}

\section{Empirical Content and Falsifiability of SE-OB}
\label{sec:falsifiability}
A central purpose of structural equilibrium models is to deliver \textit{testable restrictions} linking
observable behavior (\textit{e.g.}, choice probabilities) to primitives (payoffs and behavioral parameters), thereby enabling the discipline of counterfactual predictions.
If a model can rationalize essentially any observed behavior, then choice data cannot refute the model, and the discipline is lost.
A canonical illustration is QRE, for which, without strong restrictions on payoff shocks, QRE can rationalize arbitrary behavior in any normal-form game~\citep{haile2008empirical}.

We let $\mathcal G$ denote the class of finite $M$-player, $N$-action normal-form games and let the outcome space be $\mathcal Y = (\Delta_N)^M$. Since $\Delta_N$ is $(N-1)$-dimensional, its topological interior in $\R^N$ is empty; throughout, we therefore use the relative interior, $\Delta_N^\circ = \{\bs p \in \Delta_N: p_i > 0,~\forall i \in [N]\}$. We write $\mc Y^\circ = (\Delta^\circ_N)^M$ for the fully mixed outcome region. 
 An equilibrium model is a correspondence:
\[\mc E : \mc G \times \Theta \rightrightarrows \mc Y^\circ, \]
where $\theta \in \Theta$ collects the primitive objects. 
In the SE-OB model, $
\theta=(\mathcal B_1,\dots,\mathcal B_M)$,~ $\mathcal B_j\subseteq \mathcal P(\R^N)$, 
i.e., $\theta$ is the profile of players' belief sets over payoff perturbations. Under marginal belief sets, each
$\mathcal B_j$ is the set of all joint laws on $\R^N$ with fixed action-wise marginals $(F_{j,i})_{i\in[N]}$, so
equivalently, $\theta$ may be taken as $(F_{j,i})_{j\in[M],~i\in[N]}$ or a finite-dimensional parameterization thereof.

For a fixed game $G \in \mc G$, the model class $\Theta$ predicts the set of rationalizable outcomes $\mc E_{\Theta}(G)= \bigcup_{\theta\in\Theta}\mathcal E(G,\theta)\subseteq\mathcal Y^\circ$.
The model class is said to be non-falsifiable if, for every game $G \in \mc G$, every mixed action profile in $\mc Y^\circ$ can be rationalized by some $\theta \in \Theta$. The following definition formalizes this notion.
\begin{definition}[Falsifiability]\label{def:falsifiability}
The model class is {non-falsifiable} on $\mathcal G$ if $\mathcal E_{\Theta}(G)=\mathcal Y^\circ$ for all
$G\in\mathcal G$.
The model class is
{falsifiable} if there exist $G\in\mathcal G$ and $\bs P\in\mathcal Y^\circ$ such that $\bs P\notin\mathcal E_{\Theta}(G)$.
\end{definition}

Our equilibrium correspondence is defined on the full outcome space
$\mathcal Y=(\Delta_N)^M$, which allows for boundary behavior (e.g., sparse play under
the uniform model in Example~\ref{ex:uniform}). For empirical-content and falsifiability
statements, we often restrict attention to the relative interior
$\mathcal Y^\circ=(\Delta_N^\circ)^M$, where $\Delta_N^\circ=\{p\in\Delta_N:\; p_i>0\ \forall i\}$.
This is standard in perturbation-based models: full-support payoff shocks (e.g., logit)
imply strictly positive choice probabilities and hence predictions in $\Delta_N^\circ$
(and avoid zero-likelihood issues in estimation).\footnote{For logit,
$p_i=\exp(u_i/\eta)/\sum_k \exp(u_k/\eta)>0$ for all $i \in [N]$.}
Formally, non-falsifiability on $\mathcal Y$ is stronger than on $\mathcal Y^\circ$:
$E_\Theta(G)=\mathcal Y$ implies $E_\Theta(G)\supseteq \mathcal Y^\circ$, but not conversely,
since boundary profiles impose additional restrictions. Economically, the weaker notion is
already decisive: if a model is non-falsifiable on $\mathcal Y^\circ$, then it cannot be
rejected whenever observed play is fully mixed (even if some probabilities are small).
In what follows, we first formalize that SE-OB inherits the {non-testable restrictions} property of unrestricted QRE. 
\begin{proposition}
Fix a game $G \in \mc G$. For any $\bs Q \in \mc Y^\circ$, there exists non-empty belief sets $\{\mc B_j\}_{j \in [M]}$ such that $\bs Q$ forms an SE-OB of the game $G$. 
\label{prop:non-falsifiable-se-ob}
\end{proposition}
The following theorem establishes that SE-OB is falsifiable on $\mathcal Y^\circ$, and hence delivers nontrivial empirical content.
\begin{theorem}[Falsifiability of SE-OB with Marginal Belief Sets]
Fix a collection of marginal distributions $(F_{j,i})_{j\in[M],\,i\in[N]}$ satisfying Assumption~\ref{ass:marginal-belief-sets-f-strict},
and consider the SE-OB model in which each player belief set is the corresponding marginal belief set.
Then the induced equilibrium correspondence is falsifiable on $Y$.
\label{thm:falsifiability-seob}
\end{theorem}

\section{Empirical Evidence}
\label{sec:empirical}

This section illustrates the empirical implications of SE-OB through two complementary exercises. 
First, we consider a stylized choice environment that isolates violations of Luce independence of irrelevant alternatives. 
Second, we use synthetic normal form games to compare equilibrium predictions under logit QRE and SE-OB across different marginal specifications.\footnote{The code is publicly available in the Github repository \url{https://github.com/yugjerry/seob}.}

\subsection{Illustration: Choice Probabilities under Luce's IIA Violations}
\label{sec:experiments-bus-train}
Consider a two-player coordination game\footnote{A coordination game is one in which players achieve positive payoffs only when choosing identical actions; otherwise, their payoffs are zero.} with $K+1$ actions, where $K \in \mathbb N_+$. For a fixed strategy profile of the second player, the action ``Train'' yields the highest expected payoff for the first player, set at $u_T = 1$, while each ``Bus'' action provides an identical lower expected payoff $u_B = u_{B_k} = 0.9$ for all $k\in[K]$.
When the choice probabilities of the first player are computed according to Example~\ref{ex:logit-QRE} with $\eta = 1$, then the probability of selecting ``Train'' is given by $\Pr(\textrm{Train}) = {\exp(1)}/({\exp(1) + K \exp(0.9)})$. Note that as $K \to \infty$, $\Pr(\textrm{Train}) \to 0$. 
Thus, each bus clone equally draws choice probabilities away from the Train, illustrating that the choice probabilities indeed satisfy Luce's IIA.

In contrast, when the first player's choice probabilities are computed according to Example~\ref{ex:uniform} with each $\xi_i \sim \mathrm{Unif}([-\gamma/2, \gamma/2])$ for all $i \in [N]$ with $\gamma =1$.
Then, the probability of selecting ``Train'' is given by $\Pr(\textrm{Train})=\textrm{sparsemax}(u_T, u_{B_1}, \ldots, u_{B_K})_1$. Because all Bus alternatives have equal utilities, $\Pr(\text{Bus}) = (1- \Pr(\text{Train})) / K$ on each Bus when buses are on the support. For $\gamma > u_T - u_{B}$, we have the explicit formula:
\[\Pr(\text{Train}) = \frac{1 + \frac{K}{\gamma}(u_T - u_B)}{ K +1} \xrightarrow[K\to\infty]{}\frac{u_T-u_B}{\gamma}
=
\frac{0.1}{\gamma}. \]
Thus, unlike Logit-QRE (as presented in Example~\ref{ex:exp-marginal}) $\Pr(\text{Train})$ does not collapse to zero as more identical Bus clones are added to the menu: it converges to a strictly positive constant $0.1$ when $\gamma =1$. 
\begin{wrapfigure}{r}{0.45\textwidth}
  \centering
  \vspace{-0.5em}
  \includegraphics[width=\linewidth]{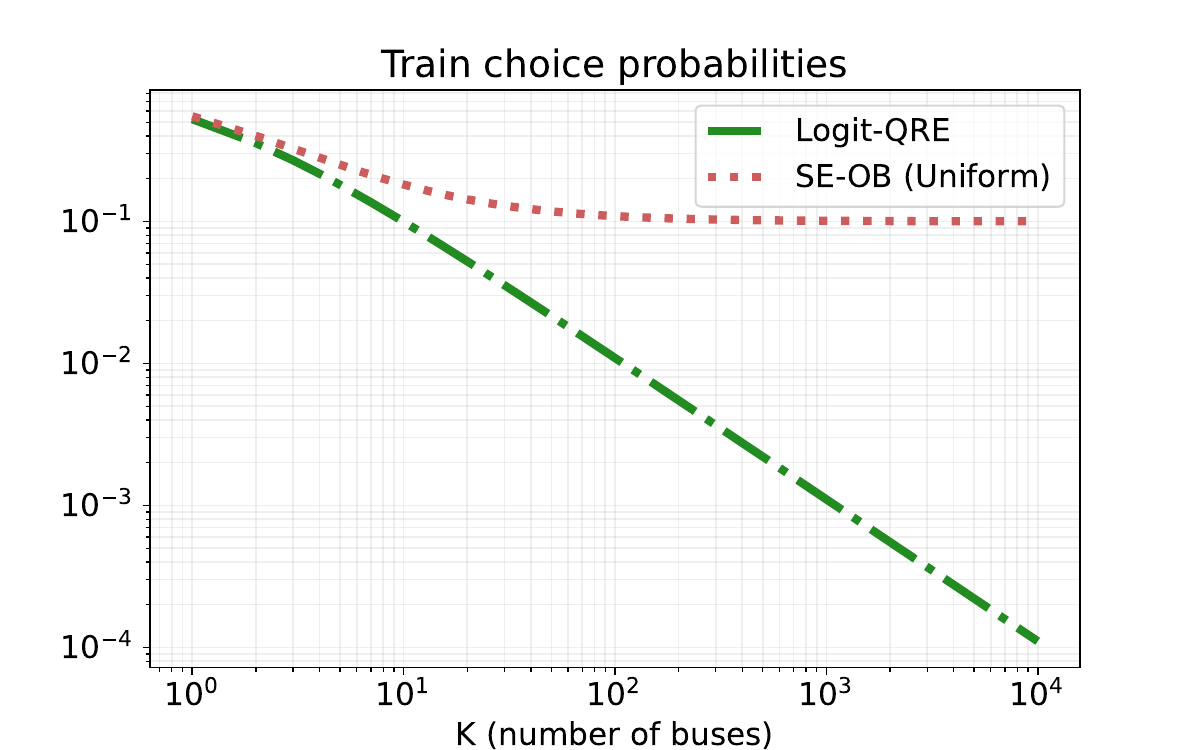}
  \caption{Log-log plot of $\Pr(\text{Train})$ as a function of the number of buses $K$ under logit-QRE and SE-OB with uniform marginal distributions.}
  \label{fig:train-prob}
  \vspace{-0.5em}
\end{wrapfigure}

Figure~\ref{fig:train-prob} reveals three key shortcomings of Luce's independence-of-irrelevant-alternatives (IIA) axiom. First, such choice probabilities suffer from \textit{predictive fragility}, as logit-QRE erroneously predicts that equilibria increasingly shift toward dominated alternatives when identical options proliferate. Second, it leads to {welfare erosion}, effectively ``pricing out'' the superior action as the number of inferior clones grows.
Moreover, such proportional dilution is difficult to reconcile
with laboratory evidence on context effects, where adding alternatives can change relative choice odds
among existing options in ways that violate IIA \citep{ref:huber1982adding}.
These pathologies vanish when the restrictive IIA assumption is relaxed. The SE-OB framework achieves such relaxation in a principled manner: if each player's belief set $\mathcal{B}_j$ consists solely of i.i.d. Gumbel distributions, it recovers logit-QRE exactly, or if $\mathcal B_j$ consists of marginal distributions of exponential distributions (see Example \ref{ex:exp-marginal}). However, if players' belief sets allow correlated or bounded-support disturbances, SE-OB with uniform marginal distributions naturally generates sparse choice probabilities, IIA-violating equilibria, as illustrated in Figure~\ref{fig:train-prob}.

\subsection{Equilibrium predictions}
We use two synthetic $2\times2$ normal-form games to illustrate how equilibrium predictions under the parametric SE-OB specifications of Section~\ref{sec:SE-OB} depend on the parameter~$\gamma$. 
Here $\gamma$ controls the strength of regularization (equivalently, the level of payoff noise): smaller $\gamma$ yields behavior closer to best responses, while larger $\gamma$ produces more diffuse mixing.
In each game, the row player (player~1) and the column player (player~2) choose between two actions,
denoted $A$ and $B$.
Payoffs are given by matrices $U_1,U_2\in\mathbb{R}^{2\times 2}$, where
$[U_1]_{ij}$ (respectively $[U_2]_{ij}$) is the payoff to the row (respectively column) player when the row chooses
$i\in\{A,B\}$ and the column chooses $j\in\{A,B\}$.
More concretely, the first game is defined by
% \begin{align*}
% U_1=\begin{pmatrix}1500 & 0\\ 0 & 500\end{pmatrix},\qquad
% U_2=\begin{pmatrix}0 & 45\\ 48 & 0\end{pmatrix},
% \end{align*}
\begin{align*}
U_1^{\mathrm{syn}}=\begin{pmatrix}
1 & 0\\[2pt]
0 & \tfrac{1}{3}
\end{pmatrix},\qquad
U_2^{\mathrm{syn}}=\begin{pmatrix}
0 & \tfrac{3}{100}\\[2pt]
\tfrac{4}{125} & 0
\end{pmatrix}.
\end{align*}
which admits one unique Nash equilibrium with $(p^*_1,q^*_1) = (16/31,1/4)$. The second game is defined by
% \begin{align*}
% U_1=\begin{pmatrix}0 & 500\\ 160 & 0\end{pmatrix},\qquad
% U_2=\begin{pmatrix}0 & 40\\ 500 & 0\end{pmatrix},
% \end{align*}
\begin{align*}
U_1^{\mathrm{syn}}=\begin{pmatrix}
0 & 1\\[2pt]
\tfrac{8}{15} & 0
\end{pmatrix},\qquad
U_2^{\mathrm{syn}}=\begin{pmatrix}
0 & \tfrac{2}{15}\\[2pt]
1 & 0
\end{pmatrix}.
\end{align*}
which has three Nash equilibria: two pure equilibria, $(p,q)=(1,0)$ and $(p,q)=(0,1)$, and one mixed
equilibrium $(p_1^*,q_1^*)=({15}/ {17},{15}/{23})$.

We compare logit-QRE to SE-OB under two marginal families as presented in Examples~\ref{ex:uniform} and~\ref{ex:pareto}.
For a given $\gamma>0$, let $p(\gamma)\in[0,1]$ denote the equilibrium probability that the row player plays $A$,
and let $q(\gamma)\in[0,1]$ denote the equilibrium probability that the column player plays $A$.
For each model, we compute the induced equilibrium pair $(p(\gamma),q(\gamma))$ over a range of $\gamma$ values.

The parameter $\gamma$ controls the amount of regularization (pay-off noise) in the players' response rules:
larger $\gamma$ corresponds to noisier (more diffuse) behavior, while smaller $\gamma$ corresponds to more
nearly best-response behavior. {Equivalently, $\gamma$ plays the role of a ``temperature'' parameter:
as $\gamma\downarrow 0$, regularization vanishes and equilibria approach Nash (up to equilibrium selection when
multiple Nash equilibria exist), whereas as $\gamma\uparrow\infty$ the induced choice probabilities become close
to uniform.}
\begin{figure}[h]
  \centering
  \begin{subfigure}[t]{0.49\textwidth}
    \centering
    \includegraphics[width=\linewidth]{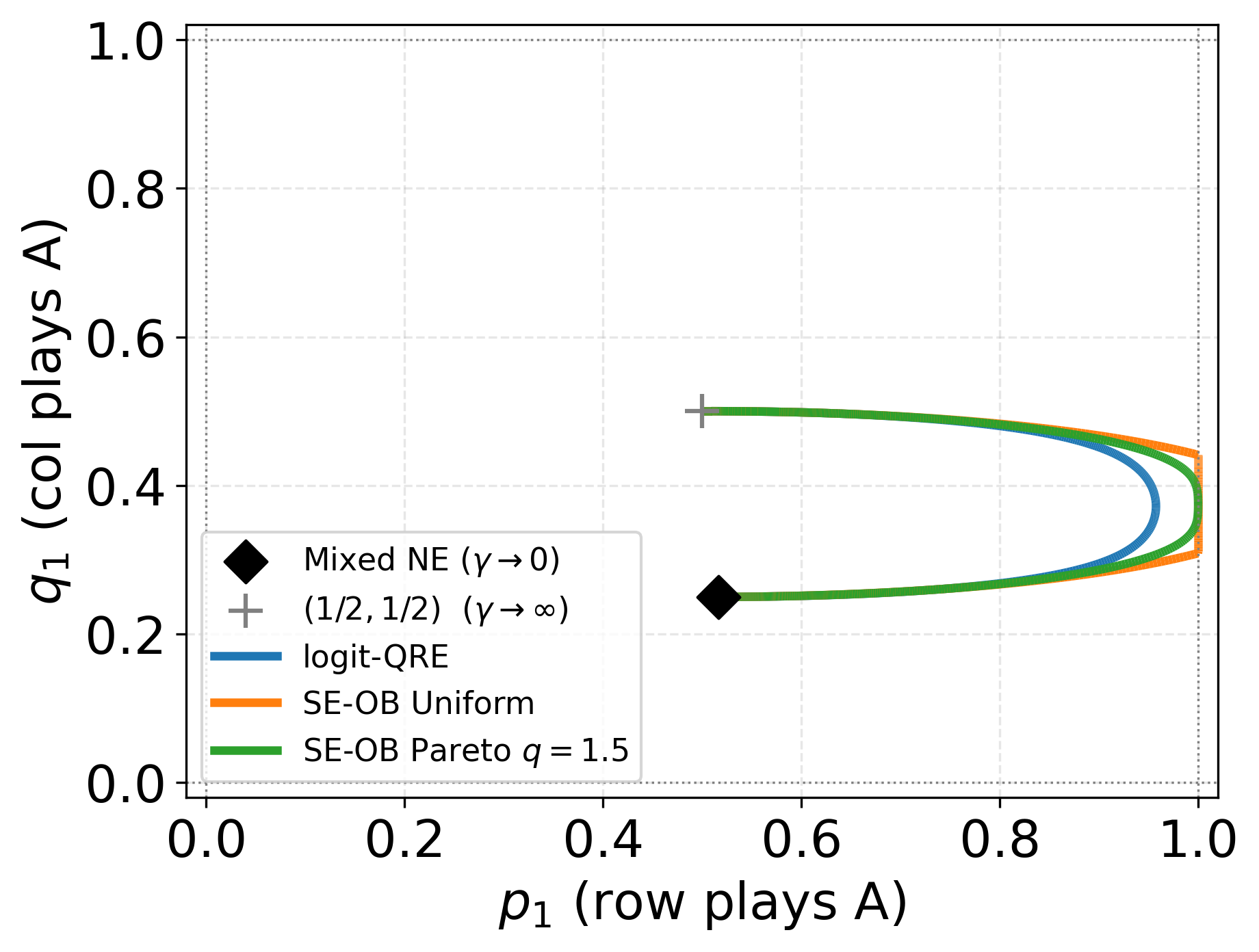}
    \caption{Game 1: SE-OB vs QRE trajectories.}
    \label{fig:traj-g1}
  \end{subfigure}\hfill
  \begin{subfigure}[t]{0.49\textwidth}
    \centering
    \includegraphics[width=\linewidth]{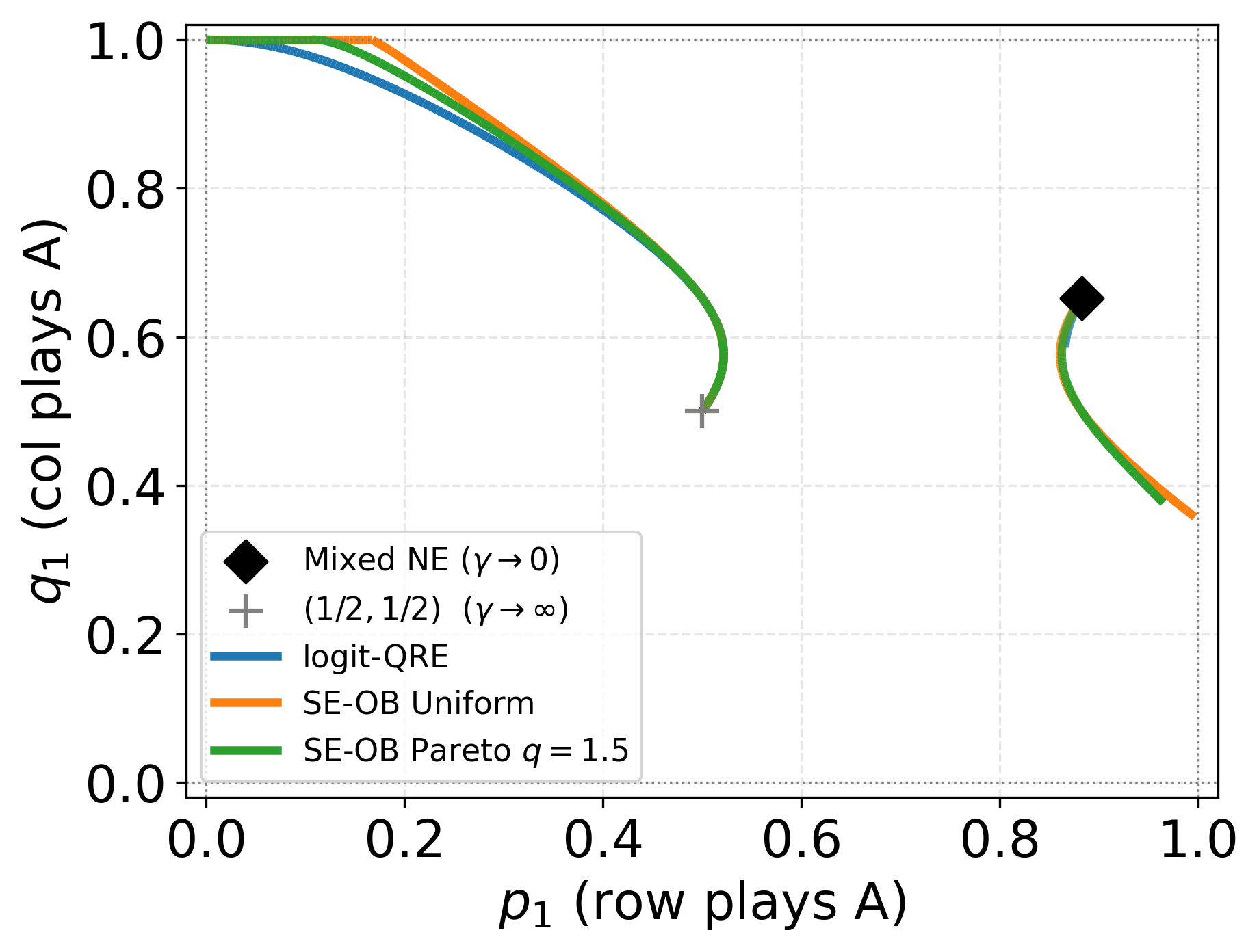}
    \caption{Game 2: SE-OB vs QRE trajectories.}
    \label{fig:traj-g2}
  \end{subfigure}
  \caption{Comparison of equilibrium trajectories $(p_1(\gamma),q_1(\gamma))$ under logit-QRE and SE-OB (uniform and Pareto marginals).}
  \label{fig:traj-side-by-side}
\end{figure}

Figure~\ref{fig:traj-side-by-side} shows a clear qualitative difference between logit-QRE and SE-OB variants as $\gamma$ varies.
Logit-QRE is generated by a softmax response rule, which assigns strictly positive probability to every action for any
finite $\gamma>0$. As a result, the logit-QRE trajectory remains strictly in the interior of the strategy simplex and
moves smoothly as $\gamma$ changes; in particular, it cannot produce exactly sparse behavior (probabilities equal to
$0$ or $1$) except in the limiting case $\gamma\downarrow 0$.
In contrast, SE-OB with uniform marginals admits sparse responses (and hence sparse equilibria) already at finite $\gamma$:
as $\gamma$ decreases, actions that are sufficiently suboptimal can drop out of the support entirely.
This generates trajectories that can approach (or hit) the boundary of the simplex and can exhibit kinks at the values of
$\gamma$ where the equilibrium support changes.
SE-OB with Pareto marginals provides an intermediate behavior (depending on the shape parameter of the Pareto distributions),
allowing support shrinkage while typically producing smoother transitions than the uniform case.
Overall, across both games, SE-OB traces a wider set of admissible equilibrium mixtures in the $(p(\gamma),q(\gamma))$ plane
than logit-QRE. This additional flexibility is particularly relevant in environments where observed play is near-deterministic,
features sharp regime shifts as incentives change, or displays strong asymmetries that are difficult to reconcile with the
smooth interior paths implied by logit-QRE.

\section{Concluding Remarks and Limitations}
\label{sec:discuss}

This paper introduced the \emph{Statistical Equilibrium of Optimistic Beliefs} (SE-OB) for the mixed extension of finite normal-form games. The key ingredient is \emph{dependence ambiguity} in payoff perturbations: players may be able to discipline action-wise marginal uncertainty from data, yet the dependence structure of shocks across actions (the copula) is typically not identified when counterfactual payoffs are unobserved. SE-OB formalizes a behavioral resolution of this ambiguity. Given opponents’ strategies, a player selects, from a prescribed belief set over joint perturbation laws, the distribution that maximizes $\mathbb{E}[\max_{i \in [N]}\{u_i+\xi_i\}]$ and then plays the induced random-utility choice rule. The SE-OB is a fixed point of this two-step mapping, and its existence is established in this paper.

SE-OB nests familiar benchmarks: Nash equilibrium is recovered as perturbations vanish, and structural QRE is recovered when belief sets collapse to singletons. Beyond existence under standard regularity conditions, we showed that SE-OB becomes especially tractable for \emph{marginal belief sets}. In this economically natural class, optimistic belief selection reduces to an optimal coupling problem, and the equilibrium can be characterized as a Nash equilibrium of an associated smooth, regularized game. This representation yields a computationally convenient fixed-point formulation, supports comparative statics, and connects SE-OB to a broad family of regularization-based response rules, including softmax/logit behavior and sparse responses.

On the empirical side, unrestricted SE-OB inherits the non-falsifiability of unrestricted QRE on the fully mixed region; however, restricting attention to marginal belief sets restores testable restrictions on outcomes. In numerical exercises, SE-OB captures systematic departures from logit-QRE predictions, including context effects and violations of independence of irrelevant alternatives that arise naturally when dependence across payoff shocks is left unrestricted.

\noindent\textbf{Limitations and directions.} Important next steps include (i) allowing intermediate attitudes toward dependence ambiguity (beyond the optimistic $\max$ criterion) and studying identification of these parameters; (ii) developing alternative, empirically grounded specifications of belief sets and understanding how estimation error in marginals propagates to equilibrium predictions; and (iii) strengthening the algorithmic and empirical side, robustness to noisy/partial feedback, extensions beyond finite normal-form games, and direct laboratory or field tests that vary the informativeness of counterfactual payoff feedback.

% \section*{Broad impact}
% This paper introduces the concept of the Statistical Equilibrium of Optimistic Beliefs (SE-OB) for the mixed extension of multi-player finite normal-form games, which shows great potential in modeling players' actual behaviors in reality, as is validated in asymmetric matching pennies games.

\bibliographystyle{abbrvnat}
\bibliography{bib}

\appendix

\newpage
\appendix

% --- Start Title Block ---
\begin{center}
  \vspace*{2\baselineskip} % Optional: Adds space at the top of the page
  {\Huge \bfseries Appendix} % Huge, Bold, Centered
  \vspace{2\baselineskip}   % Space between title and TOC
\end{center}
% --- End Title Block ---
\noindent 
This document contains supplementary material for the paper
\emph{``Statistical Equilibrium of Optimistic Beliefs.''}
It is organized as follows.
Appendices~\ref{app:proofs-existence}--\ref{app:proofs-section-falsifiability} collect proofs of the results stated in the main text.
Appendix~\ref{sec:computation} describes additional computational details, including the numerical routine and its convergence properties.
Appendix~\ref{app:add-lit-review} provides additional discussion of related literature.
Appendix~\ref{sec:examples-proofs} supplies derivations and verifies the assumptions for the examples studied in the main text.
Finally, Appendix~\ref{app:additional-tech-results} contains further technical results used in the appendices.

\addtocontents{toc}{\protect\setcounter{tocdepth}{3}}

{
  \renewcommand{\contentsname}{} % Hides "Contents" header
  \vspace*{-2em}                 % Removes gap left by hidden header
  \tableofcontents
}

% \section{Definitions}
% \label{sec:defs}
% {\color{red}\begin{definition}[Sparsemax] 
% \label{def:sparsemax}
% Given $\boldsymbol{u} \in \mathbb{R}^N$, let $\rho$ be a permutation of $[N]$ with $u_{\rho(1)} \geq u_{\rho(2)} \geq \cdots \geq u_{\rho(N)}$, for all $i \in [N]$, set
% \[
% k=\max \left\{j \in[N]: 2 N +{j \cdot u_{\rho(j)}} >\sum_{i=1}^j  u_{\rho(i)}\right\} \quad \text { and } \quad \kappa^{\star}=\frac{\left(\sum_{i=1}^k  u_{\rho(i)}\right)-2N}{k} .
% \]
% Then, $\text{sparsemax}(\bs u)=([u_i-\kappa^{\star}]_{+} / (2N))_{i \in[N]}$, where $[\cdot]_{+}=\max \{0, \cdot\}$ stands for the ramp function.
% \end{definition}}

\section{Proofs of the results in Section~\ref{sec:SE-OB}}
\label{app:proofs-existence}
This section collects the proofs of the main theoretical results stated in Section~\ref{sec:SE-OB}. 
We begin by establishing the benchmark equivalences that position SE-OB relative to standard equilibrium notions: Proposition~\ref{prop:se-ob-ne} shows that when belief sets collapse to a constant Dirac perturbation, the SE-OB condition reduces to Nash equilibrium, and Proposition~\ref{prop:qre-ne} shows that when belief sets are singletons, SE-OB coincides with the corresponding structural QRE. 
We then present the full existence proof of SE-OB (Theorem~\ref{thm:exist-endog}) via Kakutani’s fixed point theorem by verifying the required regularity properties of the optimistic-belief and response correspondences. 
Finally, we prove Theorem~\ref{thm:equiv-se-ob-smooth-game}, which characterizes SE-OB under marginal belief sets through the Nash equilibria of an associated smooth (regularized) game.

\subsection{Proofs of the results in Section~\ref{sec:existence}}
\label{app:existence}
In this section, we provide a complete proof of Theorem \ref{thm:exist-endog}, establishing the existence of an SE-OB. 
For completeness and to keep the paper self-contained, we begin by stating Kakutani's fixed point theorem \citep[Theorem 1]{kakutani1941generalization}, which is the main fixed point tool used in the argument. 
The proof proceeds by formulating each player's optimistic best response as a set-valued correspondence on the strategy simplex and verifying the hypotheses of Kakutani's theorem. 
Specifically, we establish nonemptiness, convexity, compactness, and closed graph properties of the underlying response correspondences, and then deduce upper hemicontinuity of the induced best response map. To begin, we formally present Kakutani's fixed point theorem \citep[Theorem 1]{kakutani1941generalization}.

\begin{theorem}[Kakutani's Fixed-Point Theorem \citep{kakutani1941generalization}]\label{thm:kakutani}
Let $\mc{X}$ be a non-empty, compact, convex subset of a finite-dimensional Euclidean space. If $\Phi: \mc{X} \rightrightarrows \mc{X}$ is upper hemicontinuous with non-empty, convex, closed values, then $\Phi$ has a fixed point.
\end{theorem}

In what follows, we let $\mc A$ denote the graph of the correspondence $(\bs u, \bs \xi) \rightrightarrows \Delta(A(\bs u, \bs \xi))$ defined by $
\mc A = \{(\bs u,\bs \xi,\bs a)\in \mathbb R^N\times\mathbb R^N\times \Delta_N:\bs a\in \Delta(A(\bs u,\bs \xi))\}.$ 
\begin{lemma}\label{lem:graph-closed}
$\mc A$ is closed.
\end{lemma}
\begin{proof}[Proof of Lemma~\ref{lem:graph-closed}]
Let $(\bs u_n,\bs \xi_n,\bs a_n)\to(\bs u,\bs \xi,\bs a)$ with $\bs a_n\in\Delta(A(\bs u_n,\bs \xi_n))$ for all $n \in \mathbb N$.
Our goal is to show that $\bs a\in\Delta(A(\bs u,\bs \xi))$, \textit{i.e.}, ${\mathrm{supp}}(\bs a)\subseteq A(\bs u,\bs \xi)$.

Now, take any index $i \in [N]$ with $a_i>0$. Since $\bs a_n\to \bs a$, for all large $n$ we have $a_{n,i}>0$.
Because $\bs a_n\in\Delta(A(\bs u_n,\bs \xi_n))$, it follows that $i\in A(\bs u_n,\bs \xi_n)$ for all large $n$.
Assume for contradiction that $i\notin A(\bs u,\bs \xi)$. Then there exists $\delta>0$ such that
\[
u_i+\xi_i \le \max_{k \in [N]}u_k+\xi_k -\delta.
\]
By convergence $(\bs u_n,\bs \xi_n)\to(\bs u,\bs \xi)$, the same strict inequality holds for all large $n$:
\[
u_{n,i}+\xi_{n,i} \le \max_{k\in [N]}u_{n,k}+\xi_{n,k}-\delta/2,
\]
implying $i\notin A(\bs u_n,\bs \xi_n)$ for large $n$, a contradiction.
Hence every $i$ with $a_i>0$ belongs to $A(\bs u,\bs \xi)$ and thus $\bs a\in\Delta(A(\bs u,\bs \xi))$.
\end{proof}
\begin{lemma}\label{lem:tau-props}
For every $(\mu,\bs u)\in\mathcal P(\mathbb R^N)\times\mathbb R^N$, the set $T(\mu;\bs u)$ is nonempty,
convex, and compact in $\Delta_N$. Moreover, the correspondence $(\mu,\bs u)\rightrightarrows T(\mu;\bs u)$
has a closed graph. 
\end{lemma}
\begin{proof}[Proof of Lemma~\ref{lem:tau-props}]
\noindent \textit{(Nonemptiness)}
We fix $\bs u \in \R^N$. 
Define a measurable selection by picking the smallest maximizer:
\[
i^\mathrm{min}(\bs u,\bs \xi)=\min A(\bs u,\bs \xi),
\qquad
 \alpha_0(\bs \xi)=\bs e_{i^\mathrm{min}(\bs u,\bs \xi)}.
\]
Note that for every $i \in [N]$, the set $\{\bs \xi \in \R^N: i^\mathrm{min}(\bs u,\bs \xi)=i\}$ is defined by finitely many linear inequalities and is therefore Borel. 
Since $\alpha_0(\bs \xi)$ takes only finitely many values and the pre-image of each value is a Borel set, $\alpha_0$ is measurable. Moreover, by definition,
$\alpha_0(\bs \xi)\in\Delta(A(\bs u,\bs \xi))$ for all $\bs \xi \in \R^N$. Consequently, $\int \alpha_0\diff \mu\in T(\mu;\bs u)$.

\noindent \textit{(Convexity)}
If $\bs p^k=\int \alpha_k\diff \mu\in T(\mu; \bs u)$ for $k=1,2$ and $\lambda\in[0,1]$,
define $\alpha = \lambda\alpha_1+(1-\lambda)\alpha_2$.
Because each $\Delta(A(\bs u,\bs \xi))$ is convex, $\alpha(\bs \xi)\in\Delta(A(\bs u,\bs \xi))$ $\mu$-almost surely,
and $\int \alpha\diff \mu=\lambda \bs p^1+(1-\lambda)\bs p^2$.

\noindent \textit{(Closed graph)}
Note that the graph of $T$ is in the following form:
\begin{equation}\{(\mu, \bs u, \bs p) \in \mc P(\R^N) \times \R^N \times \Delta_N~\vert~ \bs p \in T(\mu; \bs u)\}.
\label{eq:graph-t}
\end{equation}
Let $(\mu_n, \bs u_n, \bs p_n)\to(\mu,\bs u,\bs p)$ with $\bs p_n\in T (\mu_n;\bs u_n)$.
Pick measurable $\bs \alpha_n$ such that
\[
\bs p_n=\int_{\R^N} \alpha_n(\bs \xi)\mu_n(\diff \bs \xi),
\qquad
\alpha_n(\bs \xi)\in\Delta(A(\bs u_n,\bs \xi))\ \mu_n\text{-a.s.}
\]
Define probability measures $\pi_n$ on $\mathbb R^N\times\Delta_N$ by push-forward: $
\pi_n = (\mathrm{Id}, \alpha_n)_\# \mu_n.$
Then, the first marginal of $\pi_n$ is $\mu_n$, and since $\Delta_N$ is compact,
tightness of $\{\mu_n\}$ implies tightness of $\{\pi_n\}$; hence, along a subsequence,
$\pi_n \Rightarrow \pi$ weakly for some $\pi\in\mathcal P(\mathbb R^N\times\Delta_N)$.
Standard marginal convergence gives that the first marginal of $\pi$ is $\mu$.

% Recall the set
% $
% \mathcal A = \{(\bs v,\bs \xi,\bs a)\in \mathbb R^N\times\mathbb R^N\times \Delta_N:\bs a\in \Delta(A(\bs v,\bs \xi))\},$ which is closed by Lemma~\ref{lem:graph-closed}.
Next, we define the augmented measures on $\mathbb R^N\times\mathbb R^N\times\Delta_N$ as $
\bar\pi_n = \delta_{\bs u_n}\otimes \pi_n.$
By construction, $\bar\pi_n(\mc A)=1$ for all $n$. Since $\bs u_n\to \bs u$ and $\pi_n\Rightarrow \pi$,
we have $\bar\pi_n\Rightarrow \bar\pi=\delta_{\bs u}\otimes\pi$.
Because $\mathcal A$ is closed by Lemma~\ref{lem:graph-closed}, Portmanteau Theorem yields $\bar\pi(\mathcal A)=1$, \textit{\textit{i.e.}},
\[
\pi(\mathcal A_{\bs u})=1,
\qquad
\mathcal A_{\bs u}=\{(\bs \xi,\bs a)\in\mathbb R^N\times\Delta_N: \bs a\in\Delta(A(\bs u,\bs \xi))\}.
\]

Next, disintegrate $\pi$ with respect to its first marginal $\mu$:
there exists a stochastic kernel $K(\bs \xi,\diff \bs a)$ such that
\[
\pi(\diff 
\bs \xi,\diff \bs a)=\mu(\diff \bs \xi)K(\bs \xi,\diff \bs a),
\quad\text{and}\quad
K(\bs \xi,\Delta(A(\bs u,\bs \xi)))=1\ \ \mu\text{-a.s.}
\]
Define $
\alpha(\bs \xi)= \int_{\Delta_N} \bs a K(\bs \xi,\diff \bs a)\in\Delta_N. $
Measurability of $\alpha$ follows from measurability of $K$.
Because $\Delta(A(\bs u,\bs \xi))$ is convex and closed, $\alpha(\bs \xi)\in\Delta(A(\bs u,\bs \xi))$ $\mu$-almost surely.

Finally, since $(\bs \xi,\bs a)\mapsto \bs a$ is bounded and continuous on $\mathbb R^N\times\Delta_N$,
\[
\bs p = \lim_{n\to \infty} \bs p_n
= \lim_{n\to \infty}\int_{\mathbb R^N\times\Delta_N} \bs a\pi_n(\diff \bs \xi,\diff \bs a)
= \int_{\mathbb R^N\times\Delta_N} \bs a\pi(\diff \bs \xi,\diff \bs a)
= \int_{\mathbb R^N} \alpha(\bs \xi)\mu(\diff \bs \xi),
\]
so $\bs p\in T(\mu;\bs u)$, proving \eqref{eq:graph-t} is closed.
\end{proof}

\begin{lemma}[Continuity of $\Phi_j$] Under Assumption~\ref{ass:topo-belief-sets}, for each $j\in [M]$, the map $(\mu, \bs u) \mapsto \Phi_j(\mu, \bs u)$ is jointly continuous on $\mc B_j \times \mc K$ for every compact $\mc K \subset \R^N$. 
\label{lem:cont-phi}
\end{lemma}
\begin{proof}[Proof of Lemma~\ref{lem:cont-phi}]
    Fix a compact $\mc K \subset \R^N$. Let $(\mu_n, \bs u_n) \to (\mu, \bs u)$ with $\mu_n, \mu \in \mc B_j$ and $\bs u_n, \bs u \in \mc K$ where $\mu_n \Rightarrow \mu$ weakly in $\mc P(\R^N)$ and $\bs u_n \to \bs u$ in $\R^N$. Define
    \[f_n (\bs \xi) = \max\limits_{i \in [N]} u_{n,i} + \xi_i, \quad f(\bs \xi) = \max\limits_{i \in [N]} u_i + \xi_i.\]
    Then, $\Phi_j(\mu_n, \bs u_n) = \int f_n \diff \mu_n$, and $\Phi_j(\mu, \bs u) = \int f \diff \mu$. Hence, to argue that $\Phi$ is continuous it sufficies to prove 
    \begin{equation}
    \int f_n \diff \mu_n \to \int f \diff \mu.
    \label{eq:final-convergence-f-n-mu-n-to-f-mu}
    \end{equation}

    % Fix $\bs \xi \in \R^N$ and set $\Delta_n = \|\bs u_n - \bs u\|_\infty$. For every $i \in [N]$ we have $u_{n,i} \leq u_i +\Delta_n$, hence $u_{n,i} + \xi_i \leq u_i + \xi_i +\Delta_n$, implying
    % \[f_n(\bs \xi) = \max\limits_{i \in [N]} u_{n,i} + \xi_i \leq \max\limits_{i \in [N]} u_i + \xi_i +\Delta_n = f(\bs \xi) + \Delta_n.\]
    % By symmetry, we also have $f(\bs \xi) \leq f_n(\bs \xi) + \Delta_n$. Therefore, for all $\bs \xi \in \R^N$, 
    % \[|f_n(\bs \xi) - f(\bs \xi)| \leq \| \bs u_n- \bs u\|_\infty. \]
    % Integrating both sides on the above inequality with respect to $\mu_n$ yields
    % \[\left|\int f_n \diff \mu_n - \int f \diff \mu_n\right| = \left| \int (f_n -f )\diff \mu_n\right| \leq \int |f_n -f | \diff \mu_n \leq \| \bs u_n- \bs u\|_\infty. \]
    % Therefore, we may conclude that 
    % \begin{equation}
    %     \label{eq:f_n-mu-n-to-f-mu-n}
    %     \left|\int f_n \diff \mu_n - \int f \diff \mu_n\right| \to 0~ \text{as}~n\to \infty.
    % \end{equation}
    % Hence, to prove \eqref{eq:final-convergence-f-n-mu-n-to-f-mu} holds, it remains to show $
    %     \int f \diff \mu_n \to \int f \diff \mu. $
    Now, we fix $\bs \xi \in \R^N$ and set $\Delta_n = \|\bs u_n - \bs u\|_\infty$. For every $i \in [N]$, we have $u_{n,i} \leq u_i + \Delta_n$, hence $u_{n,i} + \xi_i \leq u_i + \xi_i + \Delta_n$ implying
    \[f_n(\bs \xi) = \max\limits_{i \in [N]} u_{n,i } + \xi_i \leq \max\limits_{i \in [N]} u_i + \xi_i +\Delta_n = f(\bs \xi) + \Delta_n.\]
    By interchanging $\bs u_n$ and $\bs u$ in the above chain of expressions, we also have $f(\bs \xi) \leq f_n(\bs \xi) + \Delta_n$. Therefore, for all $\bs \xi \in \R^N$, we have $f_n(\bs \xi) - f(\bs \xi)| \leq \| \bs u_n - \bs u
    \|_\infty$. By integrating both sides of the inequality with respect to $\mu_n$, we have 
    \[\left|\int f_n \diff \mu_n - \int f \diff \mu_n\right| = \left| \int (f_n -f ) \diff \mu_n\right|\leq \int |f_n -f |\diff \mu_n \leq \|\bs u_n -\bs u\|_\infty, \]
    which implies that as $n \to \infty$, 
    \begin{equation}\left|\int f_n \diff \mu_n - \int f \diff \mu_n \right| \to 0.
    \label{eq:conv-f-n-mu-n-f-mu-n}
    \end{equation}
    Therefore, to prove \eqref{eq:final-convergence-f-n-mu-n-to-f-mu}, it remains to show 
    \begin{equation}\int f \diff \mu_n \to \int f \diff \mu. 
    \label{eq:convergence-f-mu-n-f-mu}
    \end{equation}

    Observe that the function $f$ is, in general, unbounded. Indeed, since $\bs u \in \mc K$ and $\mc K$ is compact, the constant $C= \sup_{\bs  v \in \mc K}\|\bs v\|_\infty$ is finite. Then, for any $\bs \xi \in \R^N$, 
    \begin{equation}|f(\bs \xi)| = \left|\max\limits_{i \in [N]} u_i + \xi_i \right| \leq \max\limits_{i \in [N]} |u_i + \xi_i|\leq \|\bs u\|_\infty + \|\bs \xi \|_\infty \leq C+ \|\bs \xi\|_\infty.
    \label{eq:f-upper-bound}
    \end{equation}
    Hence, $f$ grows in $\|\bs\xi\|_\infty$ and is therefore unbounded. As a consequence, weak convergence of $\mu_n$ to $\mu$ alone is not sufficient to guarantee \eqref{eq:convergence-f-mu-n-f-mu}. Hence, in what follows we will truncate $f$ to apply weak convergence. For $R> 0$ define 
    \[f^{(R)}(\bs \xi ) = (f(\bs \xi) \wedge R ) \vee (-R). \]
    Because $f$ is continuous and the operations $\wedge $ and $\vee$ preserve continuity, $f^{(R)}$ is continuous; and also $f^{(R)}(\bs \xi) \in [-R, R]$ for all $\bs \xi \in \R^N$, so $f^{(R)}$ is bounded. Hence, by weak convergence $\mu_n\rightrightarrows \mu$, 
    \begin{equation}\int f^{(R)} \diff \mu_n \to \int f^{(R)} \diff \mu\quad \text{for every fixed}~R> 0. 
    \label{eq:truncated-convergence}
    \end{equation}
    Hence, it remains to control the truncation error: $\int (f- f^{(R)}) \diff \mu_n$. By definition of $f^{(R)}$, we have $f^{(R)}(\bs \xi) = f(\bs \xi)$ on the event $\{\bs \xi \in \R^N : |f(\bs \xi)| \leq R\}$, so $f(\bs\xi) - f^{(R)}(\bs \xi) =0$ there. Moreoever, on $\{\bs \xi \in \R^N :  |f(\bs \xi)|  > R\} $ we have $|f(\bs \xi) - f^{(R)}(\bs \xi)| \leq |f(\bs \xi)|$. Therefore, for all $\bs \xi$, 
    \begin{equation}
        \label{eq:f-truncation-error-upper-bound}
        |f(\bs\xi) - f^{(R)}(\bs \xi)| \leq |f(\bs \xi)| \bs 1_{\{|f(\bs \xi)| > R\}}. 
    \end{equation}
    Next, we relate the event $\{|f| > R\}$ to a tail event in $\|\bs \xi \|_{\infty}$. From \eqref{eq:f-upper-bound} for any $\bs \xi \in \R^N$, we have 
    \[|f(\bs \xi)| \leq C + \|\bs \xi\|_\infty \Rightarrow \{\bs \xi \in \R^N: |f(\bs \xi) | >R \}\subseteq \{\bs \xi \in \R^N : \|\bs \xi \|_\infty > R-C\}.\]
    Combination of the above implication with \eqref{eq:f-upper-bound} and \eqref{eq:f-truncation-error-upper-bound} yields the pointwise bound:
    \[|f(\bs \xi) - f^{(R)}(\bs \xi)| \leq (C+ \|\bs \xi\|_\infty) \bs 1_{\{\|\bs \xi\|_\infty > R-C\}}.\]
    Integrating the above inequality with respect to $\mu_n$ gives
\begin{align}
\Big|\int (f-f^{(R)})\diff \mu_n\Big|
&\le \int |f-f^{(R)}|\diff \mu_n
\le \int (C+\|\bs\xi\|_\infty)\mathbf 1_{\{\|\bs\xi\|_\infty>R-C\}}\diff\mu_n.
\label{eq:trunc-error-mun}
\end{align}

Now assume that $R$ is selected such that $R> C$ and set $S = R-C > 0$. Fix $\veps > 0$ as in Assumption \ref{ass:topo-belief-sets} such that \eqref{eq:moment-bound-belief-sets} is satisfied. Then, on the event $\{\bs \xi \in \R^N : \| \bs \xi \|_\infty > S\}$, we have $\|\bs \xi \|_{\infty}^{1+\veps} \geq S^{1+\veps}$, hence $\bs 1_{\{\|\bs \xi\|_\infty > S\}} \leq \|\bs \xi\|_\infty^{1+\veps} / S^{1+\veps}$. Integrating that inequality with respect to some $\nu \in \mc B_j$ yields
\begin{equation}\nu\left(\|\bs \xi\|_\infty  > S\right) = \int\bs 1_{\{\|\bs \xi\|_\infty > S\}} \diff \nu \leq \frac{1}{S^{1+\veps}}\int \|\bs \xi\|_\infty^{1+ \veps} \diff \nu. 
\label{eq:truncation-error-nu-ineq}
\end{equation}
Additionally, on $\{\bs \xi \in \R^N : \|\bs \xi\|_\infty > S\}$, we have $\|\bs \xi\|_\infty^\veps \geq S^\veps$, so $\|\bs \xi\|_\infty \bs 1_{\{\|\bs \xi \|_\infty > S\}} \leq \|\bs \xi\|_\infty \|\bs \xi\|_\infty^\veps / S^\veps = \|\bs \xi \|^{1+ \veps}_\infty / S^{\veps}$, and integrating both sides with respect to $\nu$ gives
\begin{equation}
    \int \|\bs \xi\|_\infty\bs 1_{\{\|\bs \xi\|_\infty > S\}} \diff \nu \leq \frac{1}{S^\veps} \int \|\bs \xi\|_\infty^{1+\veps} \diff \nu. 
    \label{eq:tail-bound}
\end{equation}
Using \eqref{eq:truncation-error-nu-ineq} and \eqref{eq:tail-bound}, we bound for every $\nu\in\mc B_j$:
\begin{align}
\int (C+\|\bs\xi\|_\infty)\mathbf 1_{\{\|\bs\xi\|_\infty>S\}}\diff\nu
&=
C\int\mathbf 1_{\{\|\bs\xi\|_\infty>S\}}\diff\nu
+
\int \|\bs\xi\|_\infty \mathbf 1_{\{\|\bs\xi\|_\infty>S\}}\diff\nu
\notag\\
&\le
\frac{C}{S^{1+\veps}}\int \|\bs\xi\|_\infty^{1+\veps}\diff\nu
+
\frac{1}{S^\veps}\int \|\bs\xi\|_\infty^{1+\veps}\diff\nu
\notag\\
&=
\left(\frac{C}{S^{1+\veps}}+\frac{1}{S^\veps}\right)
\int \|\bs\xi\|_\infty^{1+\veps}\diff\nu.
\label{eq:tail-bound-nu}
\end{align}
Taking the supremum over $\nu\in\mc B_j$ in \eqref{eq:tail-bound-nu} yields
\begin{equation}\label{eq:tail-bound-sup}
\sup_{\nu\in\mc B_j}\int (C+\|\bs\xi\|_\infty)\mathbf 1_{\{\|\bs\xi\|_\infty>S\}}\diff\nu
\le
\left(\frac{C}{S^{1+\veps}}+\frac{1}{S^\veps}\right)
\sup_{\nu\in\mc B_j}\int \|\bs\xi\|_\infty^{1+\veps}\diff\nu.
\end{equation}
By Assumption~\ref{ass:topo-belief-sets}, the supremum of the $(1+\veps)$-moment on the right is finite.
Since $\veps>0$, the prefactor ${C}/{S^{1+\veps}}+{1}/{S^\veps}\to 0$ as $S\to\infty$.
Therefore,
\begin{equation}\label{eq:tail-UI}
\sup_{\nu\in\mc B_j}\int (C+\|\bs\xi\|_\infty)\mathbf 1_{\{\|\bs\xi\|_\infty>S\}}\diff\nu
\xrightarrow[S\to\infty]{}0.
\end{equation}

Since $\mu_n\in\mc B_j$ for all $n \in \mathbb N$, from \eqref{eq:trunc-error-mun} with $S=R-C$ we obtain
\begin{align}
\sup_{n \in \mathbb N} \left|\int (f-f^{(R)})\diff \mu_n\right|
&\leq \sup_{n \in \mathbb N} \int (C+\|\bs\xi\|_\infty)\mathbf 1_{\{\|\bs\xi\|_\infty>R-C\}}\diff \mu_n
\notag\\
&\le \sup_{\nu\in\mc B_j}\int (C+\|\bs\xi\|_\infty)\mathbf 1_{\{\|\bs\xi\|_\infty>R-C\}}\diff \nu.
\label{eq:trunc-error-unif-step}
\end{align}
Letting $R\to\infty$ (so $S=R-C\to\infty$) and using \eqref{eq:tail-UI} yields
\begin{equation}\label{eq:trunc-error-unif}
\lim_{R\to\infty} \sup_{n \in \mathbb N} \left|\int (f-f^{(R)})\diff \mu_n\right|=0.
\end{equation}
Similarly, since $\mu\in\mc B_j$, applying \eqref{eq:tail-UI} with $\nu=\mu$ gives
\begin{equation}\label{eq:trunc-error-mu}
\lim_{R\to\infty} \left|\int (f-f^{(R)})\diff \mu\right|=0.
\end{equation}
Fix $\delta>0$. By \eqref{eq:trunc-error-unif} and \eqref{eq:trunc-error-mu}, choose $R$ large enough so that
\begin{equation}\label{eq:choose-R}
\sup_{n \in \mathbb N} \Big|\int (f-f^{(R)})\diff \mu_n\Big|<\frac{\delta}{3}
\qquad\text{and}\qquad
\Big|\int (f-f^{(R)})\diff \mu\Big|<\frac{\delta}{3}.
\end{equation}
Now keep this $R$ fixed. By \eqref{eq:truncated-convergence}, there exists $n_0\in \mathbb N$ such that for all $n\ge n_0$,
\begin{equation}\label{eq:choose-n}
\left|\int f^{(R)}\diff \mu_n-\int f^{(R)}\diff \mu\right|<\frac{\delta}{3}.
\end{equation}
For $n\ge n_0$, we decompose and apply the triangle inequality:
\begin{align}
\Big|\int f\diff \mu_n-\int f\diff \mu\Big|
&\le
\Big|\int (f-f^{(R)})\diff \mu_n\Big|
+
\Big|\int f^{(R)}\diff \mu_n-\int f^{(R)}\diff \mu\Big|
+
\Big|\int (f^{(R)}-f)\diff \mu\Big|.
\label{eq:decompose-f}
\end{align}
Using \eqref{eq:choose-R} and \eqref{eq:choose-n} in \eqref{eq:decompose-f} gives, for all $n\ge n_0$,
\[
\Big|\int f\diff \mu_n-\int f\diff \mu\Big|
<\frac{\delta}{3}+\frac{\delta}{3}+\frac{\delta}{3}
=\delta.
\]
Since $\delta>0$ was arbitrary, this proves \eqref{eq:convergence-f-mu-n-f-mu}.

Finally,
\[
\left|\int f_n\diff \mu_n-\int f\diff \mu\right|
\le
\left|\int f_n\diff \mu_n-\int f\diff \mu_n\right|
+
\left|\int f\diff \mu_n-\int f\diff \mu\right|,
\]
and the first term tends to $0$ by \eqref{eq:conv-f-n-mu-n-f-mu-n} while the second tends to $0$
by \eqref{eq:convergence-f-mu-n-f-mu}. Hence $\int f_n \diff \mu_n \to \int f \diff \mu$ holds, \textit{i.e.},
$\Phi_j(\mu_n,\bs u_n)\to \Phi_j(\mu,\bs u)$, proving joint continuity of $\Phi_j$ on $\mc B_j\times \mc K$.
\end{proof}

Before presenting the properties of optimistic belief sets, we introduce the definition of upper hemicontinuity as follows.
\begin{definition}
A set-valued function $\Gamma : \mc A \rightrightarrows \mc B$ is said to be \emph{upper hemicontinuous} at a point
$a \in \mc A$ if, for every open set $\mc V \subset \mc B$ with $\Gamma(a) \subset \mc V$, there exists a neighborhood
$\mc U$ of $a$ such that for all $x \in \mc U$, $\Gamma(x) \subset V$.
\end{definition}
\begin{lemma}[Properties of optimistic beliefs]\label{lem:belief-props}
Under Assumption~\ref{ass:topo-belief-sets}, for each $j\in[M]$:
\begin{enumerate}[label=\textnormal{(\roman*)}]
\item $\Belief_j^{\mathrm{opt}}(\bs u)$ is nonempty and weakly compact for every $\bs u\in \R^N$.
\item $\Belief_j^{\mathrm{opt}}(\bs u)$ is convex for every $\bs u\in \R^N$.
\item The correspondence $\bs u\rightrightarrows \Belief_j^{\mathrm{opt}}(\bs u)$ is upper hemicontinuous.
\end{enumerate}
\end{lemma}
\begin{proof}[Proof of Lemma~\ref{lem:belief-props}]
(i) Fix $j \in [M]$. By Assumption~\ref{ass:topo-belief-sets}, $\mathcal B_j$ is nonempty and weakly compact. Furthermore, Lemma~\ref{lem:cont-phi} applied with $\mc K=\{\bs u\}$ guarantees that the objective function $\Phi_j(\cdot, \bs u)$ is continuous on $\mathcal B_j$ for each fixed $\bs u$. By the Weierstrass Extreme Value Theorem, a continuous function on a compact set attains its maximum; therefore, the set of optimal beliefs $\Belief_j^{\mathrm{opt}}(\bs u)\neq\emptyset$ is nonempty. 
 Finally, since the set of maximizers for a continuous function on a compact set is closed.  Because $\Belief_j^{\mathrm{opt}}(\bs u)$ is a closed subset of a weakly comapct set $\mc B_j$, it is itself weakly compact.
(ii) Since $\Phi_j(\cdot, \boldsymbol{u})$ is affine in $\mu$ and $\mathcal{B}_j$ is convex, the set of optimal beliefs $\Belief_j^{\mathrm{opt}}(\boldsymbol{u})$ is convex for every $\boldsymbol{u} \in \mathbb{R}^N$.
(iii) Because $\mathcal{B}_j$ is weakly compact and $\Phi_j$ is jointly continuous, Berge's Maximum Theorem guarantees that the set of maximizers is upper hemicontinuous.
\end{proof}

\begin{lemma}[Properties of $\Phi_j^{\mathrm{BR}}$]\label{lem:BRj-props}
Under Assumption~\ref{ass:topo-belief-sets}, for each $j\in[M]$:
$\Phi_j^{\mathrm{BR}}(\bs  P_{-j})$ is nonempty, convex, and compact for all $\bs P_{-j}$, and the correspondence
$\bs P_{-j}\rightrightarrows \Phi_j^{\mathrm{BR}}(\bs P_{-j})$ is upper hemicontinuous.
\end{lemma}
\begin{proof}[Proof of Lemma~\ref{lem:BRj-props}]
Fix $j \in [M]$.\\

\noindent\textit{(Nonemptiness)}
By Lemma~\ref{lem:belief-props}, $\Belief_j^{\mathrm{opt}}(U_j(\bs P_{-j}))\neq\emptyset$.
Pick $\mu\in\Belief_j^{\mathrm{opt}}(U_j(\bs P_{-j}))$; by Lemma~\ref{lem:tau-props}, $T(\mu$;$U_j(\bs P_{-j}) $) $\neq$ $\emptyset$,
so $\Phi_j^{\mathrm{BR}}(\bs P_{-j})\neq\emptyset$.

\noindent\textit{(Convexity)}
Take $\bs p^1,\bs p^2\in\Phi_j^{\mathrm{BR}}(\bs P_{-j})$. Then there exist $\mu^1,\mu^2\in\Belief_j^{\mathrm{opt}}(U_j(P_{-j}))$
and $\alpha^1,\alpha^2$ such that
$p^k=\int\alpha^k\diff \mu^k$ and $\alpha^k(\xi)\in\Delta(A(U_j(P_{-j}),\xi))$ $\mu^k$-a.s.
Let $\lambda\in[0,1]$ and set $\mu^\lambda = \lambda\mu^1+(1-\lambda)\mu^2$.
Since $\Belief_j^{\mathrm{opt}}(\cdot)$ is convex by Lemma~\ref{lem:belief-props}, $\mu^\lambda \in \Belief_j^\mathrm{opt}(U_j(\bs P_{-j}))$.

If $\lambda\in(0,1)$, then $\mu^1\ll \mu^\lambda$ and $\mu^2\ll \mu^\lambda$; let $f_k=\diff \mu^k/\diff \mu^\lambda$.
Then $\lambda f_1+(1-\lambda)f_2=1$ $\mu^\lambda$-almost surely.
Define
\[
\alpha^\lambda(\bs \xi)=\lambda f_1(\bs \xi)\alpha^1(\bs \xi)+(1-\lambda)f_2(\bs \xi)\alpha^2(\bs \xi).
\]

Since $f_1, f_2$ and $\alpha^1, \alpha^2$ are Borel measurable, their linear combination $\alpha^\lambda$ is measurable. 
Moreover, $\Delta(A(U_j(\bs P_{-j})$,$\bs \xi))$ is convex, and the defining constraint
for $\alpha^k$ holds on a $\mu^k$-full set; hence it holds $\mu^\lambda$-almost surely on the region where $f_k>0$.
Therefore $\alpha^\lambda(\bs \xi)\in\Delta(A(U_j(\bs P_{-j})$,$\bs \xi))$ $\mu^\lambda$-almost surely.
Finally,
\[
\int \alpha^\lambda\diff \mu^\lambda
=\lambda\int\alpha^1\diff \mu^1+(1-\lambda)\int\alpha^2\diff \mu^2
=\lambda \bs p^1+(1-\lambda) \bs p^2,
\]
so $\lambda \bs p^1+(1-\lambda)\bs p^2\in\Phi_j^{\mathrm{BR}}(\bs P_{-j})$.

\noindent\textit{(Compactness and upper hemicontinuity)}
Let $\bs P_{-j}^{(n)}\to \bs P_{-j}$ and $\bs p^{(n)}\in\Phi_j^{\mathrm{BR}}(\bs P_{-j}^{(n)})$ with $\bs p^{(n)}\to \bs p$.
Pick $\mu^{(n)}\in\Belief_j^{\mathrm{opt}}(U_j(\bs P_{-j}^{(n)}))$ such that
$\bs p^{(n)}\in T(\mu^{(n)};U_j(\bs P_{-j}^{(n)}))$.
By weak compactness of $\mathcal B_j$, along a subsequence $\mu^{(n)}\Rightarrow \mu$.

Continuity of $U_j(\cdot)$ implies
$U_j(\bs P_{-j}^{(n)})\to U_j(\bs P_{-j})$.
By Lemma~\ref{lem:belief-props},  $\Belief_j^{\mathrm{opt}}$ is upper hemicontinuous, and thus
$\mu\in\Belief_j^{\mathrm{opt}}(U_j(\bs P_{-j}))$.

Finally, by Lemma~\ref{lem:tau-props} $T$ has a closed graph, and thus 
$\bs p\in T(\mu;U_j(\bs P_{-j}))$. Therefore, we may conclude that $\bs p\in\Phi_j^{\mathrm{BR}}(\bs P_{-j})$.
Thus $\Phi_j^{\mathrm{BR}}$ has closed graph; since values are subsets of compact $\Delta_N$, they are compact.
Closed graph together with compact values implies upper hemicontinuity of the correspondence $\bs P_{-j} \rightrightarrows \mathrm{BR}_j(\bs P_{-j})$.
\end{proof}
With these preparatory results in hand, we are now equipped to present the proof of Theorem~\ref{thm:exist-endog}.

\subsubsection{Proof of Theorem~\ref{thm:exist-endog}}
The set $\mc X = (\Delta_N)^M$ is nonempty, compact, and convex in a finite-dimensional Euclidean space.
By Lemma~\ref{lem:BRj-props}, each $\Phi_j^{\mathrm{BR}}$ has nonempty, convex, compact values and is upper
hemicontinuous; hence $\BR$ has nonempty, convex, compact values and is upper hemicontinuous.
Kakutani's fixed point theorem (Theorem~\ref{thm:kakutani}) \citep[Theorem]{kakutani1941generalization} applies to $\BR:\mc X\rightrightarrows \mc X $ and yields $\bs P^\star\in \mc X$
with $\bs P^\star\in\BR(\bs P^\star)$.
Unpacking the definition of $\BR$ gives the existence of $\mu_j^\star\in\Belief_j^{\mathrm{opt}}(U_j(\bs P_{-j}^\star))$
and $\bs p_j^\star\in T(\mu_j^\star;U_j(\bs P_{-j}^\star))$ for each $j\in [M]$, \textit{i.e.}, an SE-OB.

\subsection{Proof of Proposition~\ref{prop:se-ob-ne}}
Fix $\bs u\in\R^N$. Since adding a constant does not change the set of maximizers,
\[
A(\bs u, b\bs 1) = \argmax_{i\in[N]} (u_i+b) = \argmax_{i\in[N]} u_i,
\]
and therefore
\[
\Delta(A(\bs u, b\bs 1)) = \Delta\left(\argmax_{i\in[N]} u_i\right).
\]
By definition of $T$,
\[
T(\delta_{b\bs 1};\bs u)
=
\left\{
\int_{\R^N} \alpha(\boldsymbol{\xi})\delta_{b\bs 1}(d\boldsymbol{\xi})
\middle|
\alpha:\R^N\to \Delta_N \text{ Borel measurable, }
\alpha(\boldsymbol{\xi})\in \Delta(A(\bs u,\boldsymbol{\xi}))\ \delta_{b\bs 1}\text{-a.s.}
\right\}.
\]
Because integration against a Dirac measure collapses the integral,
\[
\int_{\R^N} \alpha(\boldsymbol{\xi})\delta_{b\bs 1}(d\boldsymbol{\xi})
= \alpha(b\bs 1),
\]
and the constraint forces $\alpha(b\bs 1)\in \Delta(A(\bs u,b\bs 1))$. Hence
\[
T(\delta_{b\bs 1};\bs u) = \Delta\left(\argmax_{i\in[N]} u_i\right).
\]
On the other hand, since $\bs q\mapsto \langle \bs q,\bs u\rangle$ is linear and $\Delta_N$ is the
convex hull of $\{\bs e_i\}_{i\in[N]}$, its maximizers are exactly the convex hull of
$\{\bs e_i: i\in\argmax_{k\in[N]}u_k\}$, \textit{i.e.},
\[
\argmax_{\bs q\in\Delta_N}\langle \bs q,\bs u\rangle
=
\Delta\left(\argmax_{i\in[N]} u_i\right).
\]
Therefore,
\begin{equation}
\label{eq:T-is-argmax}
T(\delta_{b\bs 1};\bs u)
=
\argmax_{\bs q\in\Delta_N}\langle \bs q,\bs u\rangle.
\end{equation}

Now suppose $\bs P^*$ forms an SE-OB. Since each $\mc B_j$ is a singleton, necessarily
$\mu_j^*=\delta_{b\bs 1}$ for every $j\in[M]$. The SE-OB condition gives
\[
\bs p_j^* \in T\left(\delta_{b\bs 1};\ (u_j(\bs e_i;\bs P_{-j}^*))_{i\in[N]}\right)
\qquad \forall j\in[M].
\]
By~\eqref{eq:T-is-argmax}, this is equivalent to
\[
\bs p_j^* \in \argmax_{\bs p\in\Delta_N} \left\langle \bs p,\ (u_j(\bs e_i;\bs P_{-j}^*))_{i\in[N]} \right\rangle
=
\argmax_{\bs p\in\Delta_N} u_j(\bs p;\bs P_{-j}^*),
\]
so $\bs P^*$ satisfies~\eqref{eq:NE}, i.e., it is a Nash equilibrium.

Conversely, suppose $\bs P^*$ satisfies~\eqref{eq:NE}. Then for every $j\in[M]$,
\[
\bs p_j^* \in \argmax_{\bs p\in\Delta_N} u_j(\bs p;\bs P_{-j}^*).
\]
By~\eqref{eq:T-is-argmax}, this implies
\[
\bs p_j^* \in T\left(\delta_{b\bs 1};\ (u_j(\bs e_i;\bs P_{-j}^*))_{i\in[N]}\right).
\]
Finally, since $\mc B_j$ is a singleton, $\mu_j^*=\delta_{b\bs 1}$ is trivially optimal for the
optimistic-belief condition. Hence $\bs P^*$ forms an SE-OB.

\subsection{Proof of Proposition~\ref{prop:qre-ne}}
Fix $\bs{P}^{\star} \in\left(\Delta_N\right)^M$. 
( $\Rightarrow$ ) Assume $\bs{P}^{\star}$ is an SE-OB. By definition, for each $j \in[M]$ there exists $\mu_j^{\star} \in \mathcal{B}_j$ such that (i) $\mu_j^{\star}$ maximizes the optimistic value functional over $\mathcal{B}_j$, and (ii)
\[
\bs{p}_j^{\star} \in T\left(\mu_j^{\star} ; (u_j(\bs e_i; \bs P_{-j}\opt))_{i \in [N]}\right) .
\]
Since $\mathcal{B}_j=\left\{\mu_j\right\}$, necessarily $\mu_j^{\star}=\mu_j$, and the maximization condition (i) is vacuous. Therefore
\[
\bs{p}_j^{\star} \in  T\left(\mu_j^{\star} ; (u_j(\bs e_i; \bs P_{-j}\opt))_{i \in [N]}\right) \quad \forall j \in[M],
\]
which is the definition of QRE induced by $\left(\mu_j\right)_{j \in[M]}$.

$(\Leftarrow)$ Conversely, assume $\bs{P}^{\star}$ is a QRE induced by $\left(\mu_j\right)_{j \in[M]}$, \textit{i.e.}, $\bs{p}_j^{\star} \in  T(\mu_j ; (u_j(\bs e_i; \bs P_{-j}\opt))_{i \in [N]})$ for all $j \in [M]$. Take $\mu_j^{\star}=\mu_j \in \mathcal{B}_j$. The SE-OB belief optimality condition holds automatically because $\mathcal{B}_j$ is a singleton, and the response condition is precisely the QRE condition. Hence $\bs{P}^{\star}$ forms an SE-OB.
\subsection{Proof of Theorem~\ref{thm:equiv-se-ob-smooth-game}}
$(\Leftarrow)$ Let $\bs P^* = (\bs p^*_1, \ldots, \bs p_M^*)$ be an Nash equilibrium of the smooth game $\mc G([M], \Delta_N, (\bar u_j)_{j \in [M]})$. Fix a player $j \in [M]$ and define the baseline payoff vector against $\bs P^*_{-j}$ by $\bs u^j = (u_j(\bs e_i, \bs P^*_{-j}))_{i \in [N]}$. Since, $\bs P^*$ is a Nash equilibrium of the smooth game, it satisfies:
   % \[\bs p_{j}^* \in \argmax\limits_{\bs p \in \Delta_N} \bar u_j(u_j ; (\bs p, \bs P_{-j}^*); (F_{j,i})_{i \in [N]}).\]
   \[\bs p_{j}^* \in \argmax\limits_{\bs p \in \Delta_N} \bar u_j(\bs p, \bs P_{-j}^*).\]
   By Lemma~\ref{lem:distributional-regularization}, for the marginal belief set $\mc B_j$ induced by
$(F_{j,i})_{i\in[N]}$,
\begin{equation}\label{eq:proof-reg-duality}
\sup_{\mu\in\mc B_j}\EE_{\bs\xi\sim\mu}\left[\max_{i\in[N]}u_i^j+\xi_i\right]
=
\max_{\bs p\in\Delta_N}
\left\{
\bs p^\top \bs u^j
+
\sum_{i=1}^N\int_{1-p_i}^1 F_{j,i}^{-1}(t)dt
\right\},
\end{equation}
and, moreover, any maximizer $\bs p$ on the right-hand side can be implemented as a random-utility
response under some {optimistic} belief:
there exists $\mu_j^*\in B_j^{\mathrm{opt}}(\bs u^j)$ such that $
\bs p \in T(\mu_j^*;\bs u^j).$
Applying this to $\bs p=\bs p_j^*$ yields the existence of $\mu_j^*\in\mc B_j$ such that
\[
\EE_{\bs\xi\sim\mu}\left[\max_{i\in[N]}(u_i^j+\xi_i)\right]
\le
\EE_{\bs\xi\sim\mu_j^*}\left[\max_{i\in[N]}(u_i^j+\xi_i)\right]
\qquad \forall \mu\in\mc B_j,
\]
and
\[
\bs p_j^*\in T(\mu_j^*;\bs u^j)
=
T\left(\mu_j^*; \big(u_j(\bs e_i;\bs P_{-j}^*)\big)_{i\in[N]}\right).
\]
Because $j\in[M]$ was arbitrary, we have shown that for every player $j$ there exists a belief
$\mu_j^*\in\mc B_j$ satisfying both the optimism condition and the response condition in
Definition~\ref{def:se-ob}. Therefore, $\bs P^*$ is an SE-OB of the game
$\mc G([M],\Delta_N,(u_j)_{j\in[M]})$.

$(\Rightarrow)$ Fix a player $j \in [M]$ and set 
\[\bar{\Phi}^{\rm BR}(\bs P_{-j}) = \argmax_{\bs p_j \in \Delta_N}\bar u_j(\bs p_j; \bs P_{-j})\] be the smooth best-response correspondence. Our goal is to show that $\Phi_j^{\mathrm{BR}}(\bs P_{-j}) \subseteq \bar{\Phi}^{\rm BR}_j(\bs P_{-j})$. Let $\bs P$ be an SE-OB of the game {$\mc G([M],\Delta_N, ( u_j)_{j \in [M]})$}, then $\bs p_j \in \Phi_j^{\mathrm{BR}}(\bs P_{-j})$, and by definition there exists $\mu\opt \in  B_j\opt(\bs u)$ such that $\bs p \in T(\mu\opt; \bs u)$. By definition of $T$, there exists a Borel measurable map $\alpha : \R^N \to \Delta_N$ satisfying $\alpha(\bs \xi)\in \Delta(A(\bs u, \bs \xi))$-$\mu\opt$ almost surely and $\bs p = \int \alpha(\bs \xi) \mu\opt(\diff \bs \xi)$, where $A(\bs u , \bs \xi) = \argmax_{i \in [N]} u_i + \xi_i$. Next, with any fixed $\bs v\in \R^N$, for every $\bs \xi \in \R^N$, we have 
    \begin{align*}
        \max\limits_{i \in [N]} v_i + \xi_i 
        &\geq \sum\limits_{i=1}^N \alpha_i(\bs \xi) (v_i + \xi_i) \\
        &=\sum\limits_{i=1}^N \alpha_i(\bs \xi) (u_i + \xi_i) + \sum\limits_{i=1}^N \alpha_i(\bs \xi) (v_i - u_i)\\
        &=\max\limits_{i\in[N]} \left\{u_i + \xi_i +\langle \alpha(\bs \xi), v-u\rangle\right\},
    \end{align*}
  where the inequality holds because $\alpha (\bs \xi)$ is a probability vector, and the last equality uses $\alpha(\bs \xi) \in \Delta(A(\bs u, \bs \xi))$. Recalling the definition of $\Phi_j$ in \eqref{def:phij}, integrating both sides with respect to $\mu\opt$ yields
  \begin{equation}
      \Phi_j(\mu\opt, \bs v) \geq \Phi_j(\mu\opt, \bs u) + \langle \bs p, \bs v - \bs u \rangle. 
      \label{eq:ineq-stoc-choice}
  \end{equation}
Now define the optimistic value function {$\bar \Phi_j(\bs u) = \sup_{\mu \in \mc B_j} \bar \Phi_j^{\mathrm{BR}}(\mu, \bs u)$}. 
Since $\mu\opt \in  B_j^\text{opt}(\bs u)$, we have $\bar\Phi_j(\bs u) = \Phi_j(\mu\opt, \bs u)$, and since $\bar\Phi_j(\bs v) \geq \Phi_j(\mu\opt, \bs v)$ for all $\bs v \in \R^N$, the inequality \eqref{eq:ineq-stoc-choice} further implies 
\[\bar \Phi_j(\bs v)\geq \bar\Phi_j(\bs u) + \langle p, \bs v - \bs u\rangle\quad\forall \bs v \in \R^N.\]
Thus, $\bs p$ is a subgradient of $\bar \Phi_j$ at $\bs u$, that is, $\bs p \in \partial \bar \Phi_j(\bs u)$. By Lemma~\ref{lem:distributional-regularization}, for every $\bs w \in \R^N$, it holds that
\[\bar \Phi_j(\bs w) = \max\limits_{\bs q \in \Delta_N} \left\{\langle \bs q ,\bs w \rangle + \sum\limits_{i=1}^N \int_{1-q_i}^1 F_{j,i}^{-1}(t) \diff t\right\}.\]
Denote the maximizer set above by 
\[S(\bs w) = \argmax\limits_{\bs q \in \Delta_N} \left\{\langle \bs q,\bs w\rangle + \sum\limits_{i=1}^N \int_{1-q_i}^1 F_{j,i}^{-1}(t) \diff t\right\}.\]
Since the objective function above is concave in $\bs q$, the set $S(\bs w)$ is convex. By Danskin's theorem applied to $\bar\Phi_j$, the subdifferential of $\bar \Phi_j$ at $\bs u$ equals $S(\bs u)$ (because $\mathrm{conv}(S(\bs u )) = S(\bs u)$ as $S$ is convex). Combined with \eqref{eq:ineq-stoc-choice}, we further obtain that $\bs p  \in S(\bs u)$, that is, $\bs p$ solves the optimization problem defining $\bar \Phi_j$ at $\bs w= \bs u$. Therefore, $\bs p \in \bar{\Phi}^{\rm BR}_j(\bs P_{-j})$, which completes the proof.

\subsection{Proof of Corollary~\ref{coro:existence}}
    By Lemma~\ref{lem:concave-game}, $\mathcal G([M], \Delta_N, (\bar u_j)_{j \in [M]})$ is a concave game. Consequently, by \cite[Theorem 1]{ref:rosen1965existence}, its Nash equilibrium exists. Thus, the claim follows from Theorem~\ref{thm:equiv-se-ob-smooth-game}.
\qed

\section{Proofs of the results in Section~\ref{sec:mbs-foundations}}
\subsection{Proof of Proposition~\ref{prop:copula-nonidentification}}
Fix a policy $\pi$.
For each $t\ge 1$ and history $h_{t-1}\in([N]\times\mathbb{R})^{t-1}$, define a stochastic kernel
$Q_{t,\mu}(\cdot\mid h_{t-1})$ on $[N]\times\mathbb{R}$ by
\[
Q_{t,\mu}(\{i\}\times B \mid h_{t-1}) = \pi_t(i\mid h_{t-1})\lambda_i^\mu(B),
\qquad i\in[N],\ B\in\mathscr{B}(\mathbb{R}).
\]
Equivalently, for measurable $A\subseteq[N]$ and $B\in\mathscr{B}(\mathbb{R})$,
\[
Q_{t,\mu}(A\times B\mid h_{t-1})=\sum_{i\in A}\pi_t(i\mid h_{t-1})\lambda_i^\mu(B).
\]
By construction, $Q_{t,\mu}$ is measurable in $h_{t-1}$ because $\pi_t(\cdot\mid h_{t-1})$ is a kernel and the
measures $\lambda_i^\mu$ are fixed.

Under the i.i.d.\ assumption on $\{\xi_t\}_{t\ge 1}$, conditional on $(h_{t-1},a_t=i)$ the observation $y_t$
has law $\lambda_i^\mu$ and is independent of $h_{t-1}$. Hence, the observable process
$(a_t,y_t)_{t\ge 1}$ is generated by the initial kernel $Q_{1,\mu}$ and the transition kernels
$\{Q_{t,\mu}\}_{t\ge 2}$. By the Ionescu-Tulcea extension theorem, there exists a unique probability measure
$\mathbb{P}^{\pi}_\mu$ on the path space $([N]\times\mathbb{R})^{\mathbb{N}}$ with these one-step kernels.

Now assume $\mu_i=\mu'_i$ for all $i\in[N]$. Then $\lambda_i^\mu=\lambda_i^{\mu'}$ for all $i$, so
$Q_{t,\mu}=Q_{t,\mu'}$ for every $t$. By uniqueness in Ionescu-Tulcea, the induced path measures coincide:
$\mathbb{P}^{\pi}_\mu=\mathbb{P}^{\pi}_{\mu'}$. Restricting to the first $T$ coordinates yields the stated equality
of finite-horizon laws for every $T$.
\qed

\subsection{Proof of Lemma~\ref{lem:marginals-identified}}
Fix $i\in[N]$. By Definition~\ref{def:bandit-obs-eq}, observational equivalence holds for {every}
policy, in particular for the deterministic policy $\pi^{(i)}$ that plays $a_t\equiv i$ almost surely.
Under $\pi^{(i)}$, we have $y_1=u_i+\xi_{1,i}$, hence for every $B\in\mathscr B(\R)$,
\[
\mathbb{P}^{\pi^{(i)}}_\mu\big(y_1\in u_i+B\big)=\mu_i(B),
\qquad
\mathbb{P}^{\pi^{(i)}}_{\mu'}\big(y_1\in u_i+B\big)=\mu'_i(B).
\]
Observational equivalence implies the left-hand sides are equal for all Borel $B$, so $\mu_i=\mu'_i$.
\qed

\subsection{Proof of Corollary~\ref{cor:mbs-sharp-identified}}
Denote by $[\mu_0]_{\mathrm{band}}$ as the set of bandit-observationally equivalent distributions of $\mu_0$.
If $\mu\in[\mu_0]_{\mathrm{band}}$, then Lemma~\ref{lem:marginals-identified} gives
$\mu_i=(\mu_0)_i$ for all $i$. Conversely, if $\mu_i=(\mu_0)_i$ for all $i$, then
Proposition~\ref{prop:copula-nonidentification} implies $\mu$ and $\mu_0$ are observationally equivalent,
so $\mu\in[\mu_0]_{\mathrm{band}}$.
\qed
\subsection{Proof of Proposition~\ref{prop:sharp-lower-bound}}
Let $U\sim\mathrm{Unif}[0,1]$ and define $\xi_i^{\mathrm{co}}=F_i^{-1}(U)$ for every $i \in [N]$.
Let $\mu^{\mathrm{co}}$ be the joint law of $\xi^{\mathrm{co}}=(\xi_1^{\mathrm{co}},\dots,\xi_N^{\mathrm{co}})$. We first record the basic equivalence: for any cumulative distribution function $F$ with generalized inverse
$F^{-1}(t)=\inf\{s: F(s)\ge t\}$ and any $x\in\R$,
\[
F^{-1}(t)\le x \quad\Longleftrightarrow\quad t\le F(x),\qquad t\in[0,1].
\]
Indeed, if $t\le F(x)$ then $x$ belongs to the set $\{s:F(s)\ge t\}$, so its infimum
is $\le x$. Conversely, if $F^{-1}(t)\le x$, then there exists $s\le x$ with
$F(s)\ge t$; by monotonicity of $F$ this implies $F(x)\ge F(s)\ge t$.

Applying this with $t=U$ and $F=F_i$, for every $x\in\R$ we obtain
\[
\{\xi_i^{\mathrm{co}}\le x\}
=\{F_i^{-1}(U)\le x\}
=\{U\le F_i(x)\}.
\]
Therefore,
\[
\PP(\xi_i^{\mathrm{co}}\le x)=\PP(U\le F_i(x))=F_i(x),\qquad \forall x\in\R,
\]
since $U$ is uniform on $[0,1]$. Hence the $i$th marginal of $\mu^{\mathrm{co}}$ is $F_i$ for
every $i$, and thus $\mu^{\mathrm{co}}\in\mc B$ by definition of MBS.

Fix any $\mu\in\mc B$ and let $\xi=(\xi_1,\dots,\xi_N)\sim\mu$. For any $t\in\R$,
\begin{align*}
\PP_\mu\left(\max_{i\in[N]}(u_i+\xi_i)\le t\right)
&=\PP_\mu\left(\bigcap_{i\in[N]}\{\xi_i\le t-u_i\}\right) \\
&\le \min_{i\in[N]}\PP_\mu(\xi_i\le t-u_i),
\end{align*}
because $\bigcap_i A_i\subseteq A_k$ for each $k$ implies
$\PP(\bigcap_i A_i)\le \PP(A_k)$ for each $k$, hence it is $\le \min_k \PP(A_k)$.
Since $\mu\in\mc B$, the $i$th marginal of $\mu$ is $F_i$, so
$\PP_\mu(\xi_i\le t-u_i)=F_i(t-u_i)$. Therefore
\[
\PP_\mu\left(\max_{i\in[N]}(u_i+\xi_i)\le t\right)\le \min_{i\in[N]}F_i(t-u_i),
\qquad \forall t\in\R.
\]
Under $\mu^{\mathrm{co}}$, we have shown
$\{\xi_i^{\mathrm{co}}\le x\}=\{U\le F_i(x)\}$. Hence, for any $t\in\R$,
\begin{align*}
\PP\left(\max_{i\in[N]}(u_i+\xi_i^{\mathrm{co}})\le t\right)
&=\PP\left(\bigcap_{i\in[N]}\{\xi_i^{\mathrm{co}}\le t-u_i\}\right) \\
&=\PP\left(\bigcap_{i\in[N]}\{U\le F_i(t-u_i)\}\right) \\
&=\PP\left(U\le \min_{i\in[N]}F_i(t-u_i)\right) \\
&=\min_{i\in[N]}F_i(t-u_i),
\end{align*}
again using $U\sim\mathrm{Unif}[0,1]$.
Combining the previous steps yields that for all $t\in\R$,
\[
\PP_\mu\left(\max_{i\in[N]}(u_i+\xi_i)\le t\right)
\le
\PP\left(\max_{i\in[N]}(u_i+\xi_i^{\mathrm{co}})\le t\right).
\]
Equivalently, by the definition of first-order stochastic dominance, we have 
\[
\max_{i\in[N]}(u_i+\xi_i^{\mathrm{co}})\ \le_{\mathrm{st}}\max_{i\in[N]}(u_i+\xi_i).
\]

To pass to expectations, we use the standing integrability condition
$\int_0^1|F_i^{-1}(s)|ds<\infty$ for each $i$, which implies that each marginal
has finite first absolute moment and guarantees the inclusive value envelopes are finite, that is, the expectations below are well-defined. Under this integrability,
$\EE[|\xi_i|]<\infty$ for all $i \in [N]$ for every $\mu\in\mc B$, and thus
\[
\EE\left[\left|\max_{i\in[N]}(u_i+\xi_i)\right|\right]
\le \max_{i\in[N]}|u_i|+\EE\Big[\max_{i\in[N]}|\xi_i|\Big]
\le \max_{i\in[N]}|u_i|+\sum_{i=1}^N \EE|\xi_i| <\infty.
\]
Hence both maxima are integrable.

Finally, for integrable real-valued random variables $X,Y$ with $X\le_{\mathrm{st}}Y$, we have
$\EE[X]\le \EE[Y]$ (\textit{e.g.}, by the tail-integral identity
$\EE[X]=\int_0^\infty \PP(X>t)\diff t-\int_0^\infty \PP(X\le -t)\diff t$ and the fact that
$X\le_{\mathrm{st}}Y$ implies $\PP(X>t)\le \PP(Y>t)$ and $\PP(X\le -t)\ge \PP(Y\le -t)$
for all $t\ge0$). Applying this with
$X=\max_{i \in [N]}(u_i+\xi_i^{\mathrm{co}})$ and $Y=\max_{i \in [N]}(u_i+\xi_i)$ gives
\[
\EE\left[\max_{i\in[N]}(u_i+\xi_i^{\mathrm{co}})\right]
\le
\EE_\mu\left[\max_{i\in[N]}(u_i+\xi_i)\right]
\qquad \forall \mu\in\mc B.
\]
Since $\mu^{\mathrm{co}}\in\mc B$, this shows $\mu^{\mathrm{co}}$ attains the infimum in the
lower envelope $\underline V(\bs u)=\inf_{\mu\in\mc B}\EE_\mu[\max_i(u_i+\xi_i)]$.
Define $g(t)=\max_{i\in[N]}(u_i+F_i^{-1}(t))$ on $[0,1]$. The function $g$ is measurable
and integrable because
$|g(t)|\le \max_i|u_i|+\max_i|F_i^{-1}(t)|\le \max_i|u_i|+\sum_i|F_i^{-1}(t)|$
and $\int_0^1|F_i^{-1}(t)|dt<\infty$ for each $i$. Therefore, since $U\sim\mathrm{Unif}[0,1]$,
\[
\EE_{\mu^{\mathrm{co}}}\left[\max_{i\in[N]}(u_i+\xi_i^{\mathrm{co}})\right]
=\EE[g(U)]
=\int_0^1 g(t)\diff t
=\int_0^1 \max_{i\in[N]}(u_i+F_i^{-1}(t))\diff t.
\]
This completes the proof. \qed

% Under the comonotone coupling, for each $i$,
% $\{\xi_i^{\mathrm{co}}\le t-u_i\}=\{U\le F_i(t-u_i)\}$, hence
% \[
% \mathbb{P}\left(\max_{i \in [N]}(u_i+\xi_i^{\mathrm{co}})\le t\right)
% =
% \mathbb{P}\left(U \le \min_{i \in [N]} F_i(t-u_i)\right)
% =
% \min_{i \in [N]} F_i(t-u_i).
% \]

% Therefore, for all $t$,
% \[
% \mathbb{P}_{\mu}\left(\max_{i \in [N]}(u_i+\xi_i)\le t\right)
% \le
% \mathbb{P}\left(\max_{i \in [N]}(u_i+\xi_i^{\mathrm{co}})\le t\right),
% \]
% which is $\max_{i\in[N]}(u_i+\xi_i^{\mathrm{co}})\ \le_{\mathrm{st}}\max_{i\in[N]}(u_i+\xi_i).
% $
% Integrability implies $\mathbb{E}[\max_{i \in [N]}(u_i+\xi_i^{\mathrm{co}})]\le \mathbb{E}[\max_{i \in [N]}(u_i+\xi_i)]$ for
% all $\mu\in \mc B$, proving optimality of $\mu^{\mathrm{co}}$.
% Finally, the integral representation follows from $U\sim\mathrm{Unif}[0,1]$.
% \qed
\section{Proofs of the results in Section~\ref{sec:falsifiability}}
\label{app:proofs-section-falsifiability}
This section collects the proofs for the empirical-content and falsifiability statements in Section~\ref{def:falsifiability}. 
We proceed in two steps that mirror the economic message of the section. First, we prove Proposition~\ref{prop:non-falsifiable-se-ob}, which shows that under unrestricted belief sets SE-OB inherits the classical non-falsifiability result for unrestricted QRE: any fully mixed profile can be rationalized by constructing player-specific singleton belief sets that implement the desired choice probabilities via the response correspondence $T$. 
Second, we prove Theorem~\ref{thm:falsifiability-seob}, which establishes that restricting attention to marginal belief sets restores testable restrictions: in a constant-payoff game, the associated smooth game has a unique equilibrium because each player objective becomes strictly concave on $\Delta_N$, so the SE-OB prediction set is a strict subset of the outcome space, implying falsifiability.
\subsection{Proof of Proposition~\ref{prop:non-falsifiable-se-ob}}
Fix $\bs Q\in\mc Y^\circ$.  For each player $j$, let $\bs u_j=U_j(\bs Q_{-j})\in\mathbb R^N$ be the baseline
expected-payoff vector induced by $\bs Q_{-j}$.
By \cite[Theorem 1]{haile2008empirical}, for each such $\bs u_j$ and each $\bs q_j\in \mc Y^\circ$ there exists a distribution $\mu_j^\star\in\mathcal P(\mathbb R^N)$ such
that $
\bs q_j \in T(\mu_j^\star;\bs u_j).$
Now set $B_j=\{\mu_j^\star\}$ for each player $j \in [M]$.  Then for each $j$,
$B_j\opt(\bs u_j)=\{\mu_j^\star\}$ and therefore
\[
\bs q_j\in T(\mu_j^\star;\bs u_j)\subseteq \Phi_j^{BR}(\bs Q_{-j}).
\]
Collecting across players, $\bs Q\in\prod_{j=1}^M \Phi_j^{BR}(\bs Q_{-j})=\Phi^{BR}(\bs Q)$, \textit{i.e.}, $\bs Q$ is an
SE-OB under the constructed belief sets.  Since $\bs Q$ was arbitrary in $\mc Y^\circ$, we complete our proof.
\qed

\subsection{Proof of Theorem~\ref{thm:falsifiability-seob}}
Fix a marginal profile $F=(F_{j,i})_{j\in[M],\,i\in[N]}$ satisfying
Assumption~\ref{ass:marginal-belief-sets-f-strict}, and define the corresponding parameter space by $
\Theta_F = \{F\}$.
For any game $G\in\mathcal G$, the model-class prediction set reduces to $
\mathcal E_{\Theta_F}(G)
= \bigcup_{\theta\in\Theta_F}\mathcal E(G,\theta)
= \mathcal E(G,F)$.
By Definition~\ref{def:falsifiability}, it suffices to exhibit a game $G^0\in\mathcal G$
and an outcome $\bs P\in \mc Y^\circ$
such that $\bs P\notin \mathcal E(G_0,F)$.

Consider the constant-payoff game $G^0$ in which $u_j(\bs e_i; \bs P_{-j})=0$ for all $j$, all $i\in[N]$, and all $\bs P_{-j}$.
By Theorem~\ref{thm:equiv-se-ob-smooth-game}, SE-OB profiles in $G^0$ coincide with Nash equilibria of the associated smooth game,
whose payoffs reduce to
\[
\bar u_j(\bs p_j;\bs P_{-j})=\sum_{i=1}^N \int_{1-p_{j,i}}^{1} F_{j,i}^{-1}(t)\diff t.
\]

For each $i$, the map $p\mapsto \int_{1-p}^{1} F_{j,i}^{-1}(t)\diff t$ is strictly concave on $(0,1)$ because its derivative
equals $F_{j,i}^{-1}(1-p)$, which is strictly decreasing in $p$ under Assumption~\ref{ass:marginal-belief-sets-f-strict}.
Hence $\bar u_j(\cdot;\bs P_{-j})$ is strictly concave on $\Delta_N$, so player $j$ has a unique best response $\bs p_j^{\mathrm{sm}}$.
Therefore the smooth game has a unique Nash equilibrium $\bs P^{\mathrm{sm}}=(\bs p_1^{\mathrm{sm}},\dots,\bs p_M^{\mathrm{sm}})$,
and thus if we pick any $\bs P \in\mathcal Y^\circ$ with $\bs P \neq \bs P^{\mathrm{sm}}$. Then $\bs P\notin \mathcal E(G^0)$, so the model is falsifiable.
\qed
\section{Computation of SE-OB}\label{sec:computation}
% \newbt{
This appendix records an iterative routine that can be used to compute (simulate) SE-OB in finite
normal-form games when each players belief set is a marginal belief set as in \eqref{eq:marginal-belief-set}.
In Section~\ref{sec:empirical}, SE-OB is computed via the Nash equilibrium
of the associated smooth game.
The iterative procedure described below is not used in generating the reported numerical results.
We include the iteration as an optional alternative implementation, useful when one has only oracle access
to expected-payoff evaluations and the marginal quantile functions. The update can also be interpreted as a
stylized adjustment dynamic, and we state conditions under which it converges.

% }

We consider a multi-agent learning model in which players interact repeatedly over time. 
For the settings in which each player exhibits individualized risk-preferences and receives only zeroth-order feedback in the optimization sense, we present a learning algorithm called \textit{Repeated Game of Risk-sensitive Players with Optimistic Beliefs} presented in Algorithm~\ref{alg:rs-obl}. 
Specifically, each player $j$ is endowed with a set of \textit{risk-preference} functions $(F_{j,i})_{i\in[N]}$ and a \textit{dynamic} belief set $\mc M_j \subseteq \mc P(\R^N)$, and observes only their expected payoff vectors evaluated at the current strategy profile. 
At each iteration, every player computes a one-shot payoff estimate adjusted according to their risk preferences, and updates their internal statistics.
Based on these updated statistics, each player selects an optimistic belief from~$\mc M_j$ that maximizes their expected maximal perturbed statistic, and finally chooses a strategy in accordance with that optimistic belief.

% \notebt{The algorithm that we are proposing now has uncoupled dynamics, that is, each player's strategy update depends only on (i) their own past actions, (ii) the realized payoffs to those actions and never to the payoff functions or realized actions of other players.}

\begin{algorithm}
\caption{: Repeated Game of Risk-sensitive Players with Optimistic Beliefs}
\label{alg:rs-obl}
\begin{algorithmic}[1]
\State\textbf{Input} Step-size sequence~$(\lambda_t)_{t \in [T]} \subseteq (\mathbb{R}_{++})^T$, initial estimates
$\hat{\bs u}^{(0)} = (\hat {\bs u}_1^{(0)}, \ldots, \hat {\bs u}_M^{(0)})$, risk-preference functions, continuous, non-decreasing, $F_{j,i} : \R \to [0,1]$ for all $j \in [M], i\in [N]$, belief sets $(\mc M_j)_{j \in [M]}$, $\mc M_j \subseteq \mc P(\R^N)$
\For{round $t = 1, \ldots, T$}
    \State Every player~$j \in [M]$ selects $\bs p_j^{(t)} \in T(\mu_j\opt; \hat{\bs u}_j^{(t-1)})$, where
    \begin{equation}\mu_j\opt \in \argmax\limits_{\mu \in \mc M_j} \EE_{\bs \xi \sim \mu}\left[\max\limits_{i \in [N]} \hat u_{j, i}^{(t-1)} + \xi_i \right]
    \label{eq:optimistic-bel-dyn}
    \end{equation}
    \State $\bs P^{(t)} = (\bs p_1^{(t)}, \ldots, \bs p_M^{(t)})$  
    \State Every player~$j \in [M]$ receives {$\bs r^{(t)}_j = (u_j(\bs e_i; \bs P_{-j}^{(t)}))_{i \in [N]}$}
    \State Every player~$j \in [M]$ updates {$\hat{ u}_{j,i}^{(t)} = \hat{u}_{j,i}^{(t-1)} + \lambda_t ( r^{(t)}_{j,i} + F^{-1}_{j,i}(1- p_{j,i}^{(t)}))$} for all $i \in [N]$ \label{line:statistic-update}
\EndFor
\State \textbf{Return} $\bs P^{(T)} $
\end{algorithmic}
\end{algorithm}
The risk sensitivity of the players is incorporated into Algorithm~\ref{alg:rs-obl} through the update rule in Line~\ref{line:statistic-update}. The risk-preference functions $F_{j,i}$ explicitly captures whether the player exhibits risk-averse or risk-seeking behavior. As each $F_{j,i}$ is non-decreasing, its inverse $F_{j,i}^{-1}$ is also non-decreasing. For simplicity, suppose that $F_{j,i} = F_j$ for all $i\in[N]$.
Then, at round $t$, if arm $i$ is assigned a probability that is no greater than that of arm $i'$, \textit{i.e.}, $p_{j,i}^{(t)} \leq p_{j,i'}^{(t)}$,
the monotonicity of $F_j^{-1}$ implies that $F_j^{-1}(1-p_{j,i}^{(t)}) \geq F_j^{-1}(1-p_{j,i'}^{(t)}).$
This inequality indicates that the reinforcement applied to action $i$ is stronger than that for action $i'$ when estimating $\hat{\bs u}_{j}^{(t)}$, thereby promoting increased exploration of action $i$ in the subsequent round. In this way, the update mechanism explicitly accommodates each player’s risk preferences, adapting automatically to different degrees of risk aversion or risk seeking.

In what follows, we demonstrate that when Algorithm~\ref{alg:rs-obl} is run with some risk-preference functions, its output converges to a strategy profile that forms a \eqref{eq:b-done} for the game $\mc G([M], \Delta_N, (u_j)_{j\in [M]})$, where the belief sets $\mathcal B_j$ are those induced by the same risk-preference functions.
\subsection{Convergence Analysis of Algorithm \ref{alg:rs-obl}}

\label{sec:convergence}

A fundamental goal in algorithm design for repeated multi-player games is ensuring convergence to an equilibrium, a stationary state where no player benefits from unilateral deviations. Since players only receive payoff estimates at each round, they adapt dynamically, naturally introducing the notion of regret, the gap between their actual cumulative payoff and the best fixed action in hindsight. While equilibrium clearly implies no regret, it is less obvious whether no-regret dynamics always yield equilibrium outcomes. Bridging this gap requires suitable static equilibrium conditions, such as equilibrium stability induced by players' risk preferences, and dynamic conditions governing belief updates and action selection. In subsequent sections, we examine these static and dynamic considerations, highlighting their joint role in ensuring equilibrium convergence.\\

\noindent \textit{Static Principles.}
To analyze the convergence of Algorithm~\ref{alg:rs-obl}, we study the equilibrium stability of the smooth game, ensuring global attraction to its NE.

\begin{definition}[Variational Stability]
$ \bs P' \in (\Delta_N)^M$ is a globally variationally stable strategy profile for the game $\mc G([M], \Delta_N, (u_j)_{j \in [M]})$ if
    \begin{align}
        \label{eq:vs}
        \langle (\nabla_{\bs p_j} u_j(\bs p_j; \bs P_{-j}))_{j \in [M]}, \bs P - \bs P'\rangle 
        &= \sum\limits_{j=1}^M \langle \nabla_{\bs p_j} u_j(\bs p_j; \bs P_{-j}), \bs p_j - \bs p_j' \rangle \\
        &< 0~ \forall \bs P \in (\Delta_N)^M, \bs P \neq \bs P'.\tag{\text{VS}}
    \end{align}
    \label{def:vs}
\end{definition}
Intuitively, a strategy profile $\bs P'$, is globally variationally stable if, whenever the players jointly consider any other feasible profile $\bs P \neq \bs P'$, then,  the instantaneous force generated by marginal incentives of all players would push them back toward $\bs P'$, not farther away.
\begin{definition}[Hessian of the Game] Hessian of the game~$\mc G([M], \Delta_N, (u_j)_{j \in[M]})$ that is defined as the block matrix $\bs H(\bs P,(u_j)_{j \in [M]}) = (\bs H_{j,k}(\bs P,(u_j,u_k)))_{j,k \in [M]}$ with
\[\bs H_{j,k}(\bs P,(u_j,u_k)) = \frac{1}{2} \nabla_{\bs p_k}\nabla_{\bs p_j} u_j(\bs p_j; \bs P_{-j}) + \frac{1}{2} (\nabla_{\bs p_j}\nabla_{\bs p_k} u_k(\bs p_k; \bs P_{-k}))^\top.\]
\label{def:hessian}
\end{definition}

\begin{assumption}\label{ass:vs-hessian}
The risk-preference functions $F_{j,i}$, $j \in [M]$, $i \in [N]$, are continuous, differentiable and strictly increasing whenever $F_{j, i}(s) \in (0,1)$, and 
\begin{equation}
    \min_{i\in[N]}\ \inf_{p\in(0,1)}\ \frac{1}{F'_{j,i}\left(F^{-1}_{j,i}(1-p)\right)}> \frac{1}{2} \sum\limits_{k=1,~k\neq j}^M  \left\|\left(u_j(\bs P)\Big|_{\substack{\bs p_j = \bs e_i\\ \bs p_k = \bs e_{i'}}} \right)_{i, i'\in [N]}+ \left(u_k(\bs P )\Big|_{\substack{\bs p_k = \bs e_{i}\\ \bs p_j = \bs e_{i'}}}\right)_{i, i' \in [N]}^\top \right\|.
    \label{eq:vs-hessian-cond-F_k}
\end{equation}

\end{assumption}
{
While Assumption~\ref{ass:vs-hessian} may initially appear restrictive, it is satisfied by examples presented in Section \ref{sec:SE-OB} see Section~\ref{sec:verify}.
\begin{proposition}\label{prop:vs-done}
If the risk-preference functions $F_{j,i}$, $j \in [M]$, $i\in [N]$, satisfy Assumption~\ref{ass:vs-hessian} and the smooth expected payoffs $(\bar u_j)_{j\in [M]}$ are constructed as in~\eqref{eq:smooth-exp-payoffs}, then the game $\mc G([M],\Delta_N, (\bar u_j)_{j \in [M]})$ admits a unique NE that is variationally globally stable. 
\end{proposition}
\begin{proof}[Proof of Proposition~{\ref{prop:vs-done}}]
For each $j\in[M]$, the gradient of $\bar u_j$ with respect to $\bs p_j$ is in the form of
\begin{equation}
    (\nabla_{\bs p_j} \bar u_{j}(u_j; (\bs p_j; \bs P_{-j}); (F_{j,i})_{i \in [N]}))_i = 
    u_j(\bs e_i; \bs P_{-j}) + F_{j,i}^{-1}(1-p_{j,i}) \quad \forall i \in [N].
\end{equation}
Then, for $k \neq j$, we have
\begin{equation}
    (\nabla_{\bs p_k}\nabla_{\bs p_j} \bar u_{j}(u_j; (\bs p_j; \bs P_{-j}); (F_{j,i})_{i \in [N]}))_{ii'} = u_j(\bs P)\Big|_{\substack{\bs p_j = \bs e_i\\ \bs p_k = \bs e_{i'}}},
    \label{eq:hessian-off-diagonal}
\end{equation}
which implies that for $k \neq j$, $\bs H_{j,k}$ defined in Definition \ref{def:hessian} satisfies
\[
\bs H_{j,k}(\bs P, (\bar u_j, \bar u_k)) = \frac{1}{2}\left(u_j(\bs P)\Big|_{\substack{\bs p_j = \bs e_i\\ \bs p_k = \bs e_{i'}}} \right)_{i, i'\in [N]}+ \frac{1}{2} \left(u_k(\bs P )\Big|_{\substack{\bs p_k = \bs e_{i}\\ \bs p_j = \bs e_{i'}}}\right)_{i, i' \in [N]}^\top.
\]
% \notebt{In the expression below, are we defining $\bs H_{j,j}$ or is there a typo?}
Additionally, for each $j \in [M]$, we have
% \notebt{I think the expression below should be $\bs H_{j,j}(\bs P, (\bar u_j, \bar u_j))$? or am I mistaken?}
\begin{align}
 \bs H_{j,j}(\bs P, (\bar u_j, \bar u_j)) = \nabla^2_{\bs p_j} \bar u_{j}(u_j; (\bs p_j; \bs P_{-j}); (F_{j,i})_{i \in [N]})= -\text{diag}\left\{\left(\frac{1}{F_{j,i}'(F_{j,i}^{-1}(1-p_{j,i})))}\right)_{i\in [N]}\right\},
 \label{eq:hessian-diags}
\end{align}
which is a diagonal matrix and is negative definite because $F_{j,i}$'s are strictly increasing under Assumption~\ref{ass:vs-hessian} and thus $F'_{j,i} > 0$. Then, for all $j\in [M]$, we have
% \notebt{Is it possible that we should have to the power of $-1$ on the right hand side? } \noteyg{left-hand side has inverse in the denominator}
\begin{align}
\label{eq:cond-negdef}
\frac{1}{\|\bs H^{-1}_{j,j}(\bs P, (\bar u_j, \bar u_j))\|} & = \left\| \text{diag}\left(\frac{1}{F_{j,i}'(F_{j,i}^{-1}(1-p_{j,i}))} \right)_{i \in [N]}\right\| \\
&=\min\limits_{i \in [N]} \frac{1}{F'_{j,i}(F_{j,i}^{-1}(1- p_{j,i}))} \\
&\geq \min\limits_{i \in [N]} \inf\limits_{p' \in (0,1)}  \frac{1}{F_{j,i}'(F_{j,i}^{-1}(1-p'))}\\
& > \frac{1}{2}\sum\limits_{k=1, k\neq j}^M  \bigg\|\left(u_j(\bs P)\Big|_{\substack{\bs p_j = \bs e_i\\ \bs p_k = \bs e_{i'}}} \right)_{i, i'\in [N]}+ \left(u_k(\bs P )\Big|_{\substack{\bs p_k = \bs e_{i}\\ \bs p_j = \bs e_{i'}}}\right)_{i, i' \in [N]}^\top \bigg\|\\
&= \sum\limits_{k=1, k\neq j}^M \| \bs H_{j,k}(\bs P, (\bar u_j, \bar u_k)) \|,
\end{align}
where the first equality follows by \eqref{eq:hessian-diags} and the second equality by the definition of operator norm. The strict inequality holds as $F_{j,i}$'s satisfy Assumption~\ref{ass:vs-hessian} and the last inequality follows by~\eqref{eq:hessian-off-diagonal}.
The ultimate inequality above implies that $\bs  H(\bs P, \bar u_j)_{j \in [M]})$ is strictly diagonally dominant (see \cite[Definition~1]{feingold1962block}). 
Therefore, by~\cite[Theorem~9]{feingold1962block}, we have
\[\bs  H\left(\bs P, (\bar u_{j})_{j \in [M]} \right) \prec 0.\] 
The claim then follows directly from Lemma~\ref{lem:vs-hessian}.
\end{proof}
% \end{proof}

\noindent \textit{Dynamic Principles.} At round $t$ of Algorithm~\ref{alg:rs-obl}, each player uses their risk-adjusted cumulative payoff estimate, $\hat{\bs u}_j^{(t-1)}$, to select the belief from their belief set $\mathcal{M}_j$ that maximizes the expected maximal perturbed cumulative payoff. 
Once this optimistic belief is identified, then the player adopts it by choosing a strategy based on the corresponding quantile function.
Throughout this section, we will assume that the belief sets~$\mc M_j$ are of the form~\eqref{eq:marginal-belief-set} induced by the cumulative distribution functions~$(G_{j,i})_{i \in [N]}$. By Lemma~\ref{lem:distributional-regularization}, if $\bs p\opt (\bs u)$ solves 
\begin{equation}
    \max_{\bs p \in \Delta_N} g\left(\bs p; \bs u; (G_i)_{i \in [N]}\right) = \bs p^\top \bs u + \sum\limits_{i=1}^N \int_{1-p_i}^1 G_{i}^{-1} (s)\diff s, 
    \label{eq:regularized-map}
\end{equation}
then $\bs p\opt(\hat{\bs u}_j^{(t-1)})$ constitutes a valid strategy for player~$j$ in round $t$ as Lemma~\ref{lem:distributional-regularization} ensures that $\bs p\opt(\hat{\bs u}_j^{(t-1)}) \in  T(\mu_j\opt; \hat{\bs u}_j^{(t-1)})$, where $\mu_j\opt$ satisfies~\eqref{eq:optimistic-bel-dyn}. Consequently, we assume that the player~$j$ selects their strategy $\bs p_j^{(t)}$ by solving the maximization problem in~\eqref{eq:regularized-map} for $\bs u = \hat{\bs u}_j^{(t-1)} $. Curiously, from an optimization perspective, choosing actions based on optimistic beliefs yields a smoothing effect on the strategy that will be implemented in the subsequent round.
While one might expect a greedy strategy-choosing actions via $\bs p_j^{(t)} \mapsto \argmax_{\bs p \in \Delta_N} \bs p^\top \hat{\bs u}_j^{(t-1)}$ (which corresponds to choosing actions with certain fixed beliefs by Proposition~\ref{prop:se-ob-ne}), to suffice, such methods have been criticized as being too aggressive in the presence of uncertainty~\cite[\S 3.2]{mertikopoulos2019learning}.

\begin{proposition} Suppose that for each $j\in[M]$, $\mc M_j$ is of the form~\eqref{eq:marginal-belief-set} with $(G_{j,i})_{i \in [N]}$ that are continuous for all $i \in[N]$. If Algorithm \ref{alg:rs-obl} is run with a step size sequence $(\lambda_t)_{t\in [T]}$ such that 
$\sum_{t=1}^\infty \lambda_t = \infty$, and if $\bs P^{(T)}$ converges to $\bs P^* \in (\Delta_N)^M$ as $T \to \infty$, then $\bs P^*$ is a NE of the game~$\mc G([M], \Delta_N, (\bar u_j)_{j\in[M]})$, where $\bar u_j$, $j\in [M]$, is as defined in~\eqref{eq:smooth-exp-payoffs} for some continuous cumulative distribution functions~$(F_{j,i})_{i \in [N]}$.
\label{prop:convergence-to-point-and-this-is-nash}
\end{proposition}
\begin{proof}[Proof of Proposition~\ref{prop:convergence-to-point-and-this-is-nash}]
For strategy profile $\bs P^* = (\bs p_1^*, \ldots, \bs p^*_M) \in (\Delta_N)^M$, denote by
\[
\bs v_j^* = \left(u_j(\bs e_i, \bs P^*_{-j}) + F^{-1}_{j,i}\left(1- p^*_{j,i}\right)\right)_{i \in [N]} \quad \forall j \in [M].
\]
Note that 
% $\bs v_j^* = \nabla_{\bs p} \bar u_j(u_j; (\bs p_j; \bs P_{-j}); (F_{j,i})_{i \in [N]})|_{\bs p = \bs p_j^*}$, 
$\bs v_j^* = \nabla_{\bs p} \bar u_j(\bs p_j; \bs P_{-j})|_{\bs p = \bs p_j^*}$, 
where $\bar u_j$ is as defined in \eqref{eq:smooth-exp-payoffs}. If $\bs P^*$ is a Nash equilibrium of the game $\mc G([M], \Delta_N, (\bar u_j)_{j\in [M]})$, then $\bs v^*_j$ should satisfy \eqref{eq:ne-first-order} for all $j \in [M]$.
Now, for the sake of argument, suppose that $\bs P^*$ is not a Nash equilibrium of the game $\mc G([M], \Delta_N, (\bar u_j)_{j\in [M]})$. 
Then, there exists $j \in [M]$ and a deviation $\bs q_j \in \Delta_N$ such that
\[
\langle \bs v^*_j, \bs q_j - \bs p^*_j \rangle > 0.
\]
We denote $\bs V^* = (\bs v^*_1, \ldots, \bs v^*_M)$, and neighborhoods of $\bs P^*$ and $\bs V^*$ as $\mc P$.
Hence, by continuity, there exists a positive constant $c > 0$ and neighborhoods $\mc Z$ and $\mc V$ of $ \bs P^*$ and $ \bs V^*$, respectively, such that,
\begin{equation}
\langle \bs v'_j, \bs q_j - \bs p'_j \rangle \geq c > 0,
\label{eq:v-q-p-ineq}
\end{equation}
for all $\bs P' \in \mc Z$ and $\bs V' \in \mc V$. 
Suppose that $\bs P^{(T)}$ converges to $\bs P^*$, then, for simplicity, we assume that $\bs P^{(T)} \in \mc Z$ and
\[\bs V^{(T)} = \left(\bs v_1^{(T)}, \ldots, \bs v_M^{(T)}\right) \in \mc V~\text{where}~\bs v^{(T)}_j = \left(u_j\left(\bs e_i, \bs P^{(T)}_{-j}\right) + F^{-1}_{j,i}\left(1- p^{(T)}_{j,i}\right)\right)_{i \in [N]}.\] 
Denote $ \hat{\bs U}^{(t)} = ((\hat{\bs u}_{1,i}^{(t)})_{i \in [N]}, \ldots, (\hat{\bs u}_{M,i}^{(t)})_{i \in [N]})$. Then, we have
\[
\hat{\bs U}^{(T)} = \hat{\bs U}^{(0)} + \sum_{t=1}^T \lambda_t \bs V^{(t)} = \hat{\bs U}^{(0)} + \left(\sum_{t=1}^T \lambda_t\right) \bar{\bs V}^{(T)},
\]
where $\bar{\bs V}^{(T)} = (\bar{\bs v}_1^{(T)}, \ldots, \bar{\bs v}_M^{(T)})$ and satisfies
\[
\bar{\bs V}^{(T)} = \frac{\sum_{t=1}^T \lambda_t \bs V^{(t)}}{\sum_{t=1}^T \lambda_t}.
\]
If we denote by
\[
h_j(\bs p) = \sum_{i=1}^N \int_{1-p_i}^1 G^{-1}_{j,i}(s)\diff s,
\]
then $\bs p^{(T)}_j$ solves
\[
\bs p^{(T)}_j \in \arg\max_{\bs p \in \Delta_N}
\Big\{ \langle \hat{\bs u}_j^{(T-1)}, \bs p \rangle + h_j(\bs p) \Big\}.
\]
Hence, for any $\bs q_j \in \Delta_N$, by optimality of $\bs p^{(T)}_j$ we have
\begin{align}\label{eq:conv-ne}
h_j(\bs q_j) - h_j(\bs p^{(T)}_j)
&\le \langle \hat{\bs u}_j^{(T-1)}, \bs p^{(T)}_j - \bs q_j \rangle \\
&= -\langle \hat{\bs u}_j^{(T-1)}, \bs q_j - \bs p^{(T)}_j \rangle \\
&= -\langle \hat{\bs u}_j^{(0)}, \bs q_j - \bs p^{(T)}_j \rangle
- \Big(\sum_{t=1}^{T-1} \lambda_t\Big)
\langle \bar{\bs v}_j^{(T-1)}, \bs q_j - \bs p^{(T)}_j \rangle .
\end{align}

Given that $\bs P^{(T)} \to \bs P^*$ as $T \to \infty$, we have $\bs V^{(T)} \to \bs V^*$, and thus
$\bar{\bs V}^{(T-1)} \to \bs V^*$.
By the Cauchy--Schwarz inequality,
\[
-\langle \hat{\bs u}_j^{(0)}, \bs q_j - \bs p^{(T)}_j \rangle
\le \big|\langle \hat{\bs u}_j^{(0)}, \bs q_j - \bs p^{(T)}_j \rangle\big|
\le \|\hat{\bs u}_j^{(0)}\| \cdot \|\bs q_j - \bs p^{(T)}_j\|
\le 2\|\hat{\bs u}_j^{(0)}\| = O(1),
\]
where the last inequality follows because $\max_{\bs p,\bs q \in \Delta_N}\|\bs p-\bs q\|=2$.
Additionally, note that we have
$\langle \bar{\bs v}_j^{(T-1)}, \bs q_j - \bs p^{(T)}_j \rangle \ge c>0$
by~\eqref{eq:v-q-p-ineq}.
Hence, \eqref{eq:conv-ne} implies
\[
h_j(\bs q_j) - h_j(\bs p^{(T)}_j)
\le O(1) - c \sum_{t=1}^{T-1}\lambda_t \xrightarrow[T\to\infty]{} -\infty,
\]
where the last limit follows because $(\lambda_t)_{t\in\mathbb N}$ is such that
$\sum_{t=1}^\infty \lambda_t = \infty$.
The conclusion above draws a contradiction because
$\max_{\bs p,\bs q \in \Delta_N}\big(h_j(\bs p)-h_j(\bs q)\big) < \infty$
implies
\[
h_j(\bs q_j) - h_j(\bs p^{(T)}_j)
= -\big(h_j(\bs p^{(T)}_j)-h_j(\bs q_j)\big)
\ge -\max_{\bs p,\bs q\in\Delta_N}\big(h_j(\bs p)-h_j(\bs q)\big)
> -\infty.
\]
Therefore, we conclude that $\bs P^*$ is a Nash equilibrium of the game
$\mc G([M], \Delta_N, (\bar u_j)_{j\in[M]})$.
\end{proof}
% \end{proof}
% }

% \subsection{Convergence Analysis}
Equipped with the discussions on static and dynamic principles of Algorithm~\ref{alg:rs-obl}, we are ready to focus on the convergence of Algorithm~\ref{alg:rs-obl} to the NE of the game~$\mc G([M], \Delta_N, (\bar u_j)_{j\in [M]})$. 
% We first state the following assumption regarding the marginal distributions associated with the dynamic belief set $\mc M_j$.
\begin{assumption}\label{ass:G-fenchel}
    For each $j \in [M]$, $\mc M_j$ is in the form of~\eqref{eq:marginal-belief-set} with cumulative distribution functions~$(G_{j,i})_{i\in [N]}$ that are Lipschitz continuous with Lipschitz constant $ L > 0$. {Let $(\bs u_n)_{n\ge 1}\subset [0,1]^N$ be any sequence, and, for each $n$, if $\bs p\opt(\bs u_n) \to \bs p$, then we have
    \begin{equation}
        \max\limits_{\bs p' \in \Delta_N} g(\bs p'; \bs u_n; (G_{j,i})_{i \in [N]}) - g(\bs p; \bs u_n ; (G_{j,i})_{i \in [N]}) \to 0.
    \end{equation}}
    \label{ass:dynamic-belief-sets-fenchel-coupling}
\end{assumption}

% \notebt{When the marginals of the belief sets of Lipschitz continuous, then by Lemma \ref{lem:strong-convex-reg}, $-g$ is strongly convex. Hence, by \cite[Corollary 4.4]{mertikopoulos2019learning} the Assumption \ref{ass:G-fenchel} is satisfied. Hence, every example of the marginal belief sets presented in \ref{sec:SE-OB} satisfies this assumption.}

Although Assumption \ref{ass:dynamic-belief-sets-fenchel-coupling} might initially appear restrictive, it is easy to check that it is satisfied whenever the belief set $\mathcal{M}_j$ is generated by cumulative distribution functions featured in the Examples \ref{ex:exp-marginal}, \ref{ex:uniform} and \ref{ex:pareto} (see our proofs in Appendix~\ref{sec:verify} and \cite[Page 16]{mertikopoulos2019learning}). 
Moreover, we present a sufficient condition for Assumption~\ref{ass:dynamic-belief-sets-fenchel-coupling} as follows.
{\color{black}\begin{lemma}\label{lem:suff}
Suppose that $G_{i}$, $i \in [N]$, are Lipschitz continuous with Lipschitz constant $L>0$ and either of the following conditions holds.
\begin{enumerate}
    \item[(i)] there exists $\epsilon >0$ such that $\epsilon\leq p_i^\star(\bs u) \leq 1-\epsilon$ for all $i\in[N]$; 
    \item[(ii)] there exists $\bar B< \infty$ such that $|G_i^{-1}(t)| \leq \bar B$ for all $t\in[0, 1]$.
\end{enumerate}
Let $\bs u_n\subset [0,1]^N$ be any sequence, and for each $n$, if $p\opt(\bs u_n)\to \bar{\bs p}$, then we have $\max_{\bs p' \in \Delta_N} g(\bs p';\bs u_n; G_i) - g(\bar {\bs p};\bs u_n ; G_i) \to 0.$
\end{lemma}
\begin{proof}[Proof of Lemma~\ref{lem:suff}]
Define
\[
R(\bs p) = \sum_{i=1}^N \int_{1-p_i}^1 G_i^{-1}(s)ds,
\qquad
g(\bs p;\bs u;(G_i)) = \bs p^\top \bs u + R(\bs p).
\]
Since each $G_i$ is $L$-Lipschitz, Lemma~G.5 implies that $-g(\cdot;\bs u;(G_i))$ is $1/L$-strongly convex
in $\bs p$, equivalently $g(\cdot;\bs u;(G_i))$ is $1/L$-strongly concave in $\bs p$; in particular
the maximizer $p^\star(\bs u)$ is unique.

Moreover, for each $i$,
\[
\frac{\diff R(\bs p)}{\diff p_i} = G_i^{-1}(1-p_i),
\quad\text{so}\quad
\nabla_{\bs p} g(\bs p;\bs u;(G_i)) = \bs u + \big(G_i^{-1}(1-p_i)\big)_{i\in[N]}.
\]

Let $\bs p_n=p^\star(\bs u_n)$ and define
\[
V_n = g(\bs p_n;\bs u_n;(G_i)) - g(\bar{\bs p};\bs u_n;(G_i)) \ge 0.
\]
By $1/L$-strong concavity of $g(\cdot;\bs u_n)$, for every $n$ we have
\[
V_n \le \left\langle \nabla_{\bs p} g(\bar{\bs p};\bs u_n;(G_i)), \bs p_n-\bar{\bs p}\right\rangle
-\frac{1}{2L}\|\bs p_n-\bar{\bs p}\|_2^2.
\]
Hence, using Cauchy--Schwarz,
\[
V_n \le \|\nabla_{\bs p} g(\bar{\bs p};\bs u_n;(G_i))\|_2\|\bs p_n-\bar{\bs p}\|_2.
\]

Because $\bs u_n\in[0,1]^N$, we have $\|\bs u_n\|_2\le \sqrt N$.
Under (i), $\bar{\bs p}\in[\epsilon,1-\epsilon]^N$ and continuity of $G_i^{-1}$ on $[\epsilon,1-\epsilon]$
implies there exists $B_\epsilon<\infty$ with $|G_i^{-1}(t)|\le B_\epsilon$ for all $t\in[\epsilon,1-\epsilon]$,
so $\| (G_i^{-1}(1-\bar p_i))_{i\in[N]}\|_2 \le \sqrt NB_\epsilon$.
Under (ii), $|(G_i^{-1}(1-\bar p_i))|\le \bar B$ for all $i$, so the same norm is bounded by $\sqrt N\bar B$.
Therefore, in both cases there is a finite constant $C$ such that
$\|\nabla_{\bs p} g(\bar{\bs p};\bs u_n;(G_i))\|_2 \le C$ for all $n$.
It follows that
\[
0 \le V_n \le C\|\bs p_n-\bar{\bs p}\|_2 \to 0,
\]
since $\bs p_n\to\bar{\bs p}$ by assumption.
\end{proof}}
% }

% one can show that every example of the marginal belief sets presented in Section~\ref{sec:SE-OB} satisfies this assumption. 
% \notebt{What was the reason for that? that $g$ is strongly convex in that case?}
% \noteyg{The reason is that the regularization part in $g$ is strictly concave. By \cite[(4.10)]{mertikopoulos2019learning}, the Bregman divergence is well-defined.}
If the risk-preference functions satisfy Assumption~\ref{ass:vs-hessian} and the dynamic belief sets $\mc M_j$ satisfy Assumption~\ref{ass:dynamic-belief-sets-fenchel-coupling}, then the following theorem demonstrates that Algorithm~\ref{alg:rs-obl} converges to a strategy profile that forms an SE-OB for the game $\mc G([M], \Delta_N, (u_j)_{j \in [M]})$ when the belief sets of the players are of the form \eqref{eq:marginal-belief-set} induced by their corresponding risk-preference functions.  

Recall that the strategies of each player are chosen by solving~\eqref{eq:regularized-map}. 
Then, the sequence of play created by Algorithm~\ref{alg:rs-obl} is equivalent to that of dual averaging algorithm~\cite[\S~3]{mertikopoulos2019learning} with regularization functions of the form 
\[\bs p \mapsto \sum_{i=1}^N \int_{1-p_{i}}^1 G_{j,i}^{-1}(s)\diff s \quad  \forall  j\in[M].\]
By Lemma~\ref{lem:concave-game},~$\mc G([M], \Delta_N, (\bar u_j)_{j \in [M]})$ is a concave game. Hence, our proof can be adapted from~\cite[Theorem 4.1]{mertikopoulos2019learning}, which is for general cases when the feedback mechanism is noisy. 
For the sake of completeness, we now present a slightly revised version of their proof.
\begin{theorem}
Under Assumptions~\ref{ass:vs-hessian} and~\ref{ass:dynamic-belief-sets-fenchel-coupling}, suppose that Algorithm~\ref{alg:rs-obl} is run with the step-size sequence $(\lambda_t)_{t\in [T]}$ such that %$\lambda_t > 0$ and
$\sum_{s=1}^t \lambda^2_s / \sum_{s=1}^t \lambda_s \to 0$. Then, the trajectory of play $\bs P^{(T)}$ converges to an \eqref{eq:b-done} of the game~$\mathcal G([M], \Delta_N, (u_j)_{j\in [M]})$ induced by the belief sets $\mc B_j$, ${j \in [M]}$, of the form~\eqref{eq:marginal-belief-set} with cumulative functions~$(F_{j,i})_{i \in [N]}$. 
\label{theorem:asymp-convergence}
\end{theorem}
\begin{proof}[Proof of Theorem~{\ref{theorem:asymp-convergence}}]
By Proposition~\ref{prop:vs-done}, the game $\mc G([M], \Delta_N, (\bar u_j)_{j\in[M]})$ admits a unique Nash equilibrium that is variationally globally stable. 
Recall that the strategies of each player are chosen by solving~\eqref{eq:regularized-map}. Then, the sequence of play created by Algorithm~\ref{alg:rs-obl} is equivalent to that of dual averaging algorithm~\cite[\S 3]{mertikopoulos2019learning} with regularization functions $\sum_{i=1}^N \int_{1-p_i}^1 G_i^{-1}(s) \diff s$, which are strongly concave thanks to Lemma~\ref{lem:strong-convex-reg}.
Finally, the claim follows by~\cite[Theorem~4.6]{mertikopoulos2019learning}.
% \end{proof}

% {\color{red}\subsection{Proof of Remark~\ref{rk:exp}}\label{sec:proof-rk}
% To justify the condition with the exponential marginal distributions, we note that $F^{-1}_{j,i}(s) = -\lambda(1+\log [\eta_i^{-1}(1-s)])$ and $F'_{j,i}(t) = \eta_i \lambda^{-1}\exp(-t/\lambda-1)$. Then, in this case, we have
% \[
% \sup_{p \in [0,1]} \max_{i \in [N]} F'_{j,i}(F_{j,i}^{-1}(1-p)) = \sup_{p \in [0,1]} \max_{i \in [N]} \eta_i \lambda^{-1}\exp\left(-1 + (1+\log(\eta^{-1}_i p))\right) = \sup_{p \in [0,1]} \lambda^{-1} p = \lambda^{-1}.
% \]
% Note that the operator norm is bounded above by the Frobenius norm and that $u_j(\bs P)\in[0,1]$, the right-hand side of~\eqref{eq:vs-hessian-cond-F_k} is at most $N(M-1)$. 
% Thus, a feasible marginal distribution satisfying Assumption~\ref{ass:vs-hessian} should have $\lambda > N(M-1)$.}
% \end{proof}
\end{proof}

% {\color{bblye}
\section{Additional Related Literature}
\label{app:add-lit-review}

This appendix collects two background strands that are conceptually related to SE-OB
but not central to the paper's main economic mechanism.
The first concerns the computational difficulty of Nash equilibrium in general finite games,
which motivates equilibrium notions with tractable regularized best-response structure.
The second concerns convergence of adaptive dynamics in repeated games, clarifying when
(iterative) adjustment procedures can be expected to reach equilibrium objects.

\vspace{0.5cm}

\noindent\textbf{Computational complexity of Nash equilibrium.}
Computing a Nash equilibrium in a normal-form game is inherently complex, primarily because of the combinatorial explosion of potential strategies and the intricate interdependencies among players’ choices. While Nash’s theorem assures us that an equilibrium exists, subsequent research has demonstrated that finding this equilibrium in finite normal-form games is PPAD-complete \citep{daskalakis2009complexity,chen2009settling}. 
This classification implies that, in the worst case, no polynomial-time algorithm is likely to exist for computing a Nash equilibrium.
As a result, researchers have shifted their focus toward the approximability of Nash equilibria, seeking to compute an equilibrium that meets a predefined tolerance level rather than an exact solution. However, even this relaxed goal has proven to be computationally demanding \citep{daskalakis2013complexity, rubinstein2015inapproximability}, that is, even approximating a Nash equilibrium remains PPAD-hard. This persistent intractability, along with the potential tradeoff between the relatively high payoffs associated with Nash equilibria and the substantial computational resources required to compute them, raises important questions about whether players will always seek to reach such equilibria in practice \citep{daskalakis2013complexity, hart2010long}.

The independent nature of players' action choices introduces additional degrees of freedom, significantly contributing to the computational complexity of strategic dynamics. Relaxing the independence assumption by permitting correlated actions leads to equilibrium concepts that are computationally tractable. Notably, coarse correlated equilibrium (CCE) \citep{moulin1978strategically} and correlated equilibrium (CE) \citep{aumann1987correlated} exemplify such relaxations. However, because these concepts depend on the existence of a public signal and the coordination of players' actions, their applicability is limited in decentralized settings or in scenarios where privacy concerns restrict such coordination.

Drawing inspiration from smoothed analysis \citep{spielman2009smoothed}, \citet{daskalakis2023smooth} introduced the concept of a smooth Nash equilibrium. In this framework, players' strategies are confined to smooth distributions, meaning that the probability assigned to any pure strategy is bounded away from one. This constraint effectively limits the use of pure strategies. Moreover, under a given smoothness parameter and assuming independent players, \citep{daskalakis2023smooth} demonstrate the existence of a polynomial-time algorithm\footnote{An algorithm with computational complexity polynomial in the number of players and the number of actions of each player.} for approximating the smooth Nash equilibrium in normal-form games, thereby substantially alleviating the computational challenges associated with computing a Nash equilibrium.
Another strategy to address these computational challenges is to incorporate risk aversion into players' decision-making. Rather than merely maximizing expected payoffs, risk-averse players optimize a specific risk measure, resulting in risk-adjusted equilibria. When players are constrained to bounded strategies, these equilibria are computationally tractable in all $n$-player matrix games and finite-horizon Markov games \citep{mazumdar2024tractable}.\\

\noindent \textbf{Learning Nash equilibrium in repeated games.} In its most natural form, a game is not a one-shot, isolated event but a dynamic process of continuous interaction, where players iteratively adapt their strategies based on the observed behavior of their opponents. This adaptive behavior is captured by the concept of fictitious play \citep{brown1951iterative,robinson1951iterative,hofbauer2002global}, in which each player chooses an action that is a \textit{best response} to the empirical distribution of the opponents' play until that time. First, \citet{robinson1951iterative} demonstrated that, in two-player zero-sum games, the joint distribution of action frequencies converges to a Nash equilibrium. However, \citet{ref:shapley1963some} later showed that this convergence result does not extend beyond zero-sum games. In repeated games, players experience regret, defined as the cumulative difference between the payoff of the strategy they followed and that of the best fixed action in hindsight. Notably, fictitious play is not Hannan consistent, that is, it does not guarantee no-regret performance. In contrast, when players adopt no-regret strategies in repeated games, \citet{hart2000simple} proved that the empirical frequency of their strategies converges to a CCE. Although this result is promising, the resulting CCE may fail to satisfy even the most basic rationalizability axioms \citep{viossat2013no}. Moreover, typical no-regret convergence results address the empirical frequency over the entire horizon rather than the strategies played in the final iteration. The study of last-iterate convergence in repeated games seeks to address this shortcoming. Specifically, in smooth concave games that admit a stable Nash equilibrium, \citet{mertikopoulos2019learning} and \citet{ref:hsieh2021adaptive} show that strategies generated by a variant of the dual averaging algorithm \citep{nesterov2009primal} converge to the Nash equilibrium. Similarly, in monotone \citep{ref:rosen1965existence} and smooth games, \citet{golowich2020tight} demonstrate that an optimistic gradient-based algorithm, augmented with an additional term, converges to a Nash equilibrium. Unfortunately, normal-form games do not invariably exhibit the desirable properties of smoothness or monotonicity, so these convergence guarantees do not directly extend to mixed extensions of general-form games.

\section{Additional details for examples of marginal belief sets}
\label{sec:examples-proofs}
In this section, we provide detailed derivations for the results presented in Examples~\ref{ex:exp-marginal}, \ref{ex:uniform}, and \ref{ex:pareto}, after which, we will verify that all three examples satisfy Assumptions~\ref{ass:vs-hessian} and \ref{ass:G-fenchel}.

\subsection{Derivation of examples}
\noindent\textbf{Detailed derivations of Example~\ref{ex:exp-marginal}.}
Fix $j\in[M]$ and suppose the marginal distributions in \eqref{eq:marginal-belief-set} are (shifted)
exponential with cumulative distribution functions
\[
F_{j,i}(s)=\max\Bigl\{0,1-\eta_{j,i}\exp\bigl(-s/\gamma_j-1\bigr)\Bigr\},
\qquad s\in\mathbb R,
\]
where $\gamma_j>0$ and $\eta_{j,i}>0$ for all $i\in[N]$.
For $t\in(0,1)$, the left-quantile function is
\[
F_{j,i}^{-1}(t)
=-\gamma_j\left(\log\left(\frac{1-t}{\eta_{j,i}}\right)+1\right).
\]
Hence, for every $p_i\in[0,1]$,
\[
\int_{1-p_i}^1 F_{j,i}^{-1}(t)\diff t
= -\gamma_j p_i \log\left(\frac{p_i}{\eta_{j,i}}\right),
\]
where we interpret $0\log 0=0$.

Recall the definition of $T(\mu;\bs u)$ in \eqref{eq:Tau-def}.
By Lemma~\ref{lem:distributional-regularization}, if $\mu_j^{\opt}$ solves \eqref{eq:opt-exp-payoff},
then any maximizer $\bs p^\star$ of the concave program
\[
\max_{\bs p\in\Delta_N}
\left\{
\bs p^\top \bs u
-\gamma_j\sum_{i=1}^N p_i \log\left(\frac{p_i}{\eta_{j,i}}\right)
\right\}
\]
satisfies $\bs p^\star\in T(\mu_j^{\opt};\bs u)$; equivalently,
\[
\arg\max_{\bs p\in\Delta_N}
\left\{
\bs p^\top \bs u
-\gamma_j\sum_{i=1}^N p_i \log\left(\frac{p_i}{\eta_{j,i}}\right)
\right\}
\subseteq T(\mu_j^{\opt};\bs u).
\]
In particular, the maximizer is unique and is given by
\[
p_i^\star(\bs u)
=\frac{\eta_{j,i}\exp(u_i/\gamma_j)}{\sum_{k=1}^N \eta_{j,k}\exp(u_k/\gamma_j)},
\qquad i\in[N],
\]
which matches Example~\ref{ex:exp-marginal}.

% \textbf{Detailed derivations of Example \ref{ex:exp-marginal}}.
% If the marginal distributions of \eqref{eq:marginal-belief-set} are set to exponential distributions, then 
% \[\int_{1-p_{i}}^1 F_{j,i}^{-1}(t) \diff t =- \gamma_j p_{i} \log\left(\frac{p_{i}}{\eta_{j,i}}\right). \]
% Recall the definition of $T(\mu; \bs u)$ in \eqref{eq:Tau-def}.
% Then, by Lemma \ref{lem:distributional-regularization}, we have
% \[ \argmax\limits_{\bs p \in \Delta_N} \bs p^\top \bs u - \gamma_j \sum\limits_{i=1}^N p_i \log\left( \frac{p_i}{ \eta_{j,i}}\right) \in T(\mu_j\opt; \bs u),\]
% where $\mu_j\opt$ solves \eqref{eq:opt-exp-payoff}.
% An application of the KKT conditions to the above optimization problem yields a closed-form expression for its maximizer, as stated in Example~\ref{ex:exp-marginal}.
% We note that $F_{j,i}^{-1}(t) = -\gamma_j (1-\log \eta_{j,i} + \log()1-t)$, then 
% \begin{align}
%     \int_{1-p_i}^1 F_{j,i}^{-1}(t) dt = -\gamma_j p_i \log\frac{p_i}{\eta_{j,i}}.
% \end{align}
% By maximizing the objective function
% \[
% \bs p^\top \bs u-\sum_{i=1}^N \gamma_j p_i \log\frac{p_i}{\eta_{j,i}},
% \]
% it yields 
% \[
% p_i \propto \eta_{j,i} \exp\left(\frac{u_i}{\gamma_j} - 1\right),
% \]
% thus, via normalization, we obtain
% \[\tau_i (\mu_j\opt; \bs u ) = {\eta_{j,i} \exp(u_i/ \gamma_j)}/{ \sum_{k=1}^N \eta_{j,k} \exp(u_k/\gamma_j)}.\]

\noindent\textbf{Detailed derivations of Example~\ref{ex:uniform}.}
Fix a player $j\in[M]$ and abbreviate $\boldsymbol\theta_j=(\theta_{j,1},\dots,\theta_{j,N})$.
Under the uniform-marginal specification in Example~\ref{ex:uniform},
\[
\xi_{j,i}\sim \mathrm{Unif}\left([\theta_{j,i}-\gamma_j/2,\ \theta_{j,i}+\gamma_j/2]\right)
\qquad (i\in[N]),
\]
so the (left) quantile function of $F_{j,i}$ is
\begin{equation}\label{eq:uniform-quantile}
F_{j,i}^{-1}(t)=\theta_{j,i}-\frac{\gamma_j}{2}+\gamma_j t,\qquad t\in[0,1].
\end{equation}
Therefore, for any $p_i\in[0,1]$,
\begin{align}
\int_{1-p_i}^1 F_{j,i}^{-1}(t)dt
&=\int_{1-p_i}^1\Big(\theta_{j,i}-\frac{\gamma_j}{2}+\gamma_j t\Big)dt \label{eq:int-start}\\
&=\Big(\theta_{j,i}-\frac{\gamma_j}{2}\Big)\underbrace{\int_{1-p_i}^1 dt}_{=p_i}
+\gamma_j\underbrace{\int_{1-p_i}^1 tdt}_{=(1-(1-p_i)^2)/2} \label{eq:int-split}\\
&=\Big(\theta_{j,i}-\frac{\gamma_j}{2}\Big)p_i+\frac{\gamma_j}{2}\Big(1-(1-2p_i+p_i^2)\Big) \label{eq:int-eval}\\
&=\theta_{j,i}p_i+\frac{\gamma_j}{2}\big(p_i-p_i^2\big). \label{eq:int-final}
\end{align}

Now apply Lemma~\ref{lem:distributional-regularization}: if $\mu_j\opt$
solves~\eqref{eq:opt-exp-payoff}, then any maximizer of
\[
\max_{\bs p\in\Delta_N}\Big\{\bs u^\top \bs p+\sum_{i=1}^N\int_{1-p_i}^1 F_{j,i}^{-1}(t)dt\Big\}
\]
belongs to $T(\mu_j\opt;\bs u)$. Substituting~\eqref{eq:int-final} yields the equivalent program
\begin{align}
\arg\max_{\bs p\in\Delta_N}
\Big\{\bs u^\top \bs p+\boldsymbol\theta_j^\top \bs p+\frac{\gamma_j}{2}\sum_{i=1}^N(p_i-p_i^2)\Big\}
&=\arg\max_{\bs p\in\Delta_N}
\Big\{(\bs u+\boldsymbol\theta_j)^\top \bs p-\frac{\gamma_j}{2}\|\bs p\|_2^2\Big\},
\label{eq:uniform-quad-prog}
\end{align}
because $\sum_{i=1}^N p_i=1$ on $\Delta_N$ makes $\frac{\gamma_j}{2}\sum_i p_i=\gamma_j/2$ a constant.

Next, complete the square in~\eqref{eq:uniform-quad-prog}. For any $\bs p\in\Delta_N$,
\begin{align}
(\bs u+\boldsymbol\theta_j)^\top \bs p-\frac{\gamma_j}{2}\|\bs p\|_2^2
&=-\frac{\gamma_j}{2}\Big\|\bs p-\frac{\bs u+\boldsymbol\theta_j}{\gamma_j}\Big\|_2^2
+\frac{1}{2\gamma_j}\|\bs u+\boldsymbol\theta_j\|_2^2. \label{eq:complete-square}
\end{align}
The final term in~\eqref{eq:complete-square} does not depend on $\bs p$, so the maximizers of
\eqref{eq:uniform-quad-prog} coincide with the Euclidean projection of
$(\bs u+\boldsymbol\theta_j)/\gamma_j$ onto the simplex:
\begin{equation}\label{eq:projection-form}
\arg\max_{\bs p\in\Delta_N}\Big\{(\bs u+\boldsymbol\theta_j)^\top \bs p-\frac{\gamma_j}{2}\|\bs p\|_2^2\Big\}
=
\arg\min_{\bs p\in\Delta_N}\Big\|\bs p-\frac{\bs u+\boldsymbol\theta_j}{\gamma_j}\Big\|_2^2 =\operatorname{sparsemax}\Big(\frac{\bs u+\boldsymbol\theta_j}{\gamma_j}\Big),
\end{equation}
where the equality follows by definition of the sparsemax operator
\citep[Proposition~1]{ref:martins2016softmax}. 
\medskip

% \noindent\textbf{Detailed derivations of Example \ref{ex:pareto}.} If the marginal distributions of \eqref{eq:marginal-belief-set} are set to Pareto distributions, then
% \[\int_{1-p_i}^1 F_{j,i}^{-1}(t) \diff t = -{\gamma_j}\eta_{j,i} \cdot \frac{(p_i/\eta_{j,i})^q - p_i/(\eta_{j,i}) }{q-1}.\]
% Then, by Lemma \ref{lem:distributional-regularization},
% \begin{equation*} \textrm{\rm argmax}_{\bs p\in \Delta_N} \bs p^\top \bs u-{\gamma_j}{(q-1)^{-1}} \sum_{i=1}^N\eta_{j,i}\left(\left({p_i}/{\eta_{j,i}}\right)^{q} - {p_i}/{\eta_{j,i}}\right) \subseteq T(\mu_j\opt; \bs u),\end{equation*}
% where $\mu_j\opt$ solves \eqref{eq:opt-exp-payoff}, is given by the expression stated in Example \ref{ex:pareto}.

\noindent\textbf{Detailed derivations of Example~\ref{ex:pareto}.}
If the marginal distributions in \eqref{eq:marginal-belief-set} are Pareto as in Example~\ref{ex:pareto},
then for each $i\in[N]$ and $p_i\in[0,1]$,
\[
\int_{1-p_i}^{1} F_{j,i}^{-1}(t) \diff t
= \frac{\gamma_j}{q-1}\left(p_i-\frac{p_i^q}{\eta_{j,i}^{q-1}}\right)
= -\frac{\gamma_j\eta_{j,i}}{q-1}\left(\left(\frac{p_i}{\eta_{j,i}}\right)^q-\frac{p_i}{\eta_{j,i}}\right).
\]
Then, by Lemma~\ref{lem:distributional-regularization},
\[
\arg\max_{\bs p\in\Delta_N}
\left\{
\bs p^\top \bs u
-\frac{\gamma_j}{q-1}\sum_{i=1}^N \eta_{j,i}\left(\left(\frac{p_i}{\eta_{j,i}}\right)^q-\frac{p_i}{\eta_{j,i}}\right)
\right\}
\subseteq T(\mu_j\opt;\bs u),
\]
where $\mu_j\opt$ solves \eqref{eq:opt-exp-payoff}. This is the expression stated in
Example~\ref{ex:pareto}.

{

% \subsection{Verification of Assumptions~\ref{ass:vs-hessian} and \ref{ass:G-fenchel} in examples}\label{sec:verify}
% In this section, we verify that Examples~\ref{ex:exp-marginal}, \ref{ex:pareto}, and \ref{ex:uniform} satisfy Assumptions~\ref{ass:vs-hessian} and \ref{ass:G-fenchel}.

\subsection{Verification of Assumptions~\ref{ass:vs-hessian} and \ref{ass:G-fenchel} in examples}\label{sec:verify}
In this section, we verify that Examples~\ref{ex:exp-marginal}, \ref{ex:pareto}, and \ref{ex:uniform}
satisfy Assumptions~\ref{ass:vs-hessian} and \ref{ass:G-fenchel}.

\noindent\textbf{Verification of Example~\ref{ex:exp-marginal}.}
Fix $j\in[M]$ and $i\in[N]$, and write $\gamma=\gamma_j$, $\eta=\eta_{j,i}$ and $F=F_{j,i}$.
Recall that
\[
F(s)=\max\Bigl\{0,1-\eta\exp\bigl(-s/\gamma-1\bigr)\Bigr\},\qquad s\in\mathbb R,
\]
with $\gamma>0$ and $\eta>0$.
Define
\[
s_0=\gamma(\log\eta-1),\qquad S=(s_0,\infty).
\]
Then $F(s)=0$ for $s\le s_0$, while for $s\in S$,
\[
F(s)=1-\eta\exp(-s/\gamma-1)\in(0,1).
\]
Since $s\mapsto 1-\eta\exp(-s/\gamma-1)$ is $C^\infty$ on $S$, $F$ is continuous on $S$ and, for
$s\in S$,
\[
F'(s)=\frac{\eta}{\gamma}\exp(-s/\gamma-1)>0,
\]
so $F$ is strictly increasing on $S=\{s\in\mathbb R: F(s)\in(0,1)\}$.
Moreover,
\[
\lim_{s\downarrow s_0}F(s)=1-\eta\exp(-s_0/\gamma-1)
=1-\eta\exp(-\log\eta)=0=F(s_0),
\]
and since $F(s)=0$ for $s\le s_0$, we also have $\lim_{s\uparrow s_0}F(s)=0$.
Therefore, $F$ is continuous on $\mathbb R$, verifying Assumption~\ref{ass:marginal-belief-sets-f-strict}.

Next, for $p\in(0,1)$, the left-quantile satisfies
\[
F^{-1}(1-p)=-\gamma\left(\log\left(\frac{p}{\eta}\right)+1\right),
\]
and hence
\[
F'\bigl(F^{-1}(1-p)\bigr)=\frac{p}{\gamma},
\qquad\text{so}\qquad
\Bigl(F'\bigl(F^{-1}(1-p)\bigr)\Bigr)^{-1}=\frac{\gamma}{p}.
\]
Consequently,
\[
\inf_{p\in(0,1)} \Bigl(F'\bigl(F^{-1}(1-p)\bigr)\Bigr)^{-1}=\gamma,
\]
and thus condition \eqref{eq:vs-hessian-cond-F_k} in Assumption~\ref{ass:vs-hessian}
reduces to $\gamma_j>C_j$ (equivalently, $\gamma>C_j$).

To verify the Lipschitz requirement in Assumption~\ref{ass:G-fenchel}, note that for $s\in S$,
\[
\exp(-s/\gamma-1)\le \exp(-s_0/\gamma-1)=\frac{1}{\eta},
\]
so $|F'(s)|\le 1/\gamma$ on $S$. Since $F'(s)=0$ for $s<s_0$, it follows that $F$ is globally
$(1/\gamma)$-Lipschitz on $\mathbb R$, proving the first part of Assumption~\ref{ass:G-fenchel}.

Finally, the corresponding (regularized) best response is the weighted softmax map
\[
p_k^\star(\bs u)=
\frac{\eta_{j,k}\exp(u_k/\gamma_j)}{\sum_{\ell=1}^N \eta_{j,\ell}\exp(u_\ell/\gamma_j)},
\qquad k\in[N],
\]
as stated in Example~\ref{ex:exp-marginal}. In particular, $p_k^\star(\bs u)\in(0,1)$ for all $\bs u$.
As $\bs u$ ranges over a compact set (e.g.\ $\bs u\in[0,1]^N$),
then continuity of $\bs u\mapsto \bs p^\star(\bs u)$ implies that there exists $\varepsilon>0$ such that
$p_k^\star(\bs u)\in[\varepsilon,1-\varepsilon]$ for all $k$ and all such $\bs u$.
Therefore, condition (i) of Lemma~\ref{lem:suff} holds, and the second part of
Assumption~\ref{ass:G-fenchel} follows from Lemma~\ref{lem:suff}.

% \noindent\textbf{Verification of Example~\ref{ex:exp-marginal}.}
% Define $s_0 = \gamma(\log(\eta) - 1)$ and $ S = (s_0, +\infty)$. When $s \in S$, $F(s)=1-\eta\exp(-s/\gamma-1)$ and satisfies $F \in (0,1)$. We note that the map $s\mapsto 1-\eta\exp({-s/\gamma-1})$ is continuous on $S$ as exponential function is $C^\infty$. Next, we compute \[F'(s)=({\eta}/{\gamma}) \exp({-s/\gamma-1}),\] 
% which is strictly positive and thus $F$ is strictly increasing on $S$. At the junction point $s=s_0$ we have
% \[\lim_{s\downarrow s_{0}}F(s)=0=F(s_{0}),\] so
% $F$ is continuous on $\mathbb R$, which verifies Assumption~\ref{ass:marginal-belief-sets-f-strict}.  

% Note that $F^{-1}(1-p)=-\gamma(\log(p/\eta)+1)$, hence
% $F'(F^{-1}(1-p))=p/\gamma$, and its infimum over $p\in (0,1)$ is
% $\gamma$. Hence, \eqref{eq:vs-hessian-cond-F_k} of Assumption~\ref{ass:vs-hessian} is equivalent to the inequality $\gamma>C_j$. As $\gamma > 0$, we have $\exp(-s/\gamma -1) \leq 1/\eta$, implying that $|F'(s)| \leq 1/\gamma$. Thus, $F$'s are Lipschitz and satisfy the first part of Assumption~\ref{ass:G-fenchel}. 

% Finally, note that $p\opt(u)=   \exp (u_i / \gamma) / \sum_{k=1}^N  \exp (u_k / \gamma)$ as presented in Example~\ref{ex:exp-marginal}, and thus the condition $(i)$ of Lemma~\ref{lem:suff} is also satisfied. Consequently, the second requirement of Assumption~\ref{ass:G-fenchel} is also met via Lemma~\ref{lem:suff}. 

\noindent\textbf{Verification of Example~\ref{ex:uniform}.}
Fix $j\in[M]$ and $i\in[N]$. Define the boundary points
\[
s_{\min}=\theta_{j,i}-\gamma_j/2,
\qquad
s_{\max}=\theta_{j,i}+\gamma_j/2,
\]
and the open interval $S=(s_{\min},s_{\max})$.
Define the affine map
\[
\varphi_{j,i}(s)=\frac{s-s_{\min}}{\gamma_j}.
\]
Because $\varphi_{j,i}$ is affine, $F_{j,i}$ has three regimes:
\begin{itemize}
    \item $F_{j,i}(s)=0$ when $s\le s_{\min}$;
    \item $F_{j,i}(s)=\varphi_{j,i}(s)$ when $s_{\min}<s<s_{\max}$;
    \item $F_{j,i}(s)=1$ when $s\ge s_{\max}$.
\end{itemize}
On $S$ we have $0<F_{j,i}(s)<1$ and
\[
F_{j,i}(s)=\varphi_{j,i}(s)
\quad\text{with}\quad
\varphi_{j,i}'(s)=\frac{1}{\gamma_j}>0.
\]
Hence $F_{j,i}$ is strictly increasing on $S$.
Moreover, $F_{j,i}$ is continuous on $\mathbb{R}$ because it is affine on $S$
and constant outside, matching continuously at the two boundary points
$s_{\min}$ and $s_{\max}$. Therefore, $F_{j,i}$ satisfies the monotonicity/continuity
requirement (i.e., it is strictly increasing whenever $F_{j,i}(s)\in(0,1)$).

Next, for $t\in[0,1]$ the (left) quantile function is
\[
F_{j,i}^{-1}(t)=s_{\min}+\gamma_j t
=\theta_{j,i}-\gamma_j/2+\gamma_j t.
\]
In particular, for any $p\in(0,1)$ we have
\[
F_{j,i}^{-1}(1-p)=\theta_{j,i}+\gamma_j/2-\gamma_j p \in S,
\]
so the derivative at the relevant point is well-defined and equals
\[
F'_{j,i}\left(F_{j,i}^{-1}(1-p)\right)=\frac{1}{\gamma_j},
\qquad \forall p\in(0,1).
\]
Since the width $\gamma_j$ is common across $i\in[N]$, this implies
\[
\min_{i\in[N]}F'_{j,i}\big(F_{j,i}^{-1}(1-p)\big)=\frac{1}{\gamma_j}
\quad\Longrightarrow\quad
\inf_{p\in(0,1)}\left(\min_{i\in[N]}F'_{j,i}\left(F_{j,i}^{-1}(1-p)\right)\right)^{-1}
=\gamma_j.
\]
Consequently, condition \eqref{eq:vs-hessian-cond-F_k} in Assumption~\ref{ass:vs-hessian} holds provided
\[
\gamma_j > C_j,
\]
where $C_j$ denotes the right-hand side constant appearing in \eqref{eq:vs-hessian-cond-F_k}.

Finally, since $|F'_{j,i}(s)|\le 1/\gamma_j$ for all $s\in\mathbb{R}$, each $F_{j,i}$ is
Lipschitz with Lipschitz constant $L=1/\gamma_j$, establishing the first requirement of
Assumption~\ref{ass:G-fenchel}.
Moreover, for all $t\in[0,1]$,
\[
|F_{j,i}^{-1}(t)|
\le \max\{|\theta_{j,i}-\gamma_j/2|,\ |\theta_{j,i}+\gamma_j/2|\}
\le |\theta_{j,i}|+\gamma_j/2,
\]
so condition (ii) of Lemma~\ref{lem:suff} is satisfied (take
$B=\max_{i\in[N]}(|\theta_{j,i}|+\gamma_j/2)<\infty$).
Therefore, the second requirement of Assumption~\ref{ass:G-fenchel} follows via
Lemma~\ref{lem:suff}.

\noindent\textbf{Verification of Example~\ref{ex:pareto}.}
Suppose that $1<q\le2$ and fix parameters $\gamma>0$ and $\eta>0$. Define
\[
h(s) \coloneqq -\frac{s(q-1)}{\gamma q}+\frac{1}{q}.
\]
Let
\[
s_{\min}\coloneqq \frac{\gamma q}{q-1}\Big(\frac{1}{q}-\eta^{1-q}\Big),
\qquad
s_{\max}\coloneqq \frac{\gamma}{q-1},
\qquad
S\coloneqq (s_{\min},s_{\max}).
\]
On $S$ we have $h(s)\in(0,\eta^{1-q})$, and the CDF in Example~\ref{ex:pareto} satisfies
\[
F(s)=1-\eta h(s)^{\frac{1}{q-1}}\in(0,1),
\qquad s\in S,
\]
while $F(s)=0$ for $s\le s_{\min}$ and $F(s)=1$ for $s\ge s_{\max}$.
Moreover, for $s\in S$,
\[
F'(s)=\frac{\eta}{\gamma q}h(s)^{-\frac{q-2}{q-1}}>0,
\]
so $F$ is strictly increasing whenever $F(s)\in(0,1)$ and is continuous at $s_{\min}$ and $s_{\max}$.
Hence, $F$ satisfies Assumption~\ref{ass:marginal-belief-sets-f-strict}.

The left-quantile function is
\[
F^{-1}(t)=\frac{\gamma}{q-1}\left(1-q\left(\frac{1-t}{\eta}\right)^{q-1}\right),
\qquad t\in[0,1].
\]
Therefore, for any $p\in(0,1)$,
\[
F'\left(F^{-1}(1-p)\right)
=\frac{\eta}{\gamma q}\left(\frac{p}{\eta}\right)^{2-q}
=\frac{\eta^{q-1}}{\gamma q}p^{2-q}.
\]
Assumption~\ref{ass:vs-hessian} involves the reciprocal of this quantity along the upper tail:
\[
\left(F'\left(F^{-1}(1-p)\right)\right)^{-1}
=\frac{\gamma q}{\eta^{q-1}}p^{q-2}.
\]
Since $1<q\le2$ implies $q-2\in(-1,0]$, the map $p\mapsto p^{q-2}$ is decreasing on $(0,1)$, and thus
\[
\inf_{p\in(0,1)}\left(F'\left(F^{-1}(1-p)\right)\right)^{-1}
=\frac{\gamma q}{\eta^{q-1}}.
\]
Note that with indices, this yields
$\inf_{p\in(0,1)}\big(F'_{j,i}(F^{-1}_{j,i}(1-p))\big)^{-1}
= \gamma_j q/\eta_{j,i}^{q-1}$, so
$\min_{i \in [N]} \inf_{p\in(0,1)}(\cdot)=\gamma_j q/\max_{i \in [N
]} \eta_{j,i}^{q-1}$.
Consequently, \eqref{eq:vs-hessian-cond-F_k} is satisfied whenever
\[
\frac{\gamma q}{\eta^{q-1}}>C_j,
\]
equivalently $\gamma>(C_j/q)\eta^{q-1}$ (and in particular, under the normalization $\eta=1$, this reduces to
$\gamma>C_j/q$).

Finally, note that on $S$ we have $h(s)\le \eta^{1-q}$, so
\[
|F'(s)|\le \frac{\eta}{\gamma q}\big(\eta^{1-q}\big)^{-\frac{q-2}{q-1}}
=\frac{\eta^{q-1}}{\gamma q},
\]
and $F'(s)=0$ outside $S$. Hence, $F$ is Lipschitz and satisfies the first part of
Assumption~\ref{ass:G-fenchel}. Moreover, for all $t\in[0,1]$,
\[
|F^{-1}(t)|\le \frac{\gamma}{q-1}\big(1+q\eta^{1-q}\big),
\]
so condition (ii) of Lemma~\ref{lem:suff} holds. Consequently, the second requirement of
Assumption~\ref{ass:G-fenchel} follows from Lemma~\ref{lem:suff}.

}

\section{Additional Technical Results}
\label{app:additional-tech-results}
\begin{lemma}
If the cumulative distribution functions $(F_{j,i})_{i \in [N]}$ are continuous, then the game $\mc G([M]$, $\Delta_N,$ $(\bar u_j)_{j \in [M]})$ is a concave game, that is, 
% $\bar u_j(u_j; (\bs p_j, \bs P_{-j}); (F_{j,i})_{i \in [N]})$ 
$\bar u_j(\bs p_j, \bs P_{-j})$
is concave in $\bs p_j$ for all $\bs P_{-j} \in (\Delta_N)^{M-1}$ and $j\in[M]$.
\label{lem:concave-game}
\end{lemma}
\begin{proof}[{Proof of Lemma~\ref{lem:concave-game}}]
   The derivatives of $\bar u_j$ with respect to $ p_{j,i}$ are in the form of $u_j(\bs e_i; \bs P_{-j}) + F_{j,i}^{-1}(1- p_{j,i})$, which is non-increasing as left-quantile functions are non-decreasing, and $1-p_{j,i}$ is strictly decreasing in~$p_{j,i}$. This observation completes the proof. 
\end{proof}
\begin{lemma}\label{lem:uniform-dists-smooth-ne}
Fix $\sigma\in(0,1)$ and let $\bs u\in[0,1]^N$.
Let $\bs p=\operatorname{sparsemax}(\bs u/\lambda)$. If $
\lambda \ge \frac{\sigma N}{1-\sigma}$,
then $p_i\in [0,(N\sigma)^{-1}]$ for all $i\in[N]$.
\end{lemma}
\begin{proof}
Set $\bs v=\bs u/\lambda$, so $v_i\in[0,1/\lambda]$ for all $i$.
By \cite[Proposition 1]{ref:martins2016softmax} $\tau\in\R$ such that
\[
p_i=\max\{v_i-\tau,0\}\quad\text{for all }i,\qquad\text{and}\qquad \sum_{i=1}^N p_i=1.
\]
We bound $\max_i p_i$ by considering two cases.

\emph{Case 1: $\tau\ge 0$.}
Then $p_i\le v_i\le 1/\lambda$ for all $i$, hence $\max_i p_i\le 1/\lambda$.

\emph{Case 2: $\tau<0$.}
Since $v_i\ge 0$ for all $i$, we have $v_i-\tau>0$ for all $i$, so all coordinates are active and
$
\bs p=\bs v-\tau\mathbf 1$. Summing gives
\[
1=\sum_{i=1}^N p_i=\sum_{i=1}^N v_i - N\tau,
\qquad\text{so}\qquad
-\tau=\frac{1-\sum_{i=1}^N v_i}{N}\le \frac{1}{N}.
\]
Therefore, for every $i$,
\[
p_i=v_i+(-\tau)\le \frac{1}{\lambda}+\frac{1}{N},
\]
and hence $\max_i p_i\le \frac{1}{\lambda}+\frac{1}{N}$.

Combining both cases, we obtain
\[
\max_{i\in[N]} p_i \le \frac{1}{\lambda}+\frac{1}{N}.
\]
If $\lambda \ge \frac{\sigma N}{1-\sigma}$, then
\[
\frac{1}{\lambda}+\frac{1}{N}
\le
\frac{1-\sigma}{\sigma N}+\frac{1}{N}
=
\frac{1}{\sigma N},
\]
so $p_i\le (N\sigma)^{-1}$ for all $i\in[N]$.
\end{proof}

The following lemma below adapts Theorem 1 of \citep{natarajan2009persistency}, originally stated in the context of discrete choice theory, to mixed extension of finite games.
\begin{lemma}
\label{lem:distributional-regularization}
If $\mathcal{B} \subseteq \mc P(\R^N)$ is a marginal belief set of the form~\eqref{eq:marginal-belief-set} and if the underlying cumulative distribution functions $F_i,i \in[N]$, are continuous, then we have

\begin{equation}
\label{eq:frechet-reg-utility}
   \tilde u(\bs u; \mc B)= \max _{\bs p \in \Delta_N} \bs p^\top \bs u + \sum_{i=1}^N \int_{1-p_i}^1 F_i^{-1}(t) \diff t.
\end{equation}

Additionally, if $F_i(s)$ is strictly increasing in $s$ whenever $F_i(s) \in (0, 1)$ for all $i \in [N]$, then the unique maximizer of the convex program on the right-hand-side of~\eqref{eq:frechet-reg-utility}, $\bs p\opt \in\Delta_N$, satisfy $\bs p\opt \in T(\mu\opt;\bs u)$, 
where $\mu\opt $ represents an optimizer of the problem~\eqref{eq:opt-exp-payoff}.
\end{lemma}
Given a finite-dimensional vector space~$\mc V$ with norm $\|\cdot \|$, and a given $x \in \mc V$, the tangent cone $\textrm{TC}(x)$ is defined as the closure of the set of all rays emanating from $x$ and intersecting $\mc V$ on at least one other point.

\begin{lemma}{\cite[Proposition 2.8]{mertikopoulos2019learning}}
\label{lem:vs-hessian}
Assume $\bs P^\star$ is a Nash equilibrium of the game~$\mc G([M], \Delta_N, (u_j)_{j\in [M]})$. If $\bs H(\bs P,(u_j)_{j \in [M]})\prec 0$ on $ {\rm TC}(\bs P)$, for all $j \in [M]$ and $\bs P \in \mathcal (\Delta_N)^M$, then $\bs P^\star$ is globally variationally stable, and is the unique equilibrium of the game.
\end{lemma}

\begin{lemma}{\cite[Proposition 4.8 (ii)]{tacskesen2023semi}}
If $G_i$, $i \in [N]$, is Lipschitz continuous with Lipschitz constant $L >0$, then $- g(\bs p; \bs u; (G_i)_{i \in [N]})$ is $1/L$-strongly convex in $\bs p$. 
\label{lem:strong-convex-reg}
\end{lemma}

\begin{lemma}
    If $\bs P^* \in (\Delta_N)^M$ is Nash equilibrium of the game $\mc G([M], \Delta_N, (u_j)_{j\in [M]})$, then 
    \begin{equation}\langle \nabla_{\bs p} u_j(\bs p; \bs P^*_{-j})|_{\bs p = \bs p_j^*},  \bs z_j\rangle  \leq 0 \quad \forall \bs z_j \in \textrm{TC}_j(\bs p_j^*),~j\in [M],
    \label{eq:ne-first-order}
    \end{equation}
    where $\textrm{TC}_j(\bs p)$ denotes the tangent cone to $\Delta_N$ at $\bs p \in \Delta_N$. 
\end{lemma}

\end{document}